\documentclass[%
 reprint,
 amsmath,amssymb,
 aps,
]{revtex4-1}

\usepackage{graphicx}
\usepackage{dcolumn}
\usepackage{bm}
\usepackage{hyperref}
\usepackage[utf8]{inputenc}
\usepackage[english]{babel}
\newtheorem{theorem}{Theorem}
\newtheorem{definition}{Definition}
\newtheorem{proof}{Proof}
\newtheorem{lemma}{Lemma}
\usepackage{color}

\begin{document}

\preprint{}

\title{Normal modes. The true story.}

\author{Emil Zak}
\email{emil.zak.14@ucl.ac.uk}
 \homepage{http://www.homepages.ucl.ac.uk/~ucapejz/}
\affiliation{%
 Department of Physics and Astronomy, University College London,\\
London, WC1E 6BT, UK
}%


\date{\today}

\maketitle


\section{\label{sec:level1} Introduction}
The aim of this article is a comprehensive description of normal modes of molecular vibrations. The starting point is chosen to be a general molecular system with separated center of mass and an arbitrary embedding of body-fixed axes. This allows to focus on internal degrees of freedom only, leaving the problem of rotational motion behind. Nevertheless, for the sake of completeness we first introduce a general quantum-mechanical Hamiltonian in Eckart-Watson representation, to make a quick leap into a simplified harmonic description for internal motion and rigid rotor for rotational degrees of freedom. This regime constitutes a basis for more sophisticated calculations. The first section introduces normal modes with emphasis on mathematical precision. Alongside we stress the 	points where approximations are made, for the reader to be aware of limitations of the current model. Second section exemplifies presented methodology on a simple system of two masses connected with harmonic springs in 1D. Section \ref{sec:watson} has supplementary character and underlines the formal route needed be taken to derive the Watson Hamiltonian, which in turn is immersed in the framework of normal vibrations. The last section (\ref{sec:rotor}) is a detailed discussion of rigid rotor model. Current review is largely based on classic books by Bunker\& Jensen \cite{Bunker} and Wilson,Decius\&Cross \cite{Wilson} supplied with original papers, and a few author's own comments, derivations and theorems.

\section{\label{sec:level2} General theory}
Before proceeding to the main topic, a few words of explanation are needed.
Section \ref{sec:watson} presents a detailed insight into the formal procedure of derivation of the Eckart-Watson form of quantum Hamiltonian for an arbitrary many-body system. This is a popular but rather specific approach, as is constrained to normal coordinates only. Any set of coordinates which is not linearly related to cartesian displacements brings significant complications. Many efforts were made to evaluate an exact representation of the kinetic energy operator (KEO) in curvlinear coordinates and Eckart embedding. Most recently an effective numerical procedure for expansion of nuclear KEO in arbitrary curvlinear coordinates was introduced by Yachmenev and Yurchenko \cite{Yachmenev2015} and successfully applied to four-atomic systems by Chubb, Yachmenev and Yurchenko \cite{Chubb}. Parallel analytic progress in the field has been made by Szalay \cite{Szalay2015}, \cite{Szalay2015a}. Several authors: Tennyson et al. \cite{Sutcliffe1991}, \cite{Kostin2002}, Jensen \cite{Jensen1988}, Csaszar \cite{Csaszar1995} and others also succeeded to evaluate practical representations for many-body Hamiltonians in an arbitrary embedding of the body-fixed frame. In each of these cases one is theoretically able to linearize the internal coordinates and subsequently transfer into normal modes level. Nonetheless we shall follow the historically first attempt due to its strong association with normal coordinates.  After separating the nuclear center of mass, and moving into the Eckart frame which guarantees minimal rotation-vibration coupling, some analytic manipulations yield the well known Watson Hamiltonian:

\begin{equation}\label{Watson-hamiltonian}
\hat{H}=\frac{1}{2}\sum_{i,j}\hat{M}_{i}\mu_{ij}\hat{M}_{j}+
\frac{1}{2}\sum_{i}\hat{P}_{i}^{2}+\hat{U}+\hat{V}
\end{equation}
where $\hat{U}$ is mass dependent contribution to the potential $\hat{V}$: 
\begin{equation}\label{Watson-hamiltonian}
\hat{U}=-\frac{\hbar^{2}}{8}\sum_{i}\mu_{ii}\equiv-\frac{\hbar^{2}}{8}Tr\mu
\end{equation}
and $\mu_{ij}$ is the inverse of the dynamical moment of inertia tensor\footnote{It is not exactly equal to the inverse of moment of inertia matrix, see section IV}. $\hat{M}_{i}=\hat{J}_{i}-\hat{p}_{i}-\hat{L}_{i}$  stands for the i-th component of the total angular momentum of the system and is defined by the rovibronic angular momentum, vibrational angular momentum and electronic angular momentum,		 respectively. The potential operator $V$ contains interactions between all particles in the system, while $\hat{P}$ is the linear momentum operator in normal coordinates representation. By imposing Born-Oppenheimer separation of nuclear and electronic degrees of freedom the molecular Hamiltonian becomes a sum of an electronic and a rovibrational part:
\begin{equation}\label{Watson-hamiltonian}
\hat{H}=\hat{H}_{el}+\hat{H}_{rv}
\end{equation}
The rovibrational Hamiltonian can be readily expanded and arranged in the following fashion:
\begin{align}\nonumber\label{eq:rv}
  \hat{H}_{rv}&=\frac{1}{2}\sum_{\alpha}\mu_{\alpha\alpha}^e\hat{J}_{\alpha}^{2}+\frac{1}{2}\sum_{i}\hat{P}_{i}^{2}+  && \\\nonumber
 &  +\frac{1}{2}\sum_{\alpha\beta}\left(\mu_{\alpha\beta}-\mu_{\alpha\beta}\right)\left(\hat{J}_{\alpha}-\hat{p}_{\alpha}\right)\left(\hat{J}_{\beta}-
\hat{J}_{\beta}\right) &&\\\nonumber
 &-\sum_{\alpha}\mu_{\alpha\alpha}^e\hat{J}_{\alpha}\hat{p}_{\alpha}+ +\sum_{\alpha}\mu_{\alpha\alpha}^e\hat{p}_{\alpha}^2 + \hat{U}+\hat{V}_{nn} &&\\
\end{align}

Mixing terms appearing in above expression are responsible for centrifugal distorsion and Coriolis couplings. For sufficiently rigid molecules they can be neglected. What remains is a purely rotational part added to internal motion energy operators decoupled from rotations:

\begin{equation}
\hat{H}_{rv}=\frac{1}{2}\sum_{\alpha}\mu_{\alpha\alpha}^e\hat{J}_{\alpha}^{2}+\frac{1}{2}\sum_{i}\hat{P}_{i}^{2}+\hat{V}_{nn}
\label{eq:rv_harm}
\end{equation}

The task now is to find normal coordinates for a general many-body system, so that Watson approach may be realized.
Hence it becomes clear that the choice of this particular type of Hamiltonian representation was strongly dictated by its association to normal coordinates. However from practical point of view, one usually prefer start with a set of geometrically defined internal coordinates or cartesian coordinates. Thus, below we present a detailed instructions to obtain normal coordinates from arbitrarily chosen internal coordinates. Assuming the knowledge of normal coordinates of the system we may expand inter-nuclear interaction potential energy as:

\begin{equation}
V_{nn}=\frac{1}{2}\sum_{i=1}^{3N-6}\lambda_{i}Q^2_i+\frac{1}{6}\sum_{i,j,k=1}^{3N-6}\Phi_{ijk}Q_iQ_jQ_k+...
\label{eq:V_int} 
\end{equation}
and by truncating this expansion on quadratic terms we stay in harmonic approximation, in which normal coordinates were introduced. Now pick up the first two sums in eq. \ref{eq:rv_harm} and combine the second one with the truncated potential energy. By neglecting all other terms we get:
 
\begin{equation}
\hat{H}_{rv}=\frac{1}{2}\sum_{\alpha}\mu_{\alpha\alpha}^e\hat{J}_{\alpha}^{2}+
\frac{1}{2}\sum_{i}\left(\hat{P}_{i}^{2}+\lambda_{i}Q^2_i\right)
\label{eq:rv}
\end{equation}
The first term in above Hamiltonian describes the energy of a rigid rotor - this model is solved and discussed in detail in section \ref{sec:rotor}. The second term is the target of the present article. It is already shown in the final form, however no recipe for finding $\lambda_i$ and $Q_i$ was yet given. Until then this expression should be considered a black-box.

Having separated both center of mass motion and rotations via embedding in the Eckart body-fixed frame with later transfer to the rigid rotor approximation we are left with internal nuclear motion problem. For a non-linear molecule, or non-linear system of atoms this leaves $3N-6$ internal degrees of freedom. Lets start from a fully classical picture. After formulating a suitable representation of the Hamilton function, quantization will be performed. From eq. \ref{eq:rv_harm} one can infer that the total energy of the non-rotating system is a sum of kinetic and potential energy:
\begin{equation}
E=T(\dot{\textbf{x}})+V(\textbf{x})
\end{equation}
where $\textbf{x}=x_1,y_1,z_1,...,x_{3N-6},y_{3N-6},z_{3N-6}$ and $\dot{\textbf{x}}=\dot{x}_1,\dot{y}_1,\dot{z}_1,...,\dot{x}_{3N-6},\dot{y}_{3N-6},\dot{z}_{3N-6}$ are position and velocity vectors of all $3N-6$ atoms, respectively. No net external force is exerted on the system and all internal forces are conservative, hence the total energy is equal to the Hamilton function.
We should be able to write \textit{ad hoc} both components in cartesian coordinates:
\begin{equation}
T_{cart}=\frac{1}{2}\sum_{i=1}^{3N-6}\sum_{\alpha=x,y,z}m_i\dot{\alpha_i}^2
\end{equation}
where $\alpha_i$ stands for cartesian coordinate in the molecule-fixed (body-fixed) frame. It is very instructive to reformulate the above expression in a matrix form:
\begin{equation}
T_{cart}=\frac{1}{2}\sum_{i=1}^{3N-6}\left(\dot{x}_i,\dot{y}_i,\dot{z}_i\right)\left(\begin{array}{c c c}
m_i & 0 & 0 \\
0 & m_i & 0 \\
0 & 0 & m_i \\
\end{array}\right)\left(\begin{array}{c}
\dot{x}_i \\
\dot{y}_i  \\
\dot{z}_i \\
\end{array}\right)
\end{equation}
and further:
\begin{equation}
T_{cart}=\frac{1}{2}\dot{\textbf{x}}^T\textbf{T}\dot{\textbf{x}}
\label{eq:kin_cart}
\end{equation}
where $\textbf{x}$ is a $3N-6$ dimensional position vector, while $\textbf{T}$ stands for a $(3N-6)\times (3N-6)$ matrix representation of a diagonal metric tensor build from $N-2$ diagonal blocks each containing three equal masses $m_i\cdot \textbf{1}_{3}$. It is clearly visible that the kinetic energy is a quadratic form in cartesian velocities. This observation will allow us to utilize some tools from quadratic forms theory in section \ref{sec:math}. The internuclear potential energy is a function of $3N$ cartesian coordinates of all nuclei. However the functional form of potential energy is usually available in a geometrically defined internal coordinates, i.e. $V_{nn}=V_{nn}(R_1,R_2,...,R_{3N-6})$, where $R_i$ is the i-th internal coordinate. 
The reason for that is because the potential energy function is in general a complicated function of cartesian coordinates. Internal coordinates chosen intuitively by geometrical considerations provide a more transparent perspective of types of internal motions and associated potential energy. Therefore internal coordinate representation is more suitable at this stage. Nevertheless, the potential energy function is typically also a fairly complicated function (but less complicated than in cartesian representation) of internal coordinates. Thus, a for non-high-accuracy calculation purposes lowest terms in the Taylor expansion are sufficient: 
\begin{equation}
V_{nn}=\frac{1}{2}\sum_{i,j=1}^{3N-6}f_{ij}R_iR_j+\frac{1}{6}\sum_{i,j,k=1}^{3N-6}f_{ijk}R_iR_jR_k+...
\label{eq:V_int_harm} 
\end{equation}
Note that the minimum of $V$ is set to 0 and it corresponds to equilibrium values of all internal coordinates: $\vec{R}=(0,0,...,0)=\vec{R}_{eq}$. First derivatives vanish at this point. $f_{ij},f_{ijk},...$ are appropriate derivatives of $V$ with respect to internal coordinates taken in equilibrium positions. \\

As it is exceptionally difficult to obtain internal coordinates representation of the kinetic energy (quantum-mechanical operator) \cite{Csaszar1995},\cite{Tennyson1986} we are pushed to look for alternative approximate ways\footnote{We will stay however in fully classical formulation, to show later simple transformation into quantum mechanics.}. One such approach is based on the assumption that in vicinity of the equilibrium the potential energy function can be well represented only with the quadratic terms in the expansion \ref{eq:V_int_harm}.  In such case the potential energy may be also written in a matrix form:
\begin{equation}
V_{harm}=\frac{1}{2}\textbf{R}^T\textbf{F}\textbf{R}
\end{equation}
with $\textbf{F}$ as force constants matrix, i.e. the matrix of second derivatives of potential energy with respect to appropriate internal coordinates taken in equilibrium. If one would like to include anharmonicities, then for example the cubic terms should be understood as a contraction of rank 3 tensor with components $f_{ijk}$ with a rank 3 tensor constructed from internal coordinate vectors $\vec{R}$.

Still, internal coordinates are nonlinear functions of cartesian displacements from equilibrium geometry, what makes it technically very difficult to transform the kinetic energy operator into internal coordinate representation. Thus, we are pushed to make another crude approximation, that is to expand each internal coordinate in a series of cartesian displacements $\Delta\alpha_k=\alpha_k-\alpha^{e}_k$ :

\begin{equation}
R_{i}=\sum_{k=1}^{N}\sum_{\alpha=x,y,z}B_{i}^{\alpha k}\Delta\alpha_k+o\left((\Delta\alpha_k)^2\right).
\end{equation}
Here $B_{i}^{\alpha k}=\left(\frac{\partial R_i}{\partial \Delta\alpha_k}\right)_{eq}$ depends only on equilibrium positions of atoms. Again in the matrix form we have:
\begin{equation}
\textbf{R}=\textbf{B}\Delta\textbf{x}
\end{equation}
Because we started from general considerations for $3N$ degrees of freedom, $\textbf{B}$ matrix links the $3N$ cartesian displacements with the $3N-6$ linearized coordinates, hence is a rectangular matrix. In order to make it invertible we shall add 6 rows which give transformation laws for 3 translations and 3 rotations of the whole system. These are defined as follows:
\begin{equation}
T_{\alpha}=\frac{1}{M^{1/2}}\sum_{i=1}^Nm_{i}\Delta\alpha_i
\end{equation}
where $M$ is the total mass of the system, and	
\begin{equation}
R_{x}=\mu_{xx}^{eq}\sum_{i=1}^Nm_{i}\left(y_i^{eq}\Delta z_i-z_i^{eq}\Delta y_i\right)
\end{equation}
while $y$ and $z$ rotational coordinates are obtained by cyclic permutation of cartesian coordinates in above definition. $\mu_{xx}^{eq}=\left[\sum_{i=1}^Nm_i((y_i^{eq})^2+(z_i^{eq})^2)\right]^{-1}$ is the $x,x$ component of the inverse moment of the inertia tensor. After adding these relations, the $\textbf{B}$ matrix is automatically reshaped into square invertable form. Further results will show however, that this step is redundant if one requires only knowledge of normal coordinates. It is rather intuitive, as neither translations nor rotations can affect the internal degrees of freedom in a decoupled model. As a matter of fact the rectangular form is sufficient to transform the KEO into normal modes form.

Note that time derivative of a cartesian displacement is exactly equal to time derivative of cartesian absolute coordinate as equilibrium position doesn't change over time. Thus, in the expression for kinetic energy in eq. \ref{eq:kin_cart} we may responsibly put cartesian displacements:
\begin{equation}
T_{cart}=\frac{1}{2}\Delta\dot{\textbf{x}}^T\textbf{T}\Delta\dot{\textbf{x}}
\label{eq:kin_del_cart}
\end{equation}
and readily transfer into our linearized coordinates:
 \begin{equation}
T_{cart}=\frac{1}{2}\dot{\textbf{S}}^T\textbf{B}^T\textbf{T}\textbf{B}\dot{\textbf{S}}\equiv \frac{1}{2}\dot{\textbf{S}}^T\textbf{G}^{-1}\dot{\textbf{S}}
\label{eq:kin_del_cart}
\end{equation}
Here we introduced a new matrix: $\textbf{G}=\textbf{B}\textbf{T}^{-1}\textbf{B}^T$. This definition of the $\textbf{G}$ matrix indicates its dependence on atomic masses and the equilibrium geometry of the system. Despite its simple form, evaluation of $G$ matrix is a tedious job, especially for polyatomic systems. Elements of this matrix for several example systems are presented in the appendix in ref.\cite{Wilson}. In contrast, the harmonic potential energy function transforms trivially:
\begin{equation}
V_{harm}=\frac{1}{2}\textbf{S}^T\textbf{F}\textbf{S}
\label{eq:V_harm}
\end{equation}
Hence we managed to represent the total classical energy of a closed system of interacting point masses in the harmonic approximation and in terms of linearized coordinates as:
\begin{equation}
H=\frac{1}{2}\dot{\textbf{S}}^T\textbf{G}^{-1}\dot{\textbf{S}}+\frac{1}{2}\textbf{S}^T\textbf{F}\textbf{S}
\end{equation}
Unfortunately such form of the Hamilton function would result in coupled equations of motion due to non-diagonal forms of both $\textbf{G}$ and $\textbf{F}$. From this reason it is highly desirable to transform this quadratic form into a diagonal representation. This task is equivalent to simultaneous diagonalization of two matrices. In general our wanted transformation may be written as:
\begin{equation}
\textbf{S}=\textbf{L}\textbf{Q}
\end{equation}
where $\textbf{L}$ is a square $3N-6$ dimensional transformation matrix and \textbf{Q} is a transformed $3N-6$ dimensional vector assuring diagonal form of the total energy:
\begin{equation}
H=\frac{1}{2}\dot{\textbf{Q}}^T\dot{\textbf{Q}}+\frac{1}{2}\textbf{Q}^T\Lambda\textbf{Q}
\label{eq:H_diag}
\end{equation}
with $\Lambda$ being a diagonal matrix. Above conditions uniquely define $\textbf{L}$ matrix for a given $\textbf{G}$ and $\textbf{F}$. The first one indicates that $L$ is not in general an orthogonal matrix, by requiring:
\begin{equation}
\textbf{L}^{T}=\textbf{L}^{-1}\textbf{G}
\end{equation}
whereas the second one refers to the force constant matrix transformation:
\begin{equation}
\textbf{L}^{T}\textbf{F}\textbf{L}=\Lambda
\end{equation}
These can be unified into a single relation:

\begin{equation}
\textbf{L}^{-1}\textbf{G}\textbf{F}\textbf{L}=\Lambda
\end{equation}
Hence $\textbf{L}$ gains the interpretation of the matrix that diagonalizes the product $\textbf{GF}$ with normalization condition $\textbf{L}^{T}\textbf{L}=\textbf{G}$. In general $\textbf{GF}$ is not a symmetric matrix, even though both $\textbf{F}$ and $\textbf{G}$ are symmetric, since they normally do not commute. As a consequence we don't have a guarantee that this matrix is diagonalizable. Therefore, it is practical to perform such a transformation of the original linearized coordinates so that $\textbf{G}$ is a unit matrix. This problem is carefully discussed in section \ref{sec:math}. Having this done, it is straightforward to relate cartesian displacements to normal coordinates and \textit{vice versa}:
\begin{equation}
\Delta\textbf{x}=\textbf{B}^{-1}\textbf{L}\textbf{Q}
\end{equation}
where it is immediate to show via \textit{Kronecker theorem} for determinants that $\textbf{B}^{-1}\textbf{L}$ is invertible. The inverse of $\textbf{B}$ can be expressed by other matrices, thereby avoiding undefined operation (probably taking pseudoinverse of \textbf{B} would give congruent result). In addition, from various reasons it is convenient to express left hand side of above equation as mass weighted cartesian displacements:
 \begin{equation}
\Delta\textbf{q}=\textbf{M}^{\frac{1}{2}}\Delta\textbf{x}=\textbf{M}^{\frac{1}{2}}\textbf{B}^{T}\textbf{G}^{-1}\textbf{L}\textbf{Q}
\end{equation}
which can be obtained by simple matrix manipulation. The last equation defines well known $\textbf{l}$ matrix:
 \begin{equation}
\textbf{l}:=\textbf{M}^{\frac{1}{2}}\textbf{B}^{T}\textbf{G}^{-1}\textbf{L}
\end{equation}
which is orthogonal $\textbf{l}^{-1}=\textbf{l}^{T}$ so that the effective formulas for cartesian displacements and normal coordinates are respectively:
 \begin{eqnarray}
 \Delta\alpha_i=\sum_{j=1}^{3N-6}m_j^{-\frac{1}{2}}l_{\alpha j,i}Q_{j}\\
  Q_i=\sum_{j=1}^{3N-6}m_j^{\frac{1}{2}}\left(l_{\alpha j,i}\right)^{T}\Delta\alpha_j
 \end{eqnarray}
We shall make extensive use of the $\textbf{l}$ matrix in section \ref{sec:watson}. What remains, is to recall the potential energy function defined in the beginning of this section and relate $\lambda_i$ and $f_{ijk}$ etc. to the $\textbf{L}$ matrix:

\begin{equation}
V_{nn}=\frac{1}{2}\sum_{i,j=1}^{3N-6}f_{ij}R_iR_j+\frac{1}{6}\sum_{i,j,k=1}^{3N-6}f_{ijk}R_iR_jR_k+...
\end{equation}

Note that if the initial internal coordinates were chosen to be cartesian displacements then $\textbf{B}$ matrix is a unit matrix, and a more familiar relation is retrieved: 

 \begin{equation}
\textbf{L}^{-1}\textbf{F}\textbf{L}=\Lambda
\end{equation}
and $\textbf{l}=\textbf{M}^{-\frac{1}{2}}\textbf{L}$.

$3N-6$ dimensional column vectors of type
\begin{equation}\vec{Q}_{s}=
\left(\begin{array}{c}
0 \\
0 \\
\vdots \\
0 \\
1 \\
0 \\
\vdots\\
0
\end{array}\right)\begin{array}{c}
  \\
  \\
 \\
 s \\
  \\
  \\
\\

\end{array}
\label{eq:basis}
\end{equation}

\begin{figure}
\includegraphics[scale=0.9]{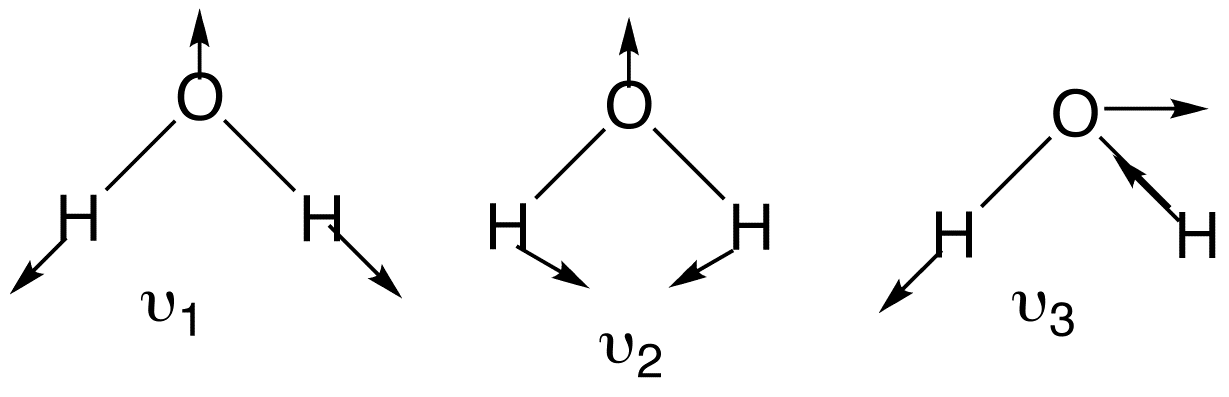} 
\label{fig:normal1}
\caption{Visualization of the normal modes for water molecule. Arrows correspond to vectors having its beginning in equilibrium positions of respective atoms and endings in points with cartesian coordinates corresponding to basis vectors of normal coordinates space from eq.\ref{eq:basis} }
\end{figure}

form an orthonormal basis for normal coordinates space. Thus, by plotting cartesian displacements corresponding to each of the basis vector one may visualize geometry distortion caused by a particular normal vibration. This is indicated with arrows in fig. \ref{fig:normal1}

\section{Simple example}

In this section the just described procedure is exemplified on an arguably simplest non-trivial model: two masses + three springs in one dimension. This is a classic system presented in textbooks for pedagogical purposes. 

\subsection{Potential energy}
We have a system of two masses interconnected by a spring with force constant $k$. Additionally these masses are connected to infinite mass, rigid walls by the same type of springs. This gives 2 degrees of freedom represented by catesian coordinates of respective masses. By investigating fig. \ref{fig:springs} we may write down an expression for the potential energy of the system in harmonic approximation:
\begin{equation}
\resizebox{\hsize}{!}{$V(x_1,x_2)=\frac{1}{2}k\left(x_1-x_1^e\right)^2+\frac{1}{2}k\left(x_2-x_2^e\right)^2+\frac{1}{2}k\left(x_2-x_2^e-x_1+x_1^e\right)^2$}
\end{equation}
\begin{figure}
\includegraphics[scale=0.4]{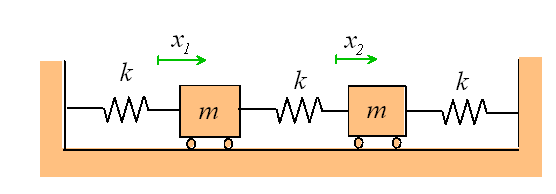}
\caption{Graphical representation of the system of two masses on three springs.}
\label{fig:springs}
\end{figure}
This is the cartesian representation. All two degrees of freedom are identified as vibrational, as in 1D case there are no rotations, and our system is being held still by infinitely heavy walls, hence no center of mass motion. Now we choose appropriate internal coordinates. In such case they can be simply displacements from equilibrium:
\begin{equation}
\begin{split}
R_1:=\Delta x_1=x_1-x_1^e, \\
R_2:=\Delta x_2=x_2-x_2^e \\
\end{split}
\end{equation}
In this new representation the potential energy reads:
\begin{equation}
V(R_1,R_2)=kR_1^2+kR_2^2-\frac{1}{2}kR_1R_2-\frac{1}{2}kR_2R_1
\label{eq:example_int}
\end{equation}
Newtonian equations of motion at this stage are coupled second order differential equations:
\begin{equation}
\left(\begin{array}{c}
\ddot{R_1} \\
\ddot{R_2} \\ 
\end{array}\right) = \left(\begin{array}{c c}
-\frac{2k}{m} & \frac{k}{m} \\
\frac{k}{m} & -\frac{2k}{m} \\ 
\end{array}\right)\left(\begin{array}{c}
R_1 \\
R_2 \\ 
\end{array}\right) 
\end{equation}
hence quite problematic to solve. 
\subsection{Linearized coordinates}
Second step in the procedure is evaluation of linearized coordinates:
\begin{equation}
S_i=\sum_{k=1}^{2}\left(\frac{\partial R_i}{\partial \Delta x_k}\right)_eq\Delta x_k
\end{equation}
In this particular case the transformation is trivial:
\begin{equation}
\begin{split}
S_1=\left(\frac{\partial R_1}{\partial \Delta x_1}\right)_{eq}\Delta x_1+\left(\frac{\partial R_1}{\partial \Delta x_2}\right)_{eq}\Delta x_2=\Delta x_1 \\
S_2=\left(\frac{\partial R_2}{\partial \Delta x_1}\right)_{eq}\Delta x_1+\left(\frac{\partial R_2}{\partial \Delta x_2}\right)_{eq}\Delta x_2=\Delta x_2 \\
\end{split}
\end{equation}
so that the \textbf{B} matrix is a unit matrix:
\begin{equation}
\textbf{B=}\left(\begin{array}{c c}
1& 0 \\
0 & 1 \\ 
\end{array}\right)=\textbf{1}_2
\end{equation}
Here one comment: we agreed that all degrees of freedom are vibrational. Therefore $\textbf{B}$ matrix is already invertable and there's no need to extend it by rotational and translational coordinates. However if one tries to calculate the x-component of the translational coordinate:
\begin{equation}
T_x=\sqrt{\mu}\left(\Delta x_1 + \Delta x_2\right)
\end{equation}
one may argue that we should have one vibrational degree of freedom and one translational center of mass motion described by above expression. Since we assumed infinite mass of walls the center of mass of the system as a whole is always frozen. In such case $T_x$ coordinate will become a vibrational coordinate, which doesn't affect the center of mass. Such seemingly ambiguous situation is an artifact of confinement of the system in 1D between static walls. Hence we agree to call all our degrees of freedom vibrational. Moving on, the \textbf{G} matrix reads:
\begin{equation}
\textbf{G}=\left(\begin{array}{c c}
1& 0 \\
0 & 1 \\ 
\end{array}\right)\left(\begin{array}{c c}
\frac{1}{m}& 0 \\
0 & \frac{1}{m} \\ 
\end{array}\right)\left(\begin{array}{c c}
1& 0 \\
0 & 1 \\ 
\end{array}\right)=\left(\begin{array}{c c}
\frac{1}{m}& 0 \\
0 & \frac{1}{m} \\ 
\end{array}\right)
\end{equation}
and

\begin{equation}
\textbf{G}^{-1}=\left(\begin{array}{c c}
m & 0 \\
0 & m \\ 
\end{array}\right)
\end{equation}
Finally from eq. \ref{eq:example_int} it is straightforward to read the elements of the force constant matrix:
\begin{equation}
\textbf{F}=\left(\begin{array}{c c}
2k & -k \\
-k & 2k \\ 
\end{array}\right)
\end{equation}
and the Hamilton function takes the following form:
\begin{equation}
H=\frac{1}{2}\dot{\textbf{S}}^T\textbf{G}^{-1}\dot{\textbf{S}}+\frac{1}{2}\textbf{S}^T\textbf{F}\textbf{S}
\end{equation}
\subsection{Normal coordinates}
We are looking for a matrix \textbf{L} that transforms normal to linearized coordinates $\textbf{S}=\textbf{L}\textbf{Q}$,  diagonalizes $\textbf{GF}$ matrix and preserve the correct norm:
\begin{equation}
\begin{split}
\textbf{L}^{-1}\textbf{G}\textbf{F}\textbf{L}=\Lambda \\
\textbf{L}^{T}\textbf{G}^{-1}\textbf{L}=\textbf{1}_2
\end{split}
\end{equation}
First lets calculate the $\textbf{GF}$ matrix:
\begin{equation}
\textbf{GF}=\left(\begin{array}{c c}
\frac{1}{m}& 0 \\
0 & \frac{1}{m} \\ 
\end{array}\right)\left(\begin{array}{c c}
2k & -k \\
-k & 2k \\ 
\end{array}\right)=\left(\begin{array}{c c}
\frac{2k}{m}& -\frac{k}{m} \\
-\frac{k}{m} & \frac{2k}{m} \\ 
\end{array}\right):=\omega_0\left(\begin{array}{c c}
2& -1 \\
-1& 2 \\ 
\end{array}\right)
\end{equation}
with $\omega_0=\frac{k}{m}$. The Kernel of the operator $\textbf{GF}-\textbf{1}\Lambda$ is spanned by two orthonormal vectors:
\begin{equation}
\frac{1}{\sqrt{2}}\left(\begin{array}{c}
1 \\
1\\
\end{array}\right), \; \frac{1}{\sqrt{2}}\left(\begin{array}{c}
1 \\
-1\\
\end{array}\right)
\end{equation}
which constitute the prospective $\textbf{L}$ matrix:
\begin{equation}
\textbf{L}=\frac{1}{\sqrt{2}}\left(\begin{array}{c c}
	1 & 1 \\
1 & -1 \\ 
\end{array}\right)
\end{equation}
and the eigenvalues matrix read:
\begin{equation}
\Lambda=\left(\begin{array}{c c}
	\omega_0 & 0 \\
0 & 3\omega_0 \\ 
\end{array}\right)
\end{equation}
One easily checks that $\textbf{L}$ is orthogonal, hence:
\begin{equation}
\left(\begin{array}{c}
Q_1 \\
Q_2 \\ 
\end{array}\right) = \frac{1}{\sqrt{2}}\left(\begin{array}{c c}
	1 & 1 \\
1 & -1 \\ 
\end{array}\right)\left(\begin{array}{c}
\Delta x_1 \\
\Delta x_2 \\ 
\end{array}\right) 
\end{equation}
giving two non-degenerate normal modes: symmetric stretching and asymmetric stretching respectively. We may also want to evaluate the $\textbf{l}$ matrix:
\begin{equation}
\textbf{l}=\textbf{M}^{-\frac{1}{2}}\textbf{B}^{T}\textbf{G}^{-1}\textbf{L}=\sqrt{\mu}\left(\begin{array}{c c}
	1 & 1 \\
1 & -1 \\ 
\end{array}\right)
\end{equation}
to give: 
\begin{equation}
\begin{split}
\sqrt{m}\Delta x_1=\sqrt{\mu}\left(Q_1+Q_2\right) \\
\sqrt{m}\Delta x_2=\sqrt{\mu}\left(Q_1-Q_2\right) \\
\end{split}
\end{equation}
The Hamilton function is now a sum of two independent components:
\begin{equation}
H(Q_1,Q_2)=H_1(Q_1)+H_2(Q_2)
\end{equation}
where $H_i(Q_i)=\frac{1}{2}\dot{Q}_i^2+\frac{1}{2}\Lambda_{ii}Q_i^2$. Having the classical expression for Hamilton function derived as well as having invertible transformation between cartesian and normal representation we are ready to turn into quantum mechanics. The kinetic energy operator in cartesian coordinates reads:
\begin{equation}
\hat{T}_{cart}=\frac{1}{2m}\left(\Delta_1+\Delta_2\right)=\frac{1}{2m}\left(\frac{\partial^2}{\partial x_1^2}+\frac{\partial^2}{\partial x_2^2}\right)
\end{equation}
Our 'curvlinear' coordinates $\vec{R}=(R_1,...,R_{3N-6})$ role is now taken by the derived normal coordinates:

\begin{equation}
\vec{R}\left(\vec{x}\right): \qquad \left(\begin{array}{c}
Q_1 \\
Q_2 \\ 
\end{array}\right) = \frac{1}{\sqrt{2}}\left(\begin{array}{c c}
	1 & 1 \\
1 & -1 \\ 
\end{array}\right)\left(\begin{array}{c}
\Delta x_1 \\
\Delta x_2 \\ 
\end{array}\right) 
\end{equation}

\begin{equation}
\vec{x}\left(\vec{R}\right): \qquad \left(\begin{array}{c}
\Delta x_1 \\
\Delta x_2 \\ 
\end{array}\right) = \frac{1}{\sqrt{2}}\left(\begin{array}{c c}
	1 & 1 \\
1 & -1 \\ 
\end{array}\right)\left(\begin{array}{c}
Q_1 \\
Q_2 \\ 
\end{array}\right) 
\end{equation}
At this stage we could step forward in twofold way. We could either start from cartesian representation and apply chain rule to get internal coordinate representation or use the derived separable form of Hamilton function and apply the Podolsky trick. An excellent paper on this freedom of choice is given in ref.\cite*{Islampour1983}. We choose the former method as more straightforward in this case:
\begin{equation}
\begin{split}
\frac{\partial^2}{\partial x_1^2}=\frac{1}{2}\left(\frac{\partial^2}{\partial Q_1^2}+\frac{\partial^2}{\partial  Q_2^2}+2\frac{\partial^2}{\partial Q_1^2}\frac{\partial^2}{\partial Q_2^2}\right) \\
\frac{\partial^2}{\partial x_2^2}=\frac{1}{2}\left(\frac{\partial^2}{\partial Q_1^2}+\frac{\partial^2}{\partial  Q_2^2}-2\frac{\partial^2}{\partial Q_1^2}\frac{\partial^2}{\partial Q_2^2}\right) \\
\end{split}
\end{equation}
Note that only the sum of above two provides form invariance under linear transformation of cartesian coordinates:
\begin{equation}
\hat{T}_Q=\frac{1}{2m}\left(\frac{\partial^2}{\partial Q_1^2}+\frac{\partial^2}{\partial Q_2^2}\right)
\end{equation}
By setting $Q_1=1$ and $Q_2=0$ and reversely we may depict cartesian displacements along a chosen normal coordinate.

\section{\label{sec:math} Mathematical background}

In order to understand the problem of finding of normal vibrations at the deepest level, it is necessary to recall several mathematical definitions and theorems, commonly used in linear operators theory \cite{Fuller}:

\begin{definition}
If an operator $A$ is diagonalizable, and $\beta$ diagonalizes it by the similarity transformation:

\begin{equation}
\begin{split}
\beta^{-1}A\beta=D, \ \ D_{ij}=\lambda_{i}\delta_{ij}
\end{split}
\end{equation}

then for any real number $\gamma$ we may define the power of the operator as follows:

\begin{equation}
\begin{split}
A^{\gamma}:=\beta D^{\gamma}\beta^{-1}
\end{split}
\end{equation}
If operator $A$ contains any zero's in its spectrum, then $\gamma>0$.

\label{def:power}
\end{definition}

\begin{theorem}
Any matrix $A$ defined above and any pair real numbers compatible with def. \ref{def:power} satisfies:
\begin{equation}
\begin{split}
A^{\gamma}A^{\beta}=A^{(\gamma+\beta)}
\end{split}
\end{equation}
\end{theorem}

\begin{proof}
\begin{equation}
\begin{split}
A^{\gamma}A^{\beta}=\beta D^{\gamma}\beta^{-1}\beta D^{\beta}\beta^{-1}=\beta D^{\gamma}D^{\beta}\beta^{-1}=...
\end{split}
\end{equation}

As long as $D$ are diagonal, exponentiation converts into simple number operations:
\begin{equation}
\begin{split}
\left(D^{\gamma}D^{\beta}\right)_{kl}=(\lambda_{k}^{\gamma}\lambda_{l}^{\beta}I^{(\gamma+\beta)})_{kl}=(\lambda_{k}^{(\gamma+\beta)}I)_{kl}
\end{split}
\end{equation}

\begin{equation}
\begin{split}
...=\beta D^{(\gamma+\beta)}\beta^{-1}=A^{(\gamma+\beta)} \quad \diamond
\end{split}
\end{equation}

\end{proof}

In classical and even more in quantum mechanics, special role is assigned to self-adjoint operators, in particular to those which are additionally positive definite $ (\forall x \ (\textbf{x},\textbf{Ax})\geq0)$. It follows from the fact that positive definite operators posses only non-negative eigenvalues, which is a physical requirement for observables. Non-negative spectrum is an immediate consequence of positive definitiness:
\begin{equation}
\begin{split}
\textbf{Ax}=\lambda \textbf{x} \Rightarrow \lambda = \frac{(\textbf{x},\textbf{A}\textbf{x})}{(\textbf{x},\textbf{x})}
\end{split}
\end{equation}

\begin{lemma}
If an operator $A$ is self-adjoint, then $A^{\gamma}$ is also self-adjoint.
\end{lemma}

\begin{proof}

Matrix that diagonalizes a self-adjoint operator is unitary. It follows from the equality:

\begin{equation}
\begin{split}
\left(A^{\gamma}\right)^\dag=\left(UD^{\gamma}U^{\dag}\right)^\dag=(U^{\dag})^{\dag}D^{\gamma}U^{\dag}=A^{\gamma}
\end{split}
\end{equation}
because \textbf{D} is diagonal and have all its elements positive( as the eigenvalues of $A$ are real).$\diamond$

\end{proof}

Lets turn into application in classical mechanics. Kinetic and potential energy functions may be expressed as expectation values of some linear operators (cf. \ref{eq:kin_cart}), namely:

\begin{equation}
\begin{split}
\langle T\rangle=\frac{1}{2}\sum_{i,j=1}^{n}T_{ij}\dot{\alpha_{i}}\dot{\alpha_{j}}
\end{split}
\end{equation}

where: $T_{ij}$ are time independent matrix elements of kinetic energy operator (classical!) in a basis constructed from time derivatives of consecutive coordinates; $n$ is the number of degrees of freedom. By utilizing natural scalar product of Euclidean space $(\textbf{e}_{i},\textbf{e}_{j})=\delta_{ij}, i,j=1,2,3$ we get:

\begin{equation}
\begin{split}
\langle T\rangle=\frac{1}{2}(\dot{\textbf{x}},\textbf{T}\dot{\textbf{x}})
\end{split}
\end{equation}
Round brackets are just another way to denote the scalar product - here to distinguish it from Dirac brackets representing scalar product in $\textit{L}^2(\textbf{R}^3,d^3r)$ space. As far as the kinetic energy is required to be non-negative (and in quantum case positive) the operator $T$ is positive-definite, real and symmetric, hence also self-adjoint.
Potential energy, as a complicated function of cartesian coordinates is usually expanded in Taylor series and, in our approach, truncated at quadratic terms (see, eq. \ref{eq:V_int} and eq. \ref{eq:V_harm}  ):

\begin{equation}
\begin{split}
\langle V\rangle=\frac{1}{2}\sum_{i,j=1}^{n}V_{ij}x_{i}x_{j}=\frac{1}{2}(\textbf{x},\textbf{Vx})
\end{split}
\end{equation}
Near the minimum, the second derivative is always positive, hence for a small displacement from equilibrium following inequality holds:

\begin{equation}
\begin{split}
\left(\frac{\partial^{2} V}{\partial \alpha_{i}\partial \alpha_{j}}\right)_{x=0}\Delta \alpha_{i}\Delta \alpha_{j}>0
\end{split}
\end{equation}
where $\alpha_i$ is a general cartesian coordinate of the the i-th atom, as defined in previous section.
Thereby we may assume that in the range of applicability of the normal coordinates the potential energy classical operator is real, symmetric and positive-definite:

\begin{equation}
\begin{split}
(\textbf{x},\textbf{Vx})\geq 0
\end{split}
\end{equation}
Lagrange function in this approximation reads:
\begin{equation}
\begin{split}
L=\frac{1}{2}(\dot{\textbf{x}},\textbf{T}\dot{\textbf{x}})-\frac{1}{2}(\textbf{x},\textbf{Vx})
\end{split}
\end{equation}

And Euler-Lagrange equations

\begin{equation}
\begin{split}
\frac{\partial L}{\partial \alpha_{i}}-\frac{d}{dt}\left(\frac{\partial L}{\partial \dot{\alpha}_{i}}\right)=0 \ \ i=1,...,n
\end{split}
\end{equation}

yield in equations of motion in the form:

\begin{equation}
\begin{split}
\sum_{j=1}^{n}T_{ij}\ddot{\alpha}_{j}+\sum_{j=1}^{n}V_{ij}\alpha_{j}=0\ \ \ \i=1,...,n
\end{split}
\end{equation}
whereas in operator notation:

\begin{equation}
\begin{split}
\textbf{T}\ddot{\textbf{x}}+\textbf{Vx}=0
\end{split}
\label{eq:system}
\end{equation}

Here one faces numerous possible ways of solving of the above set of coupled differential equations. Of course we would like to see them decoupled. By and large it means that we want to find such basis in the euclidean space that both kinetic and potential energies are diagonal. Below we present two alternative ways to achieve it:

\begin{enumerate}

 \item  A formal approach is to postulate the form of solution in analogy to the simple harmonic oscillator:

\begin{equation}
\begin{split}
\alpha_i(t)=\xi e^{\pm\omega t} \Rightarrow (V-\omega^{2}T)\xi=0
\end{split}
\end{equation}

Solutions to this equation are given by a set being the Kernel of the $V-\omega^{2}T$\ operator, i.e. $\xi \in Ker(V-\omega^{2}T)$. Next we aim in reformulation of above equation into an eigenvalue problem for some operator.

To do that, lets first act on both sides with the operator $T^{-\frac{1}{2}}$ :

\begin{equation}
\begin{split}
T^{-\frac{1}{2}}V\xi=\omega^{2}T^{\frac{1}{2}}\xi
\end{split}
\end{equation}
Using the resolution of identity $T^{-\frac{1}{2}}T^{\frac{1}{2}}=id$ :

\begin{equation}
\begin{split}
T^{-\frac{1}{2}}VT^{-\frac{1}{2}}T^{\frac{1}{2}}\xi=\omega^{2}T^{\frac{1}{2}}\xi
\end{split}
\end{equation}

and denoting: $\eta =T^{\frac{1}{2}}\xi, \ \ W=T^{-\frac{1}{2}}VT^{\frac{1}{2}}$  we find:

\begin{equation}
\begin{split}
W\eta=\omega^{2}\eta
\end{split}
\end{equation}
of course by definition of $T$ and $V$, the $W$ operator is symmetric and real, which guarantees the existence of an orthonormal eigenbasis:

\begin{equation}
\begin{split}
(\eta^{(i)},\eta^{(j)})=\delta_{ij}
\end{split}
\end{equation}

This orthonormality condition translates onto orthonormality of initial eigenstates:
\begin{equation}
\begin{split}
(\eta^{(i)},\eta^{(j)})=(T^{\frac{1}{2}}\xi^{(i)}, T^{\frac{1}{2}}\xi^{(j)})=(\xi^{(i)},T\xi^{(j)})
\end{split}
\end{equation}

Therefore the special solution to the system of equations \ref{eq:system} is given by:
\begin{equation}
\begin{split}
\alpha^{(i)}_{\pm}(t)=\xi^{(i)}e^{\pm\omega_{i}t}
\end{split}
\end{equation}
The general solution is a linear combination of all special solutions:

\begin{equation}
\begin{split}
\alpha(t)=\sum_{k=1}^{n}\xi^{(k)}\left[C_{k}e^{\pm\omega_{k}t}+D_{k}e^{-\pm\omega_{k}t}\right]
\end{split}
\end{equation}

If provided initial conditions: $\alpha(0)=\kappa, \dot{\alpha}(0)=\beta$ and restricting only to real part of solutions, we finally obtain:

\begin{equation}
\begin{split}
\alpha(t)=\sum_{k=1}^{n}\xi^{(k)}\left[(\xi^{(k)},T\kappa)\cos\omega_{k}t+\frac{(\xi^{(k)},T\beta)}{\omega_{k}}\sin\omega_{k}t\right]
\end{split}
\end{equation}

Lets notice key features of the above solution:

\begin{itemize}
  \item a net motion is a superposition of vibrations of a precisely determined frequency
  \item Relative vibrational amplitudes are determined by components $\xi^{(j)}$
\end{itemize}

On the margin note that

\begin{equation}
\begin{split}
(V-\omega^{2}T)\xi = 0  \Rightarrow \ (T^{-\frac{1}{2}}VT^{-\frac{1}{2}}-\lambda I)\xi =0 , \lambda = \omega^{2}
\end{split}
\end{equation}
because $T$ commutes with $V$.

After diagonalization of $V-\omega^{2}T$ the Lagrangian takes separable form:

\begin{equation}
\begin{split}
L=\frac{1}{2}\left[(\dot{\textbf{Q}},\dot{\textbf{Q}})-(\textbf{Q},\Lambda \textbf{Q})\right]
\end{split}
\end{equation}

where: $\Lambda$ is a diagonal matrix of normal frequencies, while  $\textbf{Q}$ is a superposition of original coordinates according to the transformation: $\textbf{Q}=O^{T}T^{\frac{1}{2}}\textbf{x}$ , whereas $O$ is transition matrix into the eigenbasis of the $W$ operator.
In consequence, one can easily transform into Hamilton function to get:
\begin{equation}
\begin{split}
H=\frac{1}{2}\left[(\dot{\textbf{Q}},\dot{\textbf{Q}})+(\textbf{Q},\Lambda \textbf{Q})\right]
\end{split}
\end{equation}
Reader may compare this result to see that it represents exactly the same Hamilton function as in eq.\ref{eq:H_diag}.
\begin{equation}
\begin{split}
\alpha_{i}=\sum_{k=1}^{n}(T^{-\frac{1}{2}}O)_{ik}Q_{k}
\end{split}
\end{equation}
To avoid confusion we shall point out that $\textbf{V}$ matrix is identical with the force constant matrix $\textbf{F}$ and was intentionally written to show the link with the $V$ operator.
\item Lagrangian is simply a sum of two quadratic forms, which stay in 1:1 correspondence with metric matrices. The task is to simultaneously diagonalize those two quadratic forms. Beforehand lets prove the following statement:

\begin{theorem}

Let $V$ be euclidean space. In such case there is 1:1 correspondence between self-adjoint operators on $V$ and symmetric metrics $\tilde{g}$ on $V$:
    
    \begin{equation}
\begin{split}
\tilde{\textbf{g}}(\textbf{x},\textbf{y})=(\textbf{x},\textbf{Ay})
\end{split}
\end{equation}
and following relation holds:
\begin{equation}
\begin{split}
Ker\tilde{\textbf{g}}=Ker\textbf{A}
\end{split}
\end{equation}

If however $\textbf{g},\tilde{\textbf{g}},A$  are matrices in an arbitrary basis of two metrics on $V$ and $A$, then:

\begin{equation}
\begin{split}
\tilde{\textbf{g}}=\textbf{gA}
\end{split}
\end{equation}
\end{theorem}
Proof of this theorem is not instructive, hence will be omitted. By using above theorem we may conclude, that for any real number $\lambda$ similar equality holds:
\begin{equation}
\begin{split}
(\tilde{\textbf{g}}-\lambda \textbf{g})(\textbf{x},\textbf{y})=(\textbf{x},(\textbf{A}-\lambda id)\textbf{y})
\end{split}
\end{equation}
As a consequence, the subspace  $V_{\lambda}\equiv Ker(\tilde{\textbf{g}}-\lambda \textbf{g})$ is non-trivial if and only if $\lambda$ is an eigenvalue of the operator $A$, and this subspace is an eigensubspace of $A$. Note that if $\textbf{y}\in Ker(\tilde{\textbf{g}}-\lambda \textbf{g})$ then for any vector $\textbf{x} $ we have $\tilde{\textbf{g}}(\textbf{x},\textbf{y})=\lambda g(\textbf{x},\textbf{y})$. Thus, subspaces $V_{\lambda}$ orthogonal with respect to $\textbf{g}$ are at the same time orthogonal with respect to $\tilde{\textbf{g}}$, while for  metrics contracted to these subspaces there is a proportionality: $\tilde{\textbf{g}}V_{\lambda}=\lambda \textbf{g}V_{\lambda}$. It simply means that determination of the subspace $V_{\lambda}$ brings us a basis where both $\tilde{\textbf{g}}$ and $\textbf{g}$ are diagonal. This statement is a subject of the next theorem:

\begin{theorem}

Let $\textbf{g}$ and $\tilde{\textbf{g}}$ be two symmetric forms on a real, finite dimensional space $V$, and assume that $\textbf{g}$ is positive-definite (it can be checked e.g. via Sylvester criterion \citep{Fuller}). Let $ \lambda_{1},...,\lambda_{s}$ be all possible real numer for which the kernel $\tilde{\textbf{g}}-\lambda_{i}\textbf{g}$ is non-trivial, i.e. $Ker(\tilde{\textbf{g}}-\lambda_{i}\textbf{g})\neq{{0}}$. Then $V$ decomposed into a direct sum:

 \begin{equation}
\begin{split}
V=\bigoplus_{i=1}^{s}Ker(\tilde{\textbf{g}}-\lambda_{i}\textbf{g}),
\end{split}
\end{equation}

and this sum is orthogonal with respect to both metrics. In a basis orthonormal with respect to $\textbf{g}$, consistent with the space decomposition, the metric matrices take diagonal form:

\begin{equation}
\textbf{g'=1},\ \ \tilde{\textbf{{g}}}'=\left(\begin{array}{cccc}
\lambda_{1}\textbf{1}_{n_{1}}& 0 &...&0\\
0&\lambda_{2}\textbf{1}_{n_{2}}&...&0\\
..........&...........&............&..........\\
0&0&...&\lambda_{s}\textbf{1}_{n_{s}}

\end{array}\right)
\end{equation}
\end{theorem}
    In practice, determination of Kernel $Ker(\tilde{\textbf{g}}-\lambda_{i}\textbf{g})$ is reduced to solution of a system of homogeneous equations, together with a solvability condition:

     \begin{equation}
\begin{split}
det(\tilde{\textbf{g}}-\lambda \textbf{g})=0
\end{split}
\end{equation}

In the next step, for any found $\lambda$  we solve the system of equations:
\begin{equation}
\begin{split}
(\tilde{\textbf{g}}-\lambda \textbf{g})\textbf{x}=0
\end{split}
\end{equation}
Linearly independent vectors corresponding to the same characteristic value are further orthogonalized with respect to $\textbf{g}$ metric using \textit{Gram-Schmidt} method. After subsequent normalization we finally get:
\begin{equation}
\begin{split}
(\textbf{e}'_{1},...,\textbf{e}'_{n})=(\textbf{e}_{1},...,\textbf{e}_{n})\beta, \ \ \ \ \textbf{g}'=\beta^{T}\textbf{g}\beta, \ \ \tilde{\textbf{g}}'=\beta^{T}\tilde{\textbf{g}}\beta
\end{split}
\end{equation}

\end{enumerate}
The most effective and practical turns out to be the second approach, just described. Therefore we will try to sketch steps needed to be taken in order to find normal vibrations via simultaneous diagonalization of two quadratic forms:
\begin{enumerate}
  \item Having found the Hessian matrix $\textbf{F}$ (cf. \ref{eq:V_int_harm}, potential energy representation) as well as the kinetic energy matrix $\textbf{G}$  (from geometric considerations), we want to find such basis transformation so that both mentioned matrices are diagonal. If our initial coordinates are cartesian, then $\textbf{G}=\textbf{T}$. First, solve the following equation:

\begin{equation}
\begin{split}
det(\textbf{F}-\omega^{2}\textbf{G})= 0
\end{split}
\end{equation}

to determine the allowed frequencies of vibrations in the system.
\item Find bases of the Kernel of the operator: $\textbf{F}-\omega^{2}_{i}\textbf{G}$ where $i=1,..,n$:

\begin{equation}
\begin{split}
(\textbf{F}-\omega^{2}_{i}\textbf{G})\textbf{Q}= 0 \ \ i=1,...,n
\end{split}
\end{equation}

\item If any of the modes is degenerate, then perform orthonormalization with respect to $T$. The kinetic energy operator defines a quadratic form, which set a metric on euclidean space. The orthogonalization is performed with respect to this metric, what we write as: $(\textbf{e}_{i},\textbf{T}\textbf{e}_{j})=\delta_{ij}$.

\item As a result the Lagrangian takes separable form:
    \begin{equation}
\begin{split}
L=\frac{1}{2}\left[(\dot{\textbf{Q}},\beta^{T}\textbf{G}\beta\dot{\textbf{Q}})-(\textbf{\textbf{Q}},\beta^{T}\textbf{F}\beta \textbf{Q})\right]=\frac{1}{2}\sum_{i=1}^{n}\left(\dot{Q}_{i}^{2}-\omega^{2}_{i}Q_{i}^{2}\right)
\end{split}
\end{equation}
Transition matrix to the common eigenbasis is build of normalized vectors $\textbf{Q}$. The relation between original (cartesian) coordinates and new coordinates (normal) is given as: $\textbf{Q}=\beta \textbf{x}$, where $\textbf{F}_{diag}=\beta^{T}\textbf{F}\beta$ , $\textbf{G}_{diag}=\beta^{T}\textbf{G}\beta$. $\textbf{G}_{diag}$ has unit values on its diagonal, while diagonal elements of $\textbf{F}_{diag}$ represent eigen-frequencies of the system.
\item The separated Lagrangian guarantees decoupled equations of motion of harmonic oscillator type, which has a general solution:
\begin{equation}
\begin{split}
Q_{i}=A_{i}\cos\omega_{i}t+B_{i}\sin\omega_{i}t\ \ i=1,...,n
\end{split}
\end{equation}
\item Utilizing the relation between new and old coordinates (components respective vectors) we have:
 \begin{equation}
\begin{split}
\textbf{x}=\beta^{T} \textbf{Q}=\textbf{G}^{-\frac{1}{2}}\textbf{O}\textbf{Q}
\end{split}
\end{equation}
\item Hence the explicit form for components reads:
\begin{equation}
\begin{split}
\alpha_{i}=\sum_{k=1}^{n}\beta^{T}_{ik}\left(A_{k}\cos\omega_{k}t+B_{k}\sin\omega_{k}t\right)
\end{split}
\end{equation}
\item From initial conditions $x(0)=\kappa, \dot{x}(0)=\beta$ we can determine the values of coefficients $A_{k}$ and $B_{k}$:
\begin{equation}
\begin{split}
A_{k}=(\alpha^{(k)},\textbf{G}\kappa) \\
B_{k}=\frac{(\alpha^{(k)},\textbf{G}\beta)}{\omega_{k}}
\end{split}
\end{equation}
\end{enumerate}

Alternative and yet interesting method of finding of normal coordinates may be the application of the Fourier transform to both sides of the equation for eigenvector of both quadratic forms. By that means the equation is transformed into frequency domain, yielding readily separated problem. This can be done on account of first powers of spatial coordinates appearing in the potential energy function. As a result we get a vector equation:

\begin{equation}
\begin{split}
(V-\omega^{2}T)\hat{\textbf{x}}(\omega)=0
\end{split}
\label{eq:fourier}
\end{equation}

where we used integration by parts twice and applied the following trick: before we can act with the Fourier transform our functions must behave 'well' enough, hence we must assume that our system was under an action of some dumping force parametrized by $\epsilon$ so that $\textbf{x}(t)$ vanishes in infinity. After that we take the limit $\epsilon \rightarrow 0$ to retrieve the form of eq. \ref{eq:fourier}.

\section{\label{sec:watson} General theory of nuclear motion}

\subsection{Introduction}

In this section we aim in giving a more detailed insight into the basics of nuclear motion theory. Unlike electronic structure, the nuclear motion calculations require 3 key factors: a) preferably exact representation of the nuclear kinetic energy operator in terms of chosen curvlinear coordinates, b) accurate potential energy surface, c) effective basis set, which allows for evaluation of matrix elements at the lowest possible cost. Having these provided a method of solution of the Schr\"{o}dinger equation is also needed.
Here we treat one aspect of point a). The last section, in turn, catches a bit of point c), serving rigorous theory that underlies commonly used rotational basis sets. Similar treatment of two most popular vibrational basis sets may be found in ref. \cite{Zak_HO} and references therein. Methods of solution of SE are discussed in wide set of textbooks and papers; here we refer to author's unpublished work on discrete variable representation \cite{Zak_DVR}, as one of possible techniques. \\

Normal coordinates are special case of curvlinear coordinates, same as straight line is a special case of a curve. In what follows we introduce several frames of reference that are commonly used by authors. After careful investigation of the procedure for separation of the center of mass we shall focus on a special case of Eckart frame, to yield so called Watson form of the rovibronic Hamiltonian.

\subsection{Rovibronic Coordinates}
This section is intended to provide comprehensive, to a reasonable extent, review of coordinate system choices and their implications on the form of molecular Hamilton operator.
\subsubsection{Laboratory Frame}
Lets consider a system with $N$ point particles. As a model for physical space we shall consider pair $(\Xi,V)$ constituting affine space over three dimensional vector space $V$ with metric $g_{ij}=\delta_{ij}$ associated with set of points $\Xi$ (Euclidean space). At this early stage we don't distinguish between electrons and nuclei, leaving the expressions more general. Such an approach may occur useful when investigating exotic systems like muon atoms (cold fusion)\cite{Jackson1957, Jones1986}, Excitons \cite{Davydov}, etc. The primary, however poorly descriptive choice of coordinate system is space-fixed cartesian frame of reference depicted in fig.\ref{LabFrame} with origin marked as $\varnothing_{0}$.

\begin{figure}[!ht]
\centering
\includegraphics[width=0.7\columnwidth]{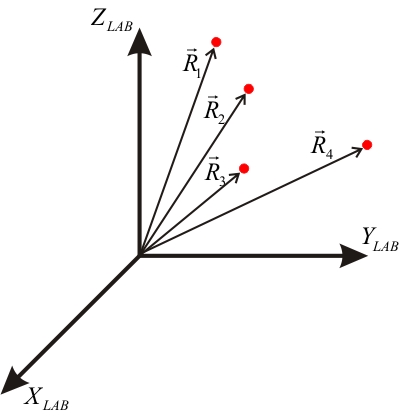}
\caption[An example of a space-fixed laboratory frame of reference for a 4-particle system.]{An example of a space-fixed laboratory frame of reference for 4-particle system.} 
\label{LabFrame}
\end{figure}

Given a set of physical quantities ( mass: $m_{i}$, charge: $C_{i}e$, electron g-factor: $g$, electron spin: $s_{i}$, nuclear g-factor: $g_{\alpha}$, nuclear spin: $I_{\alpha}$, nulcear charge quadrupole moment: $Q^{\alpha}_{ab}$, nuclear polarizability: $\chi^{\alpha}_{ij}$, etc. ) one can try to extract information about a system by solving stationary Schr\"{o}dinger equation. Resulting eigenvalues and eigenfunctions provide complete description of a system in quantum-mechanical sense. \textit{Hilbert space} of such system divides into two orthogonal subspaces for fermions and bosons, but this is yet another story \cite{Emil:QFT}. Skipping less relevant digressions, a general Hamiltonian for a molecular system may be written in the following form:
\begin{equation}
\hat{H}=\hat{T}_{Total}+\hat{V}_{Coulomb}+\hat{H}_{other}
\end{equation}

$\hat{H}_{other}$ stands here for other forms of energy not included in present discussion, e.g. mutual electron magnetic moment interactions, spin-orbit coupling, hyperfine couplings, etc. These will be the subject of another article. As the system consist of $N$ species, the quantum energy operator depend on $3N$ position operators: $(X_{1},Y_{1},Z_{1},...X_{N},Y_{N},Z_{N})$ and $3N$ momentum operators: $(P_{1x},P_{1y},P_{1z},...P_{Nx},P_{Ny},P_{Nz})$. Our aim is to change the coordinates in a way allowing us firstly to perform easier separation of electronic motion from vibrations of nuclei and rotational motion entire molecule, secondly to visualize the correspondence between expressions appearing in Hamiltonian and type of motion and finally to make some physically justified approximations leading to some exactly solvable situations. These solutions may be later utilized as a basis for more accurate calculations.
\paragraph{}
The first step in changing to rovibronic coordinates is separation of the center of mass. We therefore define isometric affine transformation $(\vec{R}_{CM},M_{CM})_{\varnothing_{LAB}}$ \cite{Katsumi} as follows:
\begin{equation} \label{LAB}
M_{CM}:\quad \Re^{3N} \rightarrow \Re^{3N}: \qquad \vec{R}_{i}=\vec{R}_{CM}+\vec{r}_{i}
\end{equation}
where the translation of the origin is given by center of mass vector in laboratory frame of reference:
\begin{equation}
\vec{R}_{CM}:=\frac{1}{M}\sum_{j=1}^{N}m_{j}\vec{R}_{j}
\end{equation}
Hence, the origin of the new coordinates system is chosen (by the above definition) in the center of mass of a molecule, i.e. $\varnothing_{CM}=\varnothing_{LAB}+\vec{R}_{CM}$ (cf. Fig.\ref{ThreeFrames}). This implies:
\begin{equation}
\sum_{j=1}^{N}m_{j}\vec{r}_{j}=0
\label{eq:CM}
\end{equation}
Thus, instead dealing with $N$ position vectors of all particles we can incorporate center of mass vector as one of the coordinates and eliminate the coordinates of one arbitrary particle, since there will always be exactly one vector linearly dependent on others $(cf.~\ref{eq:CM})$. This is formally another affine transformation with no translational part, and linear part equal to identity except first three relations replacing coordinates of the first particle by center of mass coordinates, which are 0 in our system. Following this line of thought, we've passed from laboratory to center of mass coordinates as shown below:

\begin{equation}
\resizebox{\hsize}{!}{$(X_{1},Y_{1},Z_{1},...,X_{N},Y_{N},Z_{N}) \rightarrow (X_{CM},Y_{CM},Z_{CM},x_{2},y_{2},z_{2},...x_{N},y_{N},z_{N})$}
\end{equation}

where we have erased 'particle 1' coordinates. Now we need to check how momentum operators transform under the given coordinate change. These are 1-forms on Hilbert space. In cartesian laboratory frame (denoted as $L$) the kinetic energy operator reads:
\begin{equation}
\hat{T}=-\frac{\hbar^{2}}{2}\sum_{i=1}^{N}\frac{\Delta_{i}^{L}}{m_{i}}
\end{equation}
or in more general form can be regarded as the contraction of a metric tensor with a 2-vector formed from differentiation operators.
Lets transform the first derivatives according to the \textit{chain rule}:
\begin{align}\nonumber
& \frac{\partial}{\partial X_{1}}=\sum_{i=1}^{N}\frac{\partial x_{i}}{\partial X_{1}}\frac{\partial}{\partial x_{i}}=\frac{\partial X_{CM}}{\partial X_{1}}\frac{\partial}{\partial X_{CM}}+ &&\\
& +\sum_{i=2}^{N}\frac{\partial x_{i}}{\partial X_{1}}\frac{\partial}{\partial x_{i}} =\frac{m_{1}}{M}\frac{\partial}{\partial X_{CM}}+\sum_{i=2}^{N}(\delta_{i1}-\frac{m_{1}}{M})\frac{\partial}{\partial x_{i}} &&
\end{align}
and
\begin{equation}
\frac{\partial}{\partial X_{k}}=\frac{m_{k}}{M}\frac{\partial}{\partial X_{0}}+\sum_{i=2}^{N}(\delta_{ik}-\frac{m_{k}}{M})\frac{\partial}{\partial x_{i}}
\end{equation}

Due to the convention excluding first particle's coordinates from the coordinate system, the cartesian kinetic energy operator of the first particle transforms differently from the rest:
\begin{equation}
\small
\frac{\partial^{2}}{\partial X_{1}^{2}}= \left(\frac{m_{1}}{M}\right)^{2} \left(\frac{\partial^{2}}{\partial X_{CM}^{2}}-2\sum_{i=2}^{N}\frac{\partial^{2}}{\partial X_{CM}\partial x_{i}}+\sum_{i,j=2}^{N}\frac{\partial^{2}}{\partial x_{i}\partial x_{j}}\right)
\end{equation}
consequently,
\small{
\begin{align}\nonumber
 & \frac{\partial^{2}}{\partial X_{k}^{2}} = \left(\frac{m_{k}}{M}\right)^{2} \left(\frac{\partial^{2}}{\partial X_{CM}^{2}}-2\sum_{i=2}^{N}\frac{\partial^{2}}{\partial X_{0}\partial x_{i}}+\sum_{i,j=2}^{N}\frac{\partial^{2}}{\partial x_{i}\partial x_{j}}\right)+ &&\\ 
 & +\frac{m_{k}}{M}\left(\frac{\partial ^{2}}{2\partial X_{CM}\partial x_{k}}-2\sum_{i,j=2}^{N}\frac{\partial^{2}}{\partial x_{i}\partial x_{j}}\right)+\frac{\partial^{2}}{\partial x_{k}^{2}} &&
\end{align}}
\normalsize
One could tempt to draw conclusions about couplings between particle's motions at this stage. However, one remarkable feature of many-body systems should be pointed here: although some intermediate 'one-particle' terms may seem complicated, it is often the case, that their collective treatment reveals numerous cancellations, entailing a simple final picture. Therefore one should not rush for interpretation without analyzing the system as a whole.
Accordingly, after substituting derived Laplace operators into kinetic energy terms we get:
\begin{equation}
\frac{1}{m_{1}}\Delta^{L}_{1}=\frac{m_{1}}{M^{2}}\left(\Delta_{CM}-2\sum_{j=2}^{N}\vec{\nabla}_{CM}\cdot\vec{\nabla}_{j}+\sum_{j,i=2}^{N}\vec{\nabla}_{j}\cdot\vec{\nabla}_{i}\right)
\end{equation},
\begin{align}\nonumber
& \frac{1}{m_{k}}\Delta^{L}_{1}=\frac{m_{k}}{M^{2}}\left(\Delta_{CM}-2\sum_{j=2}^{N}\vec{\nabla}_{CM}
\cdot\vec{\nabla}_{j}+\sum_{j,i=2}^{N}\vec{\nabla}_{j}\cdot\vec{\nabla}_{i}\right)+ &&\\
& +\frac{2}{M}\left(\vec{\nabla}_{CM}\cdot\vec{\nabla}_{k}-\sum_{j=2}^{N}\vec{\nabla}_{j}
\cdot\vec{\nabla}_{j}\right)+\frac{1}{m_{i}}\Delta_{k} &&
\end{align}
Taking sums with respect to all particles the total kinetic energy operator reads:
\begin{equation}
\small
\hat{T}=-\frac{\hbar^{2}}{2}\left(\frac{1}{M}\Delta_{CM}+\sum_{i=2}^{N}\frac{\Delta_{i}}{m_{i}}-\frac{1}{M}\sum_{j,i=2}^{N}\vec{\nabla}_{j}\cdot\vec{\nabla}_{i}\right)\equiv \hat{T}_{CM}+\hat{T}_{0}+\hat{T}'
\label{CM}
\end{equation}
where we can identify the kinetic energy of the center of mass $\hat{T}_{CM}$, total kinetic energy of individual, independent particles $\hat{T}_{0}$ and kinetic energy correction term resulting from coupling of correlated particles motion, called often \textit{mass polarization} or cross-terms $\hat{T}'$. This term can be treated as a perturbation to uncoupled Hamiltonian, and in first order of perturbation theory the correction appears as so called \textit{mass polarization parameter}: $K=\sum_{j,i=2}^{N}\left\langle\vec{\nabla}_{j}\cdot\vec{\nabla}_{i}\right\rangle$ \cite{Yamanaka1999}.

Lets notice at this point emerging coupling between the motions of particles. Because a change in position of any of the particles affects the position of center of mass, which we chose as dynamical variable of the system, and because the origin of coordinate system is located in the center of mass, numerical values of other particles positions undergo a change either. It is not a physical effect, rather an artifact of particular coordinate system. It's the cost one must pay in order to separate the center of mass of a system and thereafter to reduce the dimension of the problem. This cost is however in most cases viable due to the fact couplings are linear in the inverse of total mass of the system, which in case of electrons correspond to non-Born-Oppenheimer (NBO) terms. These are usually neglected on account of electron's mass being three orders of magnitude smaller than that for lightest nucleus. In other words, all except first particle's (which is excluded from the system) kinetic energy operators transform trivial when restricting to BO regime, i.e. when neglecting linear terms of electron mass and total (in practice nuclear) mass quotient. This is quite reasonable approximation, as the motion of light electrons very minutely affects the position of center of mass, therefore very weakly couples to motion of other particles.
Culomb potential energy depends only on distances between particles, hence transforms identically to the center of mass coordinate system:

\begin{equation}
\hat{V}=\frac{1}{2}\sum_{i,j=2}^{N}\frac{q_{1}q_{2}}{|\vec{r}_{i}-\vec{r}_{j}|}
\end{equation}
Because the Hamilton operator in \ref{CM} is separable, we can separate the  wavefunction of center of mass as follows:

\begin{equation}
\Phi=\Phi_{CM}(X_{0},Y_{0},Z_{0})\Phi_{int}(x_{2},y_{2},z_{2},...,x_{N},y_{N},z_{N})
\end{equation}
where $\Phi_{int}$ is the wavefunction of internal motion of the system.

\subsubsection{Nuclear center of mass coordinate frame}

\begin{figure}[tb]
\centering
\includegraphics[width=0.8\columnwidth]{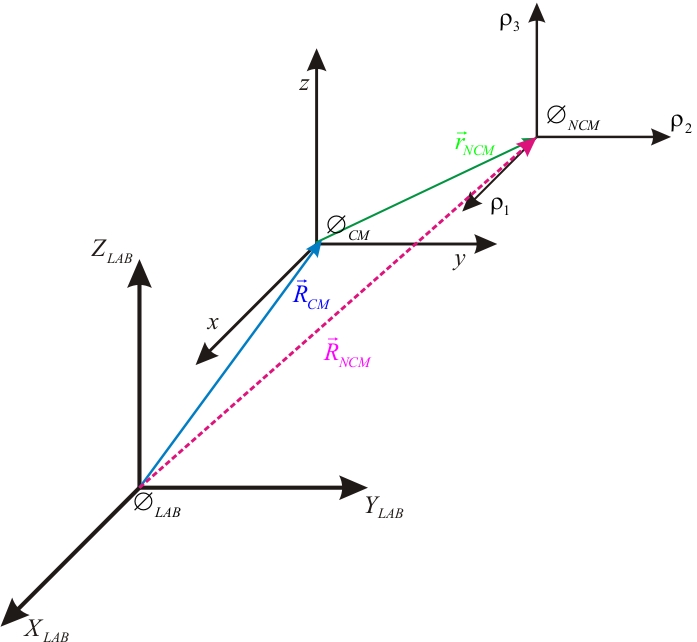}
\caption[Three space-fixed frames of reference: Laboratory, total center of mass and nuclear center of mass]{Three space-fixed frames of reference: Laboratory, total center of mass and nuclear center of mass.} 
\label{fig:ThreeFrames}
\end{figure}
 It will be very useful to begin our discussion with the following theorem:

\begin{theorem}[Center of mass coupling]
If $S$ denotes the 	entire system and $A$ is a system of particles that constitute a subset of $S$, then whenever transforming via affine map $M_{A}$ from a space-fixed cartesian coordinate system with a frozen origin $\emptyset_{0}$ to a cartesian coordinate system with origin $\emptyset_{1}$ located in the center of mass of subset $A$, the corresponding kinetic energy operator of the entire system $S$ transforms into separable form of simple sum of two kinetic energy operators $\hat{T}_{A}\oplus\hat{T}_{S/A}$ written in explicit as follows:
\begin{equation}
\hat{T}_{A}=-\frac{\hbar^{2}}{2}\sum_{i\in A}\frac{\Delta_{i}}{m_{i}}+\frac{\hbar^{2}}{2M_{A}}\sum_{j,i \in A}\vec{\nabla}_{j}\cdot\vec{\nabla}_{i}
\end{equation}
\begin{equation}
\hat{T}_{S/A}=-\frac{\hbar^{2}}{2}\sum_{i\in S/A}\frac{\Delta_{i}}{m_{i}}-\frac{\hbar^{2}}{2M_{A}}\sum_{j,i \in S/A}\vec{\nabla}_{j}\cdot\vec{\nabla}_{i}
\end{equation}
where differentiation is taken with respect to coordinates in $\emptyset_{1}$ orthogonal coordinate system.
\end{theorem}

\begin{proof}
Linear part of the affine transformation relates coordinates of considered systems as:
\begin{equation}
\vec{\rho}_{i}=\vec{r}_{i}-\frac{1}{M_{A}}\sum_{j\in A}m_{j}\vec{r}_{j}
\end{equation}
where $i\in A$.
Using the chain rule we can write the expression for the first derivative:
\begin{equation}
\frac{\partial}{\partial x_{i}}=\sum_{j\in S}\sum_{s=1}^{3}\frac{\partial \rho_{js}}{\partial x_{i}}\frac{\partial}{\partial \rho_{js}}
\end{equation}

Please note that the transformation does not mix different components (can be decomposed into three transformations of orthogonal subspaces), hence summation over $s$ vanishes immediately:
\begin{equation}
\frac{\partial}{\partial x_{i}}=\sum_{j\in S}\left(\delta_{ij}-\frac{m_{i}}{M_{A}}\delta_{iA}\right)\frac{\partial}{\partial \rho_{js}}
\end{equation}
The symbol $\delta_{iA}$ is defined by:
\begin{equation}
\delta_{iA}=\left\{ \begin{array}{ccc}
1\: when\: i\in A \\
0 \:when\: i \notin A\\
 \end{array}\equiv\sum_{l\in A}\delta_{il} \right|
\end{equation}
Consequently,
\begin{equation}
\small
\begin{split}
\frac{\partial^{2}}{\partial x_{i}^{2}}=\sum_{j,k\in S}\left(\delta_{ij}-\frac{m_{i}}{M_{A}}\delta_{iA}\right)\left(\delta_{ik}-\frac{m_{i}}{M_{A}}\delta_{iA}\right)\frac{\partial^{2}}{\partial \rho_{j1}\partial \rho_{k1}}=\\
= \frac{\partial^{2}}{\partial \rho_{i1}^{2}}-2\frac{m_{i}}{M_{A}}\delta_{iA}\sum_{j\in S}\frac{\partial^{2}}{\partial \rho_{i1}\partial \rho_{j1}}+\left(\frac{m_{i}}{M_{A}}\right)^{2}\delta_{iA}\sum_{j,k\in S}\frac{\partial^{2}}{\partial \rho_{j1}\partial \rho_{k1}}
\end{split}
\end{equation}
\normalsize
where $\rho_{k1}$ stands for $1$-st component of pointing vector of particle $k$ in $A$ center of mass system.
Mixed terms can be evaluated in similar way to yield:
\begin{align*}
\frac{\partial^{2}}{\partial x_{i} \partial x_{j}}=\frac{\partial^{2}}{\partial \rho_{i1} \partial \rho_{j1}}-\frac{m_{i}}{M_{A}}\delta_{iA}\sum_{k\in S}\frac{\partial^{2}}{\partial \rho_{j1}\partial \rho_{k1}}-& \\ -\frac{m_{j}}{M_{A}}\delta_{jA}\sum_{l\in S}\frac{\partial^{2}}{\partial \rho_{i1}\partial \rho_{l1}}
+\frac{m_{i}m_{j}}{M_{A}^{2}}\delta_{iA}\delta_{jA}\sum_{k,l\in S}\frac{\partial^{2}}{\partial \rho_{k1}\partial \rho_{l1}}
\end{align*}

After extending above formulas to three dimensions, the intermediate result appear as:
\begin{equation}
\small
\begin{split}
\sum_{i\in S}\frac{\Delta'_{i}}{m_{i}}=\left(\sum_{i\in S}\frac{\Delta_{i}}{m_{i}}-2\frac{1}{M_{A}}\sum_{i\in A}\sum_{j\in S}\vec{\nabla}_{i}\cdot\vec{\nabla}_{j}+
\frac{1}{M_{A}}\sum_{i,j\in S}\vec{\nabla}_{i}\cdot \vec{\nabla}_{j}\right)
\end{split}
\end{equation}
\normalsize
And unlike in the approach of ref.\cite{Bunker}, we obtain simple cross terms transformation:
\begin{equation}
\sum_{i,j\in S}\vec{\nabla'}_{i}\cdot\vec{\nabla'}_{j}=0
\end{equation}
Finally the total kinetic energy operator in $\emptyset_{1}$ and the new coordinates reads :
\begin{equation}
\hat{T}_{\emptyset_{1}}=-\frac{\hbar^{2}}{2}\sum_{i\in S}\frac{\Delta_{i}}{m_{i}}+\hbar^{2}\frac{1}{M_{A}}\sum_{i\in A}\sum_{j\in S}\vec{\nabla}_{i}\cdot\vec{\nabla}_{j}-
\frac{\hbar^{2}}{2M_{A}}\sum_{i,j\in S}\vec{\nabla}_{i}\cdot \vec{\nabla}_{j}
\end{equation}
By rearranging summations one obtain separable operator:
\begin{align}\nonumber
\hat{T}_{\emptyset_{1}}=&-\frac{\hbar^{2}}{2}\sum_{i=1}^{A}\frac{\Delta_{i}}{m_{i}}-\frac{\hbar^{2}}{2}\sum_{i\in S/A}\frac{\Delta_{i}}{m_{i}}-\frac{\hbar^{2}}{M_{A}}\sum_{i,j\in S/A}\vec{\nabla}_{i}\cdot\vec{\nabla}_{j}+\\
&+\frac{\hbar^{2}}{2M_{A}}\sum_{i,j\in A}\vec{\nabla}_{i}\cdot \vec{\nabla}_{j}\equiv \hat{T}_{A}+\hat{T}_{S/A}
\label{T_{S}}
\end{align}
\end{proof}
Now, on the base of relation \ref{CM} we could claim that it is possible to pass from space-fixed frozen-origin laboratory frame to $A$ frame, on account of transitive relation between these frames. To prove that and to reveal some formalism we write that:
\begin{equation}
\left(\vec{R}_{CM},M_{CM}\right)_{\emptyset_{0}}
\end{equation}
We consider three transformations:
\begin{enumerate}
  \item From the laboratory frame to the total center of mass frame, with associated affine transformation: $\left(\vec{R}_{CM},M_{CM}\right)_{\emptyset_{0}}$
  \item From the total center of mass frame to the nuclear center of mass frame: $\left(\vec{r}_{NCM},M_{NCM}\right)_{\emptyset_{1}}$
  \item From the laboratory frame to the nuclear center of mass frame: $\left(\vec{R}_{NCM},M_{LNCM}\right)_{\emptyset_{0}}$
\end{enumerate}
If we are be able to show that combining the first two transformations yields in the third function, the proof is be complete. In other words we ask if:
\begin{equation}
\small
\left(\vec{r}_{NCM},M_{NCM}\right)_{\emptyset_{1}}\circ\left(\vec{R}_{CM},M_{CM}\right)_{\emptyset_{0}}=\left(\vec{R}_{NCM},M_{LNCM}\right)_{\emptyset_{0}}
\end{equation}
holds. Utilizing the properties of affine transformation: $ \left(x,A\right)_{\emptyset}\circ\left(y,B\right)_{\emptyset}=\left(x+Ay,AB\right)_{\emptyset}$ and
$\left(x,A\right)_{\emptyset}=\left(x+(A-id)a,A\right)_{\emptyset+a}$ we can write that:
 \small
\begin{flalign}
 & \left(\vec{r}_{NCM},  M_{NCM}\right)_{\emptyset_{1}}  \circ\left(\vec{R}_{CM},M_{CM}\right)_{\emptyset_{0}}= &&\\\nonumber
 & =\left(\vec{r}_{NCM}-(M_{NCM}-id)\vec{R}_{CM},M_{NCM}\right)_{\emptyset_{0}}\circ\left(\vec{R}_{CM},M_{CM}\right)_{\emptyset_{0}}&&\\\nonumber
 & =\left(\vec{r}_{NCM}+\vec{R}_{CM},M_{NCM}M_{CM}\right)_{\emptyset_{0}}=\left(\vec{R}_{NCM},M_{LNCM}\right)_{\emptyset_{0}} &&
\end{flalign}
\normalsize
 In the present case the subsystem $A$ stands for all nuclei of the system (molecule), hence $S/A$ gathers all electrons.
Lets label by $1,2,...,N_{nuc}$ nuclei of the system and by $N_{nuc}+1,N_{nuc}+2,...,N_{nuc}+N_{el}$, where $N_{nuc}+N_{el}=N$. Linear affine map into the new coordinates system is defined as:
\begin{equation}
M_{NCM}:\quad \Re^{3N-3} \rightarrow \Re^{3N-3}: \qquad \vec{r}_{i}=\vec{r}_{NCM}+\vec{\rho}_{i}
\end{equation}
where:
\begin{equation}
\vec{r}_{NCM}:=\frac{1}{M_{n}}\sum_{j=1}^{N_{nuc}}m_{j}\vec{r}_{j}
\end{equation}
is nuclear center of mass in space fixed total center of mass coordinates system. Therefore explicit transformation may be written as:
\begin{equation}
\vec{\rho}_{i}=\vec{r}_{i}+\frac{m_{e}}{M_{n}}\sum_{j=N_{nuc}+1}^{N}\vec{r}_{j}
\end{equation}
The above theorem guarantees that we can make straightforward transformation from a space-fixed coordinate system into a nuclear center of mass coordinate system obtaining the kinetic energy operator expressed as in eq.\ref{T_{S}}.

\subsubsection{Euler Angles}
\begin{figure}[tb]
\centering
\includegraphics[width=0.8\columnwidth]{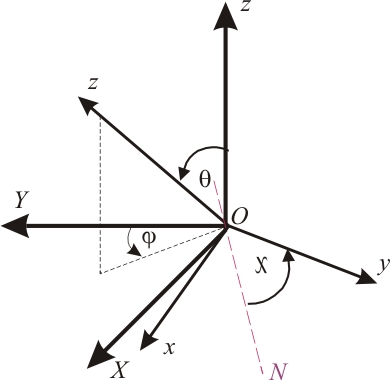}
\caption[Euler Angles in zyz' convention]{Euler angles in zyz' convention} 
\label{Euler}
\end{figure}
Euler angles define the position of the rotated frame of reference with respect to the space fixed frame. There are several conventions yielding equivalent results, but for our purposes so called \textit{zyz'} convention will be most suitable \cite{Wilson}. It means that in order to transform one frame into the other we apply in proper order: $\phi$ angle rotation around $z$ axis, then $\theta$ angle around $y$ axis, and finally $\chi$ angle around $z$ axis of already partially transformed frame. It is also convenient to introduce \textit{node line} which represents positive sense of rotation from $OZ$ to $Oz$ and lays in the intersection of $xy$ and $XY$ planes. Full coverage of space is achieved when $\phi,\chi \in [0,2\pi]$ and $\theta \in [0,\pi]$.
Explicit transformation relations for coordinates of a vector in rotating frame can be obtained by multiplying appropriate rotation matrices \cite{shankar2010}. 
\begin{align}
 & R_{x}(\alpha)=\left(\begin{array}{ccc}
1 & 0 & 0 \\
0 & \cos\alpha & -\sin\alpha\\
0 & \sin\alpha & \cos\alpha \end{array} \right) \qquad R_{y}(\alpha)=\left(\begin{array}{ccc}
\cos\alpha& 0 & \sin\alpha \\
0 & 1 & 0\\
-\sin\alpha & 0  & \cos\alpha \end{array} \right) && \\ \nonumber
 & R_{z}(\alpha)=\left(\begin{array}{ccc}
\cos\alpha & -\sin\alpha & 0\\
 \sin\alpha & \cos\alpha & 0 \\
 0 & 0 & 1
 \end{array} \right) &&
 \label{rotation}
\end{align} 
Therefore in $zyz'$ scheme the general rotation matrix in physical space reads:
\begin{widetext}
\begin{equation}
R_{zyz'}(\phi,\chi,\theta)=R_{z'}(\theta)R_{y}(\chi)R_{z}(\phi)=\left(\begin{array}{ccc}
\cos\theta\cos\phi\cos\chi-\sin\phi\sin\chi & \cos\theta\sin\phi\cos\chi+\cos\phi\sin\chi & -\sin\theta\cos\chi \\
-\cos\theta\cos\phi\sin\chi-\sin\phi\cos\chi & -\cos\theta\sin\phi\sin\chi+\cos\phi\cos\chi & \sin\theta\sin\chi\\
\sin\theta\cos\phi & \sin\theta\sin\phi & \cos\theta \end{array} \right)
\label{fig:euler}
\end{equation} 
\end{widetext}
\normalsize
and this acts on the column representation of a vector in the space fixed frame, yielding the column representation of a vector in the rotating frame:
 \begin{equation}\label{euler1}
\left(\begin{array}{ccc}
x\\
y\\
z \end{array} \right)=R_{zyz'}(\phi,\chi,\theta)\left(\begin{array}{ccc}
X\\
Y\\
Z \end{array} \right)
\end{equation} 
Of course $R_{zyz'}(\phi,\chi,\theta) \in SO(3)$, hence:
\small
\begin{equation}
R_{zyz'}(\phi,\chi,\theta)R_{zyz'}(\phi,\chi,\theta)^{T}=R_{zyz'}(\phi,\chi,\theta)^{T}R_{zyz'}(\phi,\chi,\theta)=1
\end{equation} 
\normalsize
Rotations in \textit{Hilbert space} associated with rotations in physical space are given by unitary operator: $U\left(R^{\vec{n}}_{\phi}\right)=\exp\left(-\frac{i}{\hbar}\vec{J}\cdot \vec{n}\phi\right)$ as for rotation about $\phi$ around $\vec{n}$ vector.
\subsubsection{Molecule-fixed frame}
\paragraph{Classical approach}

Main advantage of the body-fixed frames of reference together with rovibronic coordinates is intuitive form of associated Hamilton operator of a system, allowing to apply approximations leading to separation of rotational and vibrational motion, etc. Terms responsible for particular physical effects may be then identified, providing convenient way of controlling approximation procedures. They also have considerable computational advantages such as use of $3-j$ symbols in angular potential integrations instead of $6-j$ in space-fixed case \cite{Tennyson1982}.
Lets introduce the body-fixed frame, determine rotational coordinates. Mutual orientation of body-fixed and space fixed frames is commonly given by set of three \textit{Euler angles} (see fig.\ref{fig:euler}). Origins of both coordinate systems are located in molecular center of mass. Previous assumptions lead to the conclusion that three rotating axes must be embedded to the molecular system according to some arbitrary rule. So far we've got a set of $3N-3$ coordinates in molecular center of mass cartesian coordinate system: $(x_{2},y_{2},z_{2},...,x_{N},y_{N},z_{N})$. We want a transformation into rotating frame of reference described by set of new coordinates: $(\theta,\phi,\chi,Q_{1},...,Q_{3N-6})$. These are called \textit{rovibronic coordinates} \cite{Bunker}. $Q_{i}$ stand for functions of cartesian coordinates. Herein, these functions are linear giving rise to normal coordinates. However one may choose any curvlinear internal coordinates (not discussed here). Every molecule has its equilibrium geometry defined as a set of coordinates values minimizing globally internal potential energy of the molecule. Lets denote equilibrium position vectors of all particles in a system as $\vec{a}_{\alpha}$ defining displacement from equilibrium position vectors as:
\begin{equation} \label{displacement}
\vec{\rho}_{\alpha}=\vec{r}_{\alpha}-\vec{a}_{\alpha}
\end{equation}
According to fig.\ref{fig:ThreeFrames} we can write total velocity of $\alpha-th$ particle of a system in laboratory cartesian coordinates system.
\begin{equation} \label{velocity}
\vec{V}_{\alpha}=\dot{\vec{R}}_{\alpha}+\vec{\omega}\times\vec{r}_{\alpha}+\vec{v}_{\alpha}
\end{equation}
The components of above equation correspond to center of mass velocity $\dot{\vec{R}}_{\alpha}$, velocity due to rotation of body-fixed frame $\vec{\omega}\times\vec{r}_{\alpha}$ and motion of particle in the rotating coordinate system $\vec{v}_{\alpha}$ , respectively. Now lets write a classical formula for the kinetic energy of our system of $N$ particles in LAB frame:

\begin{align} \nonumber
 & 2T=\sum_{\alpha=1}^{N}m_{\alpha}\vec{V}_{\alpha}^{2}=M\dot{\vec{R}}_{\alpha}^{2}+\sum_{\alpha=1}^{N}m_{\alpha}\left(\vec{\omega}\times\vec{r}_{\alpha}\right)\cdot
\left(\vec{\omega}\times\vec{r}_{\alpha}\right)+ && \\\nonumber
 & +\sum_{\alpha=1}^{N}m_{\alpha}\vec{v}_{\alpha}^{2}+\sum_{\alpha=1}^{N}2m_{\alpha}\dot{\vec{R}}_{\alpha}+\left(\vec{\omega}\times\vec{r}_{\alpha}\right)+ && \\
 & +2\dot{\vec{R}}_{\alpha}\sum_{\alpha=1}^{N}m_{\alpha}\vec{v}_{\alpha}^{2}+2\vec{\omega}\cdot\sum_{\alpha=1}^{N}m_{\alpha}\left(\vec{r}_{\alpha}\times\vec{v}_{\alpha}\right) &&
\label{kinetic1}
\end{align}
where for the last term we've used cyclic invariance of mixed vector product. After separating centre of mass motion the expression simplifies a bit - this is equivalent to passing into center of mass frame of reference. This yield in three relations confining the coordinates of all particles of a system:

\begin{equation} \label{CM=0}
\sum_{\alpha=1}^{N}m_{\alpha}\vec{r}_{\alpha}=0
\end{equation}
which is followed by (see \ref{velocity})
\begin{equation} \label{CMV=0}
\sum_{\alpha=1}^{N}m_{\alpha}\vec{v}_{\alpha}=0
\end{equation}
The only thing needed to be specified is the embedding of the rotating frame, which will provide another three relations. The natural choice and historical one is the condition of vanishing angular momentum of a collection of all particles in rotating frame:
\begin{equation} \label{J=0}
\vec{J}=\sum_{\alpha=1}^{N}m_{\alpha}\vec{r}\times\dot{\vec{r}}_{\alpha}=0
\end{equation}
Now if $\vec{\rho}_{\alpha}$ is small for all $\alpha$'s then we can within a good approximation write $\vec{r}_{\alpha}\approx\vec{a}_{\alpha}$ and

\begin{equation} \label{J=0approx}
\vec{J}\approx\sum_{\alpha=1}^{N}m_{\alpha}\vec{a}\times\dot{\vec{r}}_{\alpha}=0
\end{equation}
This approximation hold well for most of fairly rigid molecular systems, as the amplitudes of vibrations of semi-rigid molecules generally don't exceed $5\%$ of bond length \cite{Wilson}. For more general point of view see, for example later paper by \textit{Schmiedt, et al.}\cite{Schmiedt2015}. Keeping this intuition \textit{C. Eckart} \cite{Eckart1935} postulated following general conditions:
\begin{equation} \label{eq:Eckart}
\sum_{\alpha=1}^{N}m_{\alpha}\vec{a}\times\dot{\vec{\rho}}_{\alpha}=0
\end{equation}
whereas eq.\ref{J=0approx} stating approximate angular momentum conservation law occurs as it's direct consequence. The specification of the instantaneous position of the moving axes requires six numbers, which may be taken to be the three center of mass coordinates and the Eulerian
angles \cite{Eckart1935}. Hence we should have only $3N-6$ independent internal coordinates $Q_{i}$. Utilizing these conditions we may write the kinetic energy in a new form

\begin{align}\nonumber
 & 2T=\sum_{\alpha=1}^{N}m_{\alpha}\left(\vec{\omega}\times\vec{r}_{\alpha}\right)\cdot
\left(\vec{\omega}\times\vec{r}_{\alpha}\right) &&\\\nonumber
 & +\sum_{\alpha=1}^{N}m_{\alpha}\vec{v}_{\alpha}^{2}+ &&\\\nonumber
 & +2\vec{\omega}\cdot\sum_{\alpha=1}^{N}m_{\alpha}\left(\vec{\rho}_{\alpha}\times\vec{v}_{\alpha}\right)&&
\label{kinetic2}
\end{align} 

where we have rotational energy, vibrational kinetic energy and \textit{Coriolos} coupling energy respectively. This \textit{Eckart}-derived form allows to separate rotations from vibrations with the least cost, letting for most effective perturbation treatment of the Coriolis coupling. Note that Eckart equations have to be solved for euler angles.
After introducing \textit{moment of inertia} tensor $I_{ij}=\sum_{\alpha=1}^{N}m_{\alpha}\left(\vec{r}\cdot\vec{r}I-\vec{r}_{\alpha}\otimes\vec{r}_{\alpha}\right)$; and some algebraic manipulations the kinetic energy reads

\begin{align}\nonumber
& 2T=I_{xx}\omega_{x}^{2}+I_{yy}\omega_{y}^{2}+I_{zz}\omega_{z}^{2}-2I_{xy}\omega_{x}\omega_{y}-2I_{xz}\omega_{x}\omega_{z}-&& \\
& -2I_{yz}\omega_{y}\omega_{z}+\sum_{\alpha=1}^{N}m_{\alpha}\vec{v}_{\alpha}^{2}+
+2\vec{\omega}\cdot\sum_{\alpha=1}^{N}m_{\alpha}\left(\vec{\rho}_{\alpha}\times\vec{v}_{\alpha}\right)&&
\end{align}
Deriving the above relation would be a good simple exercise, left for the reader. One of the possible ways of dealing with internal motion of molecule (relative motion of nuclei) is normal coordinates approach. They are defined by linear relation to cartesian coordinates:
\begin{equation} \label{normaldef}
\rho_{\alpha i}=\sum_{k=1}^{3N-6}\eta^{\alpha}_{ik}Q_{k}
\end{equation}
for $i=1,2,3$ - index of displacement vector components, $\eta^{\alpha}_{ik}$ stands for linear transformation matrix ($3\times N\times 3N-6$) diagonalizing simultaneously kinetic energy and quadratic terms in potential energy of a system ; $Q_{k}$ is $k$-th normal coordinate. Of course $\alpha$ enumerates particles of the system, in practice those are the nuclei. We don't transfer the notation from section \ref{sec:level2}, to save the consistency with original derivations.
Within the normal coordinates framework the Coriolis term reads
\small
\begin{align}\nonumber
 & \vec{\omega}\cdot\sum_{\alpha=1}^{N}m_{\alpha}\left(\vec{\rho}_{\alpha}\times\vec{v}_{\alpha}\right)=\sum_{i,j,k=1}^{3}\sum_{\alpha=1}^{N}\sum_{s,t=1}^{3N-6}\omega^{i}m_{\alpha}\epsilon_{ijk}\eta^{\alpha}_{js}\eta^{\alpha}_{kt}Q_{s}
\dot{Q}_{t}=&&\\
 & =\sum_{i=1}^{3}\sum_{s,t=1}^{3N-6}\omega^{i}\tau_{is}\dot{Q}_{t} &&
\label{normaldef}
\end{align}
 \normalsize
where we introduced \textit{Levi-Civita totally antisymmetric tensor} $\epsilon_{ijk}$ and defined $\tau$ matrix as follows:
\begin{equation}\label{tau}
\tau_{is}:=\sum_{j,k=1}^{3}\sum_{\alpha=1}^{N}m_{\alpha}\epsilon_{ijk}\eta^{\alpha}_{js}\eta^{\alpha}_{kt}Q_{s}
\end{equation}
Here for the first time so called \textit{Coriolis coupling constants} $\xi$ appeared,  defined by the relation
\begin{equation}\label{coriolisconst}
\xi_{is}:=\sum_{j,k=1}^{3}\sum_{\alpha=1}^{N}\epsilon_{ijk}\eta^{\alpha}_{js}\eta^{\alpha}_{kt}
\end{equation}
The kinetic energy transforms now to the form

\begin{align}\nonumber
& 2T=I_{xx}\omega_{x}^{2}+I_{yy}\omega_{y}^{2}+I_{zz}\omega_{z}^{2}-2I_{xy}\omega_{x}\omega_{y}-2I_{xz}\omega_{x}\omega_{z}-&&\\
& -2I_{yz}\omega_{y}\omega_{z}+\sum_{s,t=1}^{3N-6}\sum_{i=1}^{3}\omega^{i}\tau_{is}\dot{Q}_{t}+\sum_{k=1}^{3N-6}\dot{Q}_{t}^{2} &&
\label{kinetic4}
\end{align}

Our aim is to make transition into quantum-mechanical expression for the kinetic energy, thus Hamilton picture is required, which, in turn involves generalized momenta and coordinates. Historical and perhaps more intuitive route utilizes angular momentum representation, what \textit{de facto} makes somewhat around way into the quantum mechanical formalism.  Lets start from the definition of angular momentum:
\begin{equation} \label{Angular}
\begin{split}
\vec{J}:=\sum_{\alpha=1}^{N}m_{\alpha}\vec{r}\times\dot{\vec{r}}_{\alpha}
\end{split}
\end{equation}
Introducing  angular velocity, moment of inertia and normal coordinates we get
\begin{align}\nonumber
 & \vec{J}:=\sum_{\alpha=1}^{N}m_{\alpha}\vec{r}\times\left(\vec{\omega}\times\vec{r}_{\alpha}\right)+
\sum_{\alpha=1}^{N}m_{\alpha}\vec{r}_{\alpha}\times\vec{v}_{\alpha}=&& \\\nonumber
 & 	=\sum_{\alpha=1}^{N}m_{\alpha}\left(r_{\alpha}^{2}\vec{\omega}-\left(\vec{\omega}\cdot\vec{r}_{\alpha}\right)\vec{r}_{\alpha}\right)+
\sum_{\alpha=1}^{N}m_{\alpha}\left(\vec{r}_{\alpha}-\vec{a}_{\alpha}\right)\times\vec{v}_{\alpha}=&&\\\nonumber
 & =\hat{x}\left(I_{xx}\omega_{x}-I_{xy}\omega_{y}-I_{xz}\omega_{z}\right)+\hat{x}\left(I_{yx}\omega_{x}-I_{yy}\omega_{y}-I_{yz}\omega_{z}\right)+&&\\
 &+\hat{z}\left(I_{zx}\omega_{x}-I_{zy}\omega_{y}-I_{zz}\omega_{z}\right)+\sum_{s,t=1}^{3N-6}\vec{\tau}_{s}\dot{Q}_{t}&&
\label{Angularexpansion}
\end{align}
It is desired to relate angular and generalized momenta, therefore we write the momentum conjugated to $Q_{s}$ as
\begin{equation}\label{momentum}
P_{s}=\frac{\partial T}{\partial \dot{Q}_{s}}=\dot{Q}_{s}+\sum_{i=1}^{3}\tau^{i}_{s}\omega_{i}
\end{equation}
now we can construct \textit{vibrational angular momentum}.

Its presence can be easily observed when considering vibrations of acetylene, where some motions of atoms that destroy linearity may contribute to some internal angular momentum of the molecule. Another example may be quasi-free rotations of methyl groups in hydrocarbons. These motions are in fact oscillations, which 'look' like rotations, carrying also some internal angular momentum.
\begin{equation}\label{vibrangular}
j_{i}=\sum_{k=1}^{3N-6}\tau^{i}_{k}P_{k}=\sum_{k=1}^{3N-6}\tau^{i}_{k}\dot{Q}_{k}+\sum_{k=1}^{3N-6}\tau^{i}_{k}\left(\vec{\tau}\cdot\vec{\omega}\right)
\end{equation}
It is clear that
\begin{equation}\label{kineticc}
2T=\vec{J}\cdot\vec{\omega}+\sum_{k=1}^{3N-6}P_{k}\dot{Q}_{k}
\end{equation}
and
\begin{equation}\label{momentum}
\dot{Q}_{k}=P_{k}-\vec{\tau}_{s}\cdot\vec{\omega}
\end{equation}
hence,
\begin{equation}\label{kinetic5}
2T=\left(\vec{J}-\vec{j}\right)\cdot\vec{\omega}+\sum_{k=1}^{3N-6}P_{k}^{2}
\end{equation}
The only task now is to get rid of angular velocity. Making use of relation between angular momentum and angular velocity
\begin{equation}\label{angularmomentumvel}
\begin{split}
J_{x}=I_{xx}\omega_{x}-I_{xy}\omega_{y}-I_{xz}\omega_{z}+\sum_{k=1}^{3N-6}\tau^{x}_{k}\dot{Q}_{k}\\
J_{y}=-I_{yx}\omega_{x}+I_{yy}\omega_{y}-I_{yz}\omega_{z}+\sum_{k=1}^{3N-6}\tau^{y}_{k}\dot{Q}_{k}\\
J_{z}=-I_{zx}\omega_{x}-I_{zy}\omega_{y}+I_{zz}\omega_{z}+\sum_{k=1}^{3N-6}\tau^{z}_{k}\dot{Q}_{k}
\end{split}
\end{equation}

what follows from \ref{vibrangular}
\begin{widetext}
\begin{equation}
\begin{split}
J_{x}-j_{x}=\left(I_{xx}-\sum_{k=1}^{3N-6}(\tau^{x}_{k})^{2}\right)\omega_{x}-\left(I_{xy}+\sum_{k=1}^{3N-6}\tau^{x}_{k}\tau^{y}_{k}\right)\omega_{y}-\left(I_{xz}+\sum_{k=1}^{3N-6}\tau^{x}_{k}\tau^{z}_{k}\right)\omega_{z}\\
J_{y}-j_{y}=-\left(I_{yx}+\sum_{k=1}^{3N-6}\tau^{y}_{k}\tau^{x}_{k}\right)\omega_{x}+\left(I_{yy}-\sum_{k=1}^{3N-6}(\tau^{y}_{k})^{2}\right)\omega_{y}-\left(I_{yz}+\sum_{k=1}^{3N-6}\tau^{y}_{k}\tau^{z}_{k}\right)\omega_{z}\\
J_{z}-j_{z}=-\left(I_{zx}+\sum_{k=1}^{3N-6}\tau^{z}_{k}\tau^{x}_{k}\right)\omega_{x}-\left(I_{zy}+\sum_{k=1}^{3N-6}\tau^{z}_{k}\tau^{y}_{k}\right)\omega_{y}+\left(I_{zz}-\sum_{k=1}^{3N-6}(\tau^{z}_{k})^{2}\right)\omega_{z}
\end{split}
\label{angularmomentumvel}
\end{equation}
\end{widetext}

By inverting the above relations to express angular velocities on left hand side of equations lets rewrite this system in a matrix form. Formally we intoduce
\begin{equation}\label{rovibAngMom}
\vec{M}=\vec{J}-\vec{j}
\end{equation}
and
\begin{equation}\label{omegavec}
\vec{\omega}= \left(\begin{array}{ccc}
\omega_{x} \\
\omega_{y} \\
\omega_{z} \end{array} \right)
\end{equation}
relation \ref{angularmomentumvel} may be formulated in compact form
\begin{equation}\label{rovibAngMomMat}
\vec{M}=I\vec{\omega}
\end{equation}
from which it's straightforward to obtain angular velocities by inverting real symmetric matrix $I$. Then kinetic energy takes more compact form, being a good starting point for the quantum Hamiltonian:
\begin{equation}\label{kinetic6}
T=\frac{1}{2}\vec{M}^{T}\mu\vec{M}+\frac{1}{2}\sum_{k=1}^{3N-6}P_{k}^{2}
\end{equation}
where $\vec{M}:=\vec{J}-\vec{j}$ and $\mu$ is the inverse matrix of $I$. This exact expression plays a crucial role in classical mechanics of non-rigid rotating bodies \cite{Wilson}.

\subsubsection{Podolsky Trick}

When attempting to derive quantum-mechanical form of the Hamilton operator for the general system discussed in previous paragraph, one might try to apply \textit{Jordan} rules straightforwardly to eq. \ref{kinetic6}. Unfortunately such approach will result in wrong answer, mainly because we've made transformation from cartesian to rovibronic(curvlinear) coordinates: $\left(x_{1},y_{1},z_{1},...,x_{N},y_{N},z_{N}\right)= \overrightarrow{\xi}\left(q_{1},...,q_{3N}\right)$ where $\xi$ is composed of transformation into molecular center of mass system and transformation molecule-fixed rotating frame of reference by use of Euler angles \ref{fig:ThreeFrames}. Both transformations are affine transformations. The first one is simply $(\vec{R}_{CM},M_{CM})_{\varnothing_{LAB}}$ and the latter has identity as translation part, while linear part is a combination of three Euler rotations. Classical expression for total energy in Lagrange formalism may be simply written as \cite{Islampour1983}:
\begin{equation}\label{clas_Lagrange}
\resizebox{.95\hsize}{!}{$E(\vec{r}_{1},...,\vec{r}_{N},\dot{\vec{r}}_{1},...,\dot{\vec{r}}_{N})=\frac{1}{2}\sum_{ij}\sqrt{m_{i}m_{j}}\delta_{ij}\dot{\vec{r}}_{i}\cdot\dot{\vec{r}}_{j}+V(\vec{r}_{1},...,\vec{r}_{N})$}
\end{equation}
and corresponding Hamilton form:
\begin{equation}\label{clas_hamilton}
\resizebox{.95\hsize}{!}{$H(\vec{r}_{1},...,\vec{r}_{N},\vec{p}_{1},...,\vec{p}_{N})=\frac{1}{2}\sum_{ij}\frac{\delta_{ij}\vec{p}_{i}\cdot\vec{p}_{j}}{\sqrt{m_{i}m_{j}}}+V(\vec{r}_{1},...,\vec{r}_{N})$}
\end{equation}
Transformation into generalized coordinates affects the expression for Lagrange-form energy in following way:
\begin{equation}\label{clas_Lagrange}
\resizebox{.95\hsize}{!}{$E(q_{1},...,q_{3N},\dot{q}_{1},...,\dot{q}_{3N})=\frac{1}{2}\sum_{ij}g_{ij}\dot{q}_{i}\dot{q}_{j}+V(q_{1},...,q_{3N})$}
\end{equation}
where use of chain rule for coordinate change results in form of quadratic form metric tensor:
\begin{equation}\label{metric-tensor}
g_{ij}=\sum_{n,\alpha}m_{n}\left(\frac{\partial r_{n\alpha}}{\partial q_{i}}\right)\left(\frac{\partial r_{n\alpha}}{\partial q_{j}}\right)
\end{equation}
We can use the definition of generalized momentum $P_{i}=\frac{\partial(H-V)}{\partial \dot{q}}$ to transform above relations into Hamilton form:
\begin{equation}\label{clas_hamilton}
H(q_{1},...,q_{3N},P_{1},...,P_{N})=\frac{1}{2}\sum_{ij=1}^{3N}g^{ij}P_{i}P_{j}+V(q_{1},...,q_{3N})
\end{equation}
where $g^{ij}$ matrix is inverse of $g_{ij}$ (we've simply raised two indices in metric tensor).
\begin{equation}\label{metric-tensor_inverse}
g^{ij}=\sum_{n,\alpha}m_{n}^{-1}\left(\frac{\partial q_{i} }{\partial r_{n\alpha}}\right)\left(\frac{\partial q_{j} }{\partial r_{n\alpha}}\right)
\end{equation}
For clarity of notation we will up from now use \textit{Einstein} summation convention, discerning upper and lower indices. Therefore classical Hamilton function will be written as follows: $H=\frac{1}{2}g^{ij}P_{i}P_{j}+V(q_{i})$.
In cartesian coordinates the metric tensor has unit matrix representation (up to a constant factor).
However when transformed into, for example spherical or elliptical coordinates, it takes more complicated form.
The proper form of the quantum Hamiltonian for general coordinates system was given by B.Podolsky in 1928 \cite{Podolsky1928}:
\begin{equation}\label{podolsky}
\hat{H}=\frac{1}{2}g^{-\frac{1}{4}}\hat{p}_{i}g^{-\frac{1}{2}}g^{ij}\hat{p}_{j}g^{-\frac{1}{4}}+\hat{V}
\end{equation}
where $g$ is the determinant of the metric tensor. The above expression can be derived in following way. First we build \textit{Laplace-Beltrami} operator \cite{Jost} as
\begin{equation}\label{LB1}
\Delta:=div\left(\vec{\nabla}\right)
\end{equation}
with divergence of a vector field $F$ defined as
\begin{equation}\label{div}
divF:=g^{-\frac{1}{2}}\partial_{i}g^{-\frac{1}{2}}F^{i}
\end{equation}
and gradient of a scalar field $\phi$:
\begin{equation}\label{grad}
\vec{\nabla}\phi:=\partial^{i}\phi=g^{ij}\partial_{j}\phi
\end{equation}
Combining the two results in
\begin{equation}\label{LB1}
\Delta:=div\left(\vec{\nabla}\right)=g^{-\frac{1}{2}}\partial_{i}g^{-\frac{1}{2}}g^{ij}\partial_{j}
\end{equation}
Therefore if the transformation from cartesian to generalized coordinates is given by the relation $q^{i}=q^{i}(x_{1},y_{1},z_{1},...,x_{N},y_{N},z_{N})$ then the quantum-mechanical Hamilton operator transforms according to the expression:
\begin{equation}\label{podolsky0}
\hat{H}=\frac{1}{2}g^{-\frac{1}{2}}\hat{p}_{i}g^{-\frac{1}{2}}g^{ij}\hat{p}_{j}+\hat{V}
\end{equation}
and the corresponding Schr\"{o}dinger equation reads
\begin{equation}\label{podolsky0}
\frac{1}{2}g^{-\frac{1}{2}}\hat{p}_{i}g^{-\frac{1}{2}}g^{ij}\hat{p}_{j}\psi(q)+\left(\hat{V}-E\hat{id}\right)\psi(q)=0
\end{equation}
Note that wavefunction in cartesian representation was normalized according to the condition:
\begin{equation}\label{normalization}
\int_{V}|\psi(x)|^{2}dx=1
\end{equation}
where $dx=dx_{1}dy_{1}dz_{1}...dx_{N}dy_{N}dz_{N}$ is the volume element in space $V$.
Now, after transformation our wavefunction is expressed by generalized coordinates. But volume element transforms linearly with factor equal to \textit{Jacobian}: $dq=Jac[q(x)]dx=g^{\frac{1}{2}}dx$. Now demanding consistency:
\begin{equation}\label{normalization1}
\int_{V}|\psi(x)|^{2}dx=\int_{V}|\psi(q)|^{2}dq=\int_{V}|\psi(q)|^{2}g^{\frac{1}{2}}dx=1
\end{equation}
we find that
\begin{equation}\label{normalization1}
\psi(q)=g^{-\frac{1}{4}}\psi(x)
\end{equation}
Derived expression for the kinetic energy \ref{kinetic6} involves however angular momenta, while in Podolsky approach we utilize generalized momenta representation. Thus, the mapping from the classical expression
into the quantum-mechanical form cannot
be done directly.
In order to apply the Podolsky procedure to our form of the Hamilton function, we must first find the form of the podolsky Hamiltonian in the  angular momentum representation, then prove that this formula is consistent with the original one, i.e. gives the proper quantum-mechanical energy operator. Lets start from the assumption that generalized momentum is linearly related to the rovibronic angular momentum. This is very strong statement, and will need a detailed proof.
\begin{equation}\label{angular-linear}
p_{i}=\sum_{k=1}^{3}a_{ik}M'_{k}
\end{equation}
then classical expression for kinetic energy \ref{kinetic6} takes the form
\begin{equation}\label{kinetic7}
2T=\sum_{i,j,k,l}a_{ik}a_{jl}g^{ij}M'_{k}M'_{l}\equiv\sum_{k,l}G_{kl}M'_{k}M'_{l}
\end{equation}
Ask then what conditions must be satisfied in order that
\begin{equation}\label{podolskyangular}
\hat{H}=\frac{1}{2}G^{\frac{1}{4}}\sum_{ij}\hat{M}'_{i}G^{ij}G^{-\frac{1}{2}}\hat{M}'_{j}G^{\frac{1}{4}}+\hat{V}
\end{equation}
while knowing $G$? First lets invert eq. \ref{angular-linear}
\begin{equation}\label{linear-angular}
M'_{i}=\sum_{k=1}^{3}(a^{-1})_{ik}p_{k}\equiv\sum_{k=1}^{3}a^{ik}p_{k}
\end{equation}
of course
\begin{equation}\label{identity}
\sum_{k=1}^{3}a_{ik}a^{kj}=\delta_{i}^{k}
\end{equation}
Inserting \ref{linear-angular} into \ref{podolskyangular} yields in
\begin{align}\nonumber
 & \hat{H}=\frac{1}{2}a^{\frac{1}{2}}g^{\frac{1}{4}}\sum_{i,j,l,k,r,t}a^{ik}\hat{p}_{k}a_{ir}a_{tj}g^{ij}a^{-1}g^{-\frac{1}{2}}a^{jl}p_{l}a^{\frac{1}{2}}g^{\frac{1}{4}}+\hat{V}=&&\\\nonumber
 & =\frac{1}{2}a^{\frac{1}{2}}g^{\frac{1}{4}}\sum_{i,j,k,r,t}a^{ik}\hat{p}_{k}a_{ir}g^{ij}a^{-1}g^{-\frac{1}{2}}p_{l}a^{\frac{1}{2}}g^{\frac{1}{4}}+\hat{V}\equiv && \\
 & \equiv \frac{1}{2}s^{-\frac{1}{2}}g^{\frac{1}{4}}\sum_{i,j}\hat{p}_{i}g^{-\frac{1}{2}}g^{ij}\hat{p}_{j}g^{\frac{1}{4}}s^{\frac{1}{2}}+\hat{V} &&
\end{align}
only if
\begin{equation}\label{podolsky-condition}
a^{\frac{1}{2}}\sum_{k,r}a^{ik}\hat{p}_{k}a_{ir}a^{-1}=a^{-\frac{1}{2}}\hat{p}_{i}
\end{equation}
This result mean that if we're able to find linear relation between angular (or any other quantity) and linear momentum, and this relation would fulfill above condition, then we can replace classical Hamilton function represented by \ref{kinetic7} with \ref{podolskyangular}. Knowing the coefficients $a_{ij}$ and the metric tensor $g_{ij}$ is sufficient to find the quantum-mechanical expression for the Hamiltonian of a non-rigid body. This enables us to set up the next step, namely finding geometrical relation between rovibronic angular momentum and linear generalized momentum conjugated with euler angles.
Consequently lets expand angular momentum in terms linear momenta via the chain rule:

\begin{equation}
\resizebox{0.99\hsize}{!}{$\vec{J}=\frac{\partial T}{\partial\vec{\omega}}=\frac{\partial \dot{\theta}}{\partial\vec{\omega}}\frac{\partial T }{\partial\dot{\theta}}+
\frac{\partial \dot{\phi}}{\partial\vec{\omega}}\frac{\partial T }{\partial\dot{\phi}}+
\frac{\partial \dot{\chi}}{\partial\vec{\omega}}\frac{\partial T }{\partial\dot{\chi}}\equiv \frac{\partial \dot{\theta}}{\partial\vec{\omega}}p_{\theta}+
\frac{\partial \dot{\phi}}{\partial\vec{\omega}}p_{\phi}+
\frac{\partial \dot{\chi}}{\partial\vec{\omega}}p_{\chi}$}
\end{equation}
assuming that potential energy is independent of generalized velocities. Thereby, all we need is relation between the components of $\dot{\vec{\phi}},\dot{\vec{\theta}},\dot{\vec{\chi}}$ and components of angular velocity in molecule-fixed rotating frame $\omega_{x},\omega_{y},\omega_{z}$.
Imagine space and molecule-fixed frames of reference with angular velocities $\omega$ drawn along the corresponding axes. Vector as a tensor object must be invariant to coordinate change, while its components transform according to general tensor transformation. This implies vector equality
\begin{equation}
\vec{\omega}'=\vec{\omega}
\end{equation}
of angular velocities in both coordinate frames.
General expression for angular velocity in space-fixed frame reads
\begin{equation}
\vec{\omega}'=\omega'_{X}\hat{i}'+\omega'_{Y}\hat{j}'+\omega'_{Z}\hat{k}'
\end{equation}
and in rotating molecule-fixed frame:
\begin{equation}
\vec{\omega}=\omega_{x}\hat{i}+\omega_{y}\hat{j}+\omega_{z}\hat{k}
\end{equation}
After that we need to express $\dot{\vec{\phi}},\dot{\vec{\theta}},\dot{\vec{\chi}}$ components in a space-fixed basis. Lets make use of geometrical relations:
\begin{align}\nonumber
 & \dot{\vec{\phi}}=\dot{\phi}\hat{k}' &&\\
 & \dot{\vec{\theta}}=\frac{\hat{k}'\times\hat{k}}{||\vec{k}'\times\vec{k}||}\dot{\theta}=\frac{1}{\sin\theta}\hat{k}'\times\left(\cos\theta\cdot\hat{k}'+\sin\theta \cos\phi\hat{i}'+\sin\theta \sin\phi\hat{j}'\right)\dot{\theta}=\cos\phi\cdot\dot{\theta}\cdot\hat{j}'-\sin\phi\cdot\dot{\theta}\cdot\hat{i}'&& \\
& \dot{\vec{\chi}}=\cos\theta\cdot\dot{\chi}\cdot\hat{k}'+\sin\theta \cos\phi \cdot \dot{\chi}\cdot\hat{i}'+\sin\theta \sin\phi \cdot \dot{\chi}\cdot\hat{j}'&&
\label{omegavec}
\end{align}
On the other hand the components of angular velocities in both frames are related by Euler angles transformation matrix:
\begin{widetext}
\begin{equation}\label{omegaa}
\left(\begin{array}{ccc}
\omega_{x} \\
\omega_{y} \\
\omega_{z} \end{array} \right) = \left(\begin{array}{ccc} \cos\theta\cos\phi\cos\chi-\sin\phi\sin\chi & \cos\theta\sin\phi\cos\chi & -\sin\theta\cos\chi \\
-\cos\theta\cos\phi\sin\chi-\sin\phi\cos\chi & -\cos\theta\sin\phi\sin\chi+\cos\phi\cos\chi & \sin\theta\sin\chi \\
\sin\theta\cos\phi & \sin\theta\sin\phi & \cos\theta \\ \end{array}\right)\left(\begin{array}{ccc}
\omega'_{X}\\
\omega'_{Y}\\
\omega'_{Z} \end{array}\right)
\end{equation}
\end{widetext}
Provided components of $\vec{\omega}'$ expressed by Euler angles time derivatives (cf.\ref{omegavec}) we can write final transformation as

\begin{widetext}
\begin{equation}\label{omegaa}
\left(\begin{array}{ccc}
\omega_{x} \\
\omega_{y} \\
\omega_{z} \end{array} \right) = \left(\begin{array}{ccc} \cos\theta\cos\phi\cos\chi-\sin\phi\sin\chi & \cos\theta\sin\phi\cos\chi & -\sin\theta\cos\chi \\
-\cos\theta\cos\phi\sin\chi-\sin\phi\cos\chi & -\cos\theta\sin\phi\sin\chi+\cos\phi\cos\chi & \sin\theta\sin\chi \\
\sin\theta\cos\phi & \sin\theta\sin\phi & \cos\theta \\ \end{array}\right)\left(\begin{array}{ccc}
\sin\theta\cos\phi\dot{\chi}-\sin\phi\dot{\theta}\\
\sin\theta\cos\phi\dot{\chi}-\sin\phi\dot{\theta}\\
\dot{\phi}+\cos\theta\dot{\chi}\end{array}\right)
\end{equation}
\end{widetext}
obtaining set of linear equations with respect to Euler angles time derivatives. We can easily identify resulting transformation matrix and write inverse relations (it's faster to inverse this linear system by simple substitutions):
\begin{equation}\label{omegaa}
\left(\begin{array}{ccc}
\dot{\theta} \\
\dot{\phi} \\
\dot{\chi} \end{array} \right) = \left(\begin{array}{ccc} \sin\chi & \cos\chi & 0 \\
-\csc\theta\cos\chi & \sin\chi\csc\theta & 0 \\
\cot\theta\cos\chi & -\cot\theta\sin\chi & 1 \\ \end{array}\right)\left(\begin{array}{ccc}
\omega_{x} \\
\omega_{y} \\
\omega_{z} \end{array}\right)
\end{equation}
Now we are ready to write explicit relation from eq. \ref{angular-linear}:
\begin{equation}\label{Angular-generalized momentum}
\left(\begin{array}{ccc}
J_{x} \\
J_{y} \\
J_{z} \end{array} \right) = \left(\begin{array}{ccc} \sin\chi & -\csc\theta\cos\chi & \cot\theta\cos\chi \\
\cos\chi & \sin\chi\csc\theta & -\cot\theta\sin\chi \\
0 & 0 & 1 \\ \end{array}\right)\left(\begin{array}{ccc}
p_{\theta} \\
p_{\phi} \\
p_{\chi} \end{array}\right)
\end{equation}
After introducing vibrational angular momentum $\vec{j}$ to complete general momenta space, we are able to read coefficients in eq.\ref{linear-angular}:
\begin{widetext}
\begin{equation}\label{a-matrix}
\left(\begin{array}{ccc}
M_{x} \\
M_{y} \\
M_{z} \\
P_{1}\\
\vdots\\
P_{3N-6}\end{array} \right) = \left(\begin{array}{cccccc} \sin\chi & -\csc\theta\cos\chi & \cot\theta\cos\chi & -\tau^{x}_{1} &\ldots& -\tau^{x}_{3N-6} \\
\cos\chi & \sin\chi\csc\theta & -\cot\theta\sin\chi & -\tau^{y}_{1} &\ldots& -\tau^{y}_{3N-6} \\
0 & 0 & 1 & -\tau^{z}_{1} &\ldots& -\tau^{z}_{3N-6}\\
 & & & & & \\
 & 0_{3N-6\times 3} & & & 1_{3N-6} & \\
  & & & & & \\ \end{array}\right)\left(\begin{array}{ccc}
p_{\theta} \\
p_{\phi} \\
p_{\chi}\\
P_{1} \\
\vdots\\
P_{3N-6}\\ \end{array}\right)
\end{equation}
\end{widetext}
Note that we have 3 generalized momenta associated with Euler angles and $3N-6$ momenta incorporated in vibrational angular momentum and vibrational kinetic energy $\vec{j}=\sum_{k}\tau_{k}\cdot\vec{P}_{k}$. Relation \ref{linear-angular} requires also inverse $a$ matrix coefficients, which can be obtained with little mainpulation on \ref{a-matrix}:
\scriptsize
\begin{widetext}
\begin{equation}\label{a-matrix}
\left(\begin{array}{ccc}
p_{\theta} \\
p_{\phi} \\
p_{\chi}\\
P_{1} \\
\vdots\\
P_{3N-6}\\
\end{array} \right) = \left(\begin{array}{cccccccccc} \sin\chi & \cos\chi & 0 & -\sin\chi\cdot\tau^{x}_{1} &\cos\chi\cdot\tau^{y}_{1} & 0 & \ldots & -\sin\chi\cdot\tau^{x}_{3N-6} & \cos\chi\cdot\tau^{y}_{3N-6} & 0 \\
-\sin\chi\cos\chi & \sin\chi\sin\theta & \cos\theta & -\sin\theta\cos\chi\tau^{x}_{1} & \sin\theta\sin\chi\tau^{y}_{1} & \cos\theta\tau^{z}_{1} & \ldots & -\sin\theta\cos\chi\tau^{x}_{3N-6} & \sin\theta\sin\chi\tau^{y}_{3N-6} & \cos\theta\tau^{z}_{3N-6} \\
0 & 0 & 1 & 0 & 0 & \tau^{z}_{1} &\ldots& 0 & 0 & \tau^{z}_{3N-6}\\
 & & & & & \\
 & 0_{3N-6\times 3} & & & & 1_{3N-6} & &\\
  & & & & & \\ \end{array}\right)\left(\begin{array}{ccc}
M_{x} \\
M_{y} \\
M_{z} \\
P_{1}\\
\vdots\\
P_{3N-6} \end{array}\right)
\end{equation}
\end{widetext}
\normalsize
Utilizing some trigonometric identities we find that $a^{-1}\equiv\det a^{-1}=\frac{1}{\sin\theta}$; of course from Cauchy matrix theorem $a\equiv\det a=\sin\theta$. In the next step we should investigate if the condition
\begin{equation}\label{podolsky-condition}
a^{\frac{1}{2}}\sum_{i,j}a^{ij}\hat{p}_{j}a_{ij}a^{-1}=a^{-\frac{1}{2}}\hat{p}_{i}
\end{equation}
is satisfied. It may be rewritten to a simpler form:
\begin{equation}\label{podolsky-condition-simpler}
\sum_{i,j,k}a^{ji}\hat{p}_{i}a_{kj}a^{-1}=0
\end{equation}
Consequently,
\begin{equation}\label{podolsky-condition-simpler}
\sum_{i,j,k}a^{ji}a_{kj}\hat{p}_{i}a^{-1}+a^{-1}\sum_{i,j,k}a^{ji}\hat{p}_{i}a_{kj}=0
\end{equation}
\begin{equation}\label{podolsky-condition-simpler}
\sum_{k}\hat{p}_{k}a^{-1}+a^{-1}\sum_{i,j}a^{ji}\hat{p}_{i}\sum_{k}a_{kj}=0
\end{equation}
\begin{equation}\label{PC-decomposed}
-\frac{\cos\theta}{\sin^{2}\theta}+\frac{1}{\sin\theta}\sum_{i,j}a^{ij}\left(\hat{p}_{i}\sum_{k}a_{kj}\right)=0
\end{equation}
Now lets evaluate following a-matrix sums:
\begin{align}\nonumber
 & \sum_{k}a_{k1}=\sin\chi-\sin\theta\cos\chi &&\\\nonumber
 & \sum_{k}a_{k2}=\cos\chi-\sin\theta\sin\chi &&\\\nonumber
 & \sum_{k}a_{k3}=\cos\theta+1 &&\\\nonumber
 & \sum_{k}a_{kj}=(\sin\chi-\sin\theta\cos\chi)\tau^{x}_{j}+(\cos\chi-\sin\theta\sin\chi)\tau^{y}_{j}+&&\\
 & +(\cos\theta+1)\tau^{z}_{j}+1 \qquad j>3&&
\end{align}
Acting with momenta operators on respective sums:\\
$p_{1}\equiv p_{\theta}$
\begin{align}\nonumber
 & p_{\theta}\sum_{k}a_{k1}=-\cos\theta\cos\chi && \\\nonumber
 & p_{\theta}\sum_{k}a_{k2}=\cos\theta\sin\chi &&\\\nonumber
 & p_{\theta}\sum_{k}a_{k3}=-\sin\theta &&\\
 & p_{\theta}\sum_{k}a_{kj}=-\cos\theta\cos\chi\tau^{x}_{j}+\cos\theta\sin\chi\tau^{y}_{j}-\sin\theta\tau^{z}_{j} \quad j>3 &&
\end{align}
$p_{2}\equiv p_{\phi}$
\begin{align}\nonumber
 & p_{\phi}\sum_{k}a_{k1}=0 &&\\\nonumber
 & p_{\phi}\sum_{k}a_{k2}=0 &&\\\nonumber
 & p_{\phi}\sum_{k}a_{k3}=0 &&\\
 & p_{\phi}\sum_{k}a_{kj}=0 \quad j>3 &&
\end{align}
$p_{3}\equiv p_{\chi}$
\begin{align}\nonumber
 & p_{\chi}\sum_{k}a_{k1}=\cos\chi+\sin\theta\sin\chi &&\\\nonumber
 & p_{\chi}\sum_{k}a_{k2}=-\sin\chi+\sin\theta\cos\chi&&\\\nonumber
 & p_{\chi}\sum_{k}a_{k3}=0 &&\\\nonumber
 & p_{\chi}\sum_{k}a_{kj}=(\sin\theta\sin\chi+\cos\chi)\tau^{x}_{j}+ &&\\
 & +(\sin\theta\cos\chi-\sin\chi)\tau^{y}_{j} \quad j>3 &&
\end{align}
$p_{i}\equiv p_{i},\qquad i>3$
\begin{align}\nonumber
 & p_{3}\sum_{k}a_{k1}=0 &&\\\nonumber
 & \vdots &&\\
 & p_{3N-6}\sum_{k}a_{k3N-6}=0 &&
\end{align}
This allows us to write explicit expression for second term in \ref{PC-decomposed}:
\begin{align}
 & \sum_{i,j}a^{ij}\left(\hat{p}_{i}\sum_{k}a_{kj}\right)=\sum_{i}a^{i1}\left(\hat{p}_{i}\sum_{k}a_{k1}\right) &&\\\nonumber
 & +\sum_{i}a^{i1}\left(\hat{p}_{i}\sum_{k}a_{k1}\right)+&&\\
 & +\sum_{i}a^{i3}\left(\hat{p}_{i}\sum_{k}a_{k3}\right)+\sum_{i,j>4}a^{ij}\left(\hat{p}_{i}\sum_{k}a_{kj}\right)=... &&
\end{align}

two last terms give no contribution to the final sum and eventually
\begin{equation}
\begin{split}
...=-\sin\chi\cos\chi\cos\theta+\cos^{2}\chi\cot\theta+\cos\theta\cos\chi\sin\theta+\\
+\sin\chi\cos\chi\cos\theta-\cos\theta\sin\chi\cos\chi+\sin^{2}\chi\cot\theta=\cot\theta
\end{split}
\end{equation}
Coming back to \ref{PC-decomposed} we find Podolsky condition fulfilled:
\begin{equation}
-\frac{\cos\theta}{\sin^{2}\theta}+\frac{\cot\theta}{\sin\theta}=0
\end{equation}
This means that we can pass from classical quadratic form of the Hamilton function to the quantum-mechanical operator form denoted below \cite{Dennison1940}
\begin{equation}\label{Hamiltonian-quantum}
\hat{H}=\frac{1}{2}\mu^{\frac{1}{4}}\sum_{i,j}\hat{M}_{i}\mu_{ij}\mu^{-\frac{1}{2}}\hat{M}_{j}\mu^{\frac{1}{4}}+
\frac{1}{2}\mu^{\frac{1}{4}}\sum_{k}\hat{P}_{k}\mu^{-\frac{1}{2}}\hat{P}_{k}\mu^{\frac{1}{4}}+\hat{V}
\end{equation}
where following relation must be satisfied:
\begin{equation}\label{mu}
\mu=a^{T}\cdot g\cdot a
\end{equation}
\subsection{Watson Simplification}
After almost 30 years after derivation of the Hamiltonian from eq.\ref{Hamiltonian-quantum} by Darling and Dennison \cite{Dennison1940} J.K.G. Watson came up with a tricky way to simplify this operator using some sum rules and commutation relations \cite{Watson1968}. The final form of the Hamiltonian is very similar to the classical one
\begin{equation}\label{Watson-hamiltonian}
\hat{H}=\frac{1}{2}\sum_{i,j}\hat{M}_{i}\mu_{ij}\hat{M}_{j}+
\frac{1}{2}\sum_{i}\hat{P}_{i}^{2}+\hat{U}+\hat{V}
\end{equation}
where $\hat{U}$ occurs to be mass dependent contribution to the potential
\begin{equation}\label{Watson-hamiltonian}
\hat{U}=-\frac{\hbar^{2}}{8}\sum_{i}\mu_{ii}\equiv-\frac{\hbar^{2}}{8}Tr\mu
\end{equation}
To remind, in present case we are treating all particles of a system as point masses, not distinguishing between nuclei and electrons. For the sake of convention lets denote instantaneous position of $i-th$ particle in body fixed rotating frame $\vec{r}_{i}=\left(r_{ix},r_{iy},r_{iz}\right)$. Reference configuration, not necessarily equilibrium, will be abbreviated as $\vec{r}^{0}_{i}$. We require now six constraints that allow to specify the position and orientation of moving axes at every instant (relative to set of particles). In fact it was already introduced as the Eckart frame \ref{eq:Eckart}.To summarize and gather all conditions within present notation:
\begin{enumerate}
  \item \begin{equation} \label{CM=0}
\sum_{\i=1}^{N}m_{\i}\vec{r}_{\i}=0
\end{equation}
setting rotating frame of reference of molecular center of mass. As long as there's no external fields translation in space is the symmetry of the molecular Hamiltonian \cite{Bunker}, due to space uniformity. Hence molecular center of mass motion can be separated and corresponding constant energy taken as a reference.
  \item 
\begin{equation} 
\sum_{i}m_{i}\vec{r}^{0}_{i}\times\dot{\vec{\rho}}_{i}=0
\end{equation}
There's no angular momentum of generated by the system with respect to the rotating frame(in fact the conditions are more general $\sum_{i}m_{i}\vec{r}^{0}_{i}\times\vec{\rho}_{i}=0$)
\end{enumerate}
In order to separate the center of mass motion and rotation of the axis system we introduce a new set of $3N$ coordinates, say normal coordinates\footnote{sometimes the whole set of $3N$ coordinates is called rovibronic coordinates, while $3N-6$ internal coordinates pertain the 'normal' attribute}, related to cartesian by orthogonal transformation $\hat{R}=\in SO(3N)$. Of course the conditions for the body-fixed frame may be written in both types of coordinates, but normal coordinates provide expressions with easily discernible rotational, vibrational and coupling parts, making the problem more transparent.
In practise we have to decompose $\hat{R}$ into two: $\hat{R}=\hat{L}\cdot\hat{W}$.

\begin{align}\nonumber
 & Q_{1}=M^{-\frac{1}{2}}\sum_{i}m_{i}\rho_{i1} &&\\\nonumber
 & Q_{2}=M^{-\frac{1}{2}}\sum_{i}m_{i}\rho_{i2} &&\\
 & Q_{3}=M^{-\frac{1}{2}}\sum_{i}m_{i}\rho_{i3} &&
\end{align}

\begin{widetext}
\begin{equation}
 \left.	\begin{aligned}
Q_{4}=\sum_{\beta,\gamma,\delta}\left(I^{0}\right)^{\frac{1}{2}}_{1\beta}\sum_{i}m^{\frac{1}{2}}_{i}\epsilon_{\beta\gamma\delta}r_{i\gamma}^{0}m^{\frac{1}{2}}_{i}\rho_{i\delta}\\
Q_{5}=\sum_{\beta,\gamma,\delta}\left(I^{0}\right)^{\frac{1}{2}}_{2\beta}\sum_{i}m^{\frac{1}{2}}_{i}\epsilon_{\beta\gamma\delta}r_{i\gamma}^{0}m^{\frac{1}{2}}_{i}\rho_{i\delta}\\
Q_{6}=\sum_{\beta,\gamma,\delta}\left(I^{0}\right)^{\frac{1}{2}}_{3\beta}\sum_{i}m^{\frac{1}{2}}_{i}\epsilon_{\beta\gamma\delta}r_{i\gamma}^{0}m^{\frac{1}{2}}_{i}\rho_{i\delta}
       \end{aligned}
 \right\}
 \qquad \left(I^{0}\right)^{\frac{1}{2}_{\alpha\beta}}\sum_{i}m_{i}\left(\vec{r}_{i}^{0}\times\vec{\rho}_{i}\right)^{i}
\end{equation}
\end{widetext}

\begin{equation}
 \left.\begin{aligned}
Q_{k}=\sum_{i,\alpha}l_{\alpha i k}m^{\frac{1}{2}}_{i}\rho_{i\alpha}
       \end{aligned}
 \right\}
 \qquad k=7,...,3N
\end{equation}
In order to apply some rotating axes embedding we make use of Eckart conditions and write $Q(1)=...=Q(6)=0$.
$L$ matrix in one of possible ways can be understood as $3N\times3N-6$ rectangular transformation to normal coordinates unitary  matrix:
\small
\begin{widetext}
\begin{equation}\label{L-matrix}
\left(\begin{array}{ccccccc}
\frac{m^{\frac{1}{2}}_{1}}{M^{\frac{1}{2}}} & 0 & 0 & \ldots &  \frac{m^{\frac{1}{2}_{N}}}{M^{\frac{1}{2}}} & 0 & 0 \\
0& \frac{m^{\frac{1}{2}}_{1}}{M^{\frac{1}{2}}} & 0 &  \ldots & 0 & \frac{m^{\frac{1}{2}}_{N}}{M^{\frac{1}{2}}} & 0 \\
0& 0& \frac{m^{\frac{1}{2}}_{1}}{M^{\frac{1}{2}}}  & \ldots & 0 &0 &\frac{m^{\frac{1}{2}}_{N}}{M^{\frac{1}{2}}} \\
\sum_{\beta,\gamma}\left(I^{0}\right)^{\frac{1}{2}}_{1\beta}m^{\frac{1}{2}}_{1}\epsilon_{\beta\gamma1}r_{1\gamma}^{0} & & &\ldots & & & \sum_{\beta,\gamma}\left(I^{0}\right)^{\frac{1}{2}}_{1\beta}m^{\frac{1}{2}}_{N}\epsilon_{\beta\gamma1}r_{N\gamma}^{0}\\
\sum_{\beta,\gamma}\left(I^{0}\right)^{\frac{1}{2}}_{2\beta}m^{\frac{1}{2}}_{1}\epsilon_{\beta\gamma1}r_{1\gamma}^{0} & & &\ldots & & & \sum_{\beta,\gamma}\left(I^{0}\right)^{\frac{1}{2}}_{2\beta}m^{\frac{1}{2}}_{N}\epsilon_{\beta\gamma1}r_{N\gamma}^{0}\\
\sum_{\beta,\gamma}\left(I^{0}\right)^{\frac{1}{2}}_{3\beta}m^{\frac{1}{2}}_{1}\epsilon_{\beta\gamma1}r_{1\gamma}^{0} & & &\ldots & & & \sum_{\beta,\gamma}\left(I^{0}\right)^{\frac{1}{2}}_{3\beta}m^{\frac{1}{2}}_{N}\epsilon_{\beta\gamma1}r_{N\gamma}^{0}\\
l_{111} & l_{211} &l_{311}  & \ldots & l_{1N1} & l_{2N1} & l_{3N1}\\
\vdots\\ & &     & \vdots &\\
\vdots\\ & &    & \vdots &\\
l_{113N-6} & l_{213N-6} &l_{313N-6} & \ldots & l_{1N3N-6} & l_{2N3N-6} & l_{3N3N-6}\\
\end{array} \right)
\end{equation}
\end{widetext}
\normalsize
Imposing unitarity of the above matrix $L^{T}L=LL^{T}=1$  provides useful relations:
\begin{itemize}
  \item 
  \begin{align}\nonumber
  &\sum_{s=1}^{3N}\left(L^{T}\right)_{1s}L_{s1}=\frac{m_{1}}{M}+&&\\
  &+\sum_{\beta,\gamma,\alpha,\eta,\nu}\left(I^{0}\right)^{\frac{1}{2}}_{\alpha\beta}m_{1}\epsilon_{\beta\gamma\alpha}\epsilon_{\eta\nu\alpha}\left(I^{0}\right)^{\frac{1}{2}}_{\alpha\eta}r^{0}_{1\gamma}r^{0}_{1\nu}+&&\\
  & +\sum_{k=1}^{3N-6}l^{2}_{11k}=1 &&
  \end{align}

  \item \begin{equation}
  \begin{split}
  \sum_{i=1}^{N}\sum_{i=1}^{N}l_{\alpha ik}l_{\alpha il}=\delta_{kl} \qquad ,\alpha=1,2,3 \quad k=1,...,3N-6
  \end{split}
  \end{equation}
  \item \begin{equation}
  \begin{split}
  \sum_{\beta,\gamma,\delta=1}^{3}\sum_{i=1}^{N}\left(I^{0}\right)^{-\frac{1}{2}}_{1\beta}m^{\frac{1}{2}}_{i}\epsilon_{\beta\gamma\alpha}r_{i\gamma}^{0}
  l_{\alpha ik}=0
  \end{split}
  \end{equation}
  Left hand side of the last equation can be rearranged to:
  \begin{equation}
  \begin{split}
  \sum_{\beta=1}^{3}\left(I^{0}\right)^{-\frac{1}{2}}_{1\beta}\sum_{\gamma,\delta=1}^{3}\sum_{i=1}^{N}m^{\frac{1}{2}}_{i}\epsilon_{\beta\gamma\alpha}r_{i\gamma}^{0}
  l_{\alpha ik}=0
  \end{split}
  \end{equation}
  We can discover that moment of inertia tensor stays symmetric under index transposition, therefore provided that the entire above expression is zero, the other factor in the equation must be antisymmetric under odd index permutation. This implies that $\sum_{\gamma,\delta=1}^{3}\sum_{i=1}^{N}\sum_{i}m^{\frac{1}{2}}_{i}\epsilon_{\beta\gamma\delta}r_{i\gamma}^{0}m^{\frac{1}{2}}$ must be antisymmetric, and because Levi-Civita tensor is antisymmetric the following relation must be fulfilled:

  \begin{equation}
  \begin{split}
  \sum_{i}m^{\frac{1}{2}}_{i}r_{i\alpha}^{0}l_{\beta ik}=\sum_{i}m^{\frac{1}{2}}_{i}r_{i\beta}^{0}l_{\alpha ik}
  \end{split}
  \end{equation}
  \item Finally summing over rows we get
\begin{align}\nonumber
  & M^{-1}m^{\frac{1}{2}}_{i}m^{\frac{1}{2}}_{j}\delta_{\xi\phi}+\sum_{\alpha,\beta,\gamma,\eta,\nu}m^{\frac{1}{2}}_{i}\epsilon_{\beta\gamma\xi}
  \left(I^{0}\right)^{-\frac{1}{2}}_{\alpha\beta}r^{0}_{i\gamma}m^{\frac{1}{2}}_{j}\epsilon_{\eta\nu\phi}
  \left(I^{0}\right)^{-\frac{1}{2}}_{\alpha\eta}r^{0}_{i\nu}+&&\\
  & +\sum_{k=1}^{3N-6}l_{\xi ik}l_{\phi jk}=\delta_{ij}\delta_{\xi\phi} &&
\end{align}
  \item The last relation arising from the orthogonality of $L$ is useful a transformation of cartesian coordinates into normal modes:
  \begin{equation}
  \begin{split}
  r_{i\alpha}=r^{0}_{i\alpha}+m^{-\frac{1}{2}}_{i}\sum_{k=1}^{3N-6}l_{\alpha ik}Q_{k}+M^{-1}\sum_{i}m_{j}\rho_{j\alpha}+
  \\+\sum_{\beta\gamma\nu}\epsilon_{\alpha\beta\gamma}r^{0}_{i\gamma}\left(I^{0}\right)^{-\frac{1}{2}}_{\beta\nu}R_{\alpha}
  \end{split}
  \end{equation}
  the last to terms vanish in Eckart frame, yielding 'standard' relations between cartesian and normal coordinates.
\end{itemize}

There are two kinds of vibration-rotation interaction coefficients

\begin{enumerate}
  \item Coriolis coupling coefficients in vibrational angular momentum:
  \begin{equation}
  j_{\alpha}=\sum_{k,l}C^{\alpha}_{kl}\hat{Q}_{k}\hat{P}_{l}
  \end{equation}
  where $\hat{P}$ is momentum conjugate to $\hat{Q}$ i.e. $\frac{\partial L}{\partial\dot{Q}_{l}}=P_{l}$.
  The coefficients are defined in following way
   \begin{equation}
  C^{\alpha}_{kl}:=\sum_{\beta,\gamma,i}\epsilon_{\alpha\beta\gamma}l_{\beta ik}l_{\gamma il}=-C^{\alpha}_{lk}
  \end{equation}
  Note that $C^{\alpha}_{kl}=\sum_{i}\left(\vec{l}_{ik}\times\vec{l}_{il}\right)^{\alpha}$, hence $C^{\alpha}_{kk}=0$.
  \item Interaction coefficients (dependence of moment of inertia on normal coordinates):
   \begin{equation}\label{interaction_coeff}
  a^{\alpha\beta}_{k}:=\left(\frac{\partial I_{\alpha\beta}}{\partial Q_{k}}\right)_{0}
  \end{equation}
  Because moment of inertia tensor can be expanded in normal modes basis in following way
  \begin{widetext}
  \small
\begin{align}\nonumber
  & I_{\alpha\beta}=\sum_{\gamma,\delta,\eta}\epsilon_{\alpha\gamma\delta}\epsilon_{\beta\eta\delta}\sum_{i}m_{i}r_{i \gamma }r_{i \delta}=&&\\\nonumber
  & =\sum_{\gamma,\delta,\eta}\epsilon_{\alpha\gamma\delta}\epsilon_{\beta\eta\delta}\sum_{i}m_{i}\left(r^{0}_{i \gamma }+m^{-\frac{1}{2}}_{i}\sum_{n=1}^{3N-6}l_{\gamma in}Q_{n}\right)\left(r^{0}_{i \eta }+m^{-\frac{1}{2}}_{i}\sum_{s=1}^{3N-6}l_{\eta is}Q_{s}\right)= &&\\
  & =\sum_{\gamma,\delta,\eta}\epsilon_{\alpha\gamma\delta}\epsilon_{\beta\eta\delta}\sum_{i}m_{i}\left(r^{0}_{i \gamma }r^{0}_{i \eta }+r^{0}_{i \gamma }\sum_{s=1}^{3N-6}m^{-\frac{1}{2}}_{i}l_{\eta is}Q_{s}+r^{0}_{i \eta }\sum_{n=1}^{3N-6}m^{-\frac{1}{2}}_{i}l_{\gamma in}Q_{n}+m^{-1}_{i}\sum_{s,n=1}^{3N-6}l_{\eta is}l_{\gamma in}Q_{n}Q_{s}\right) &&
\end{align}
\end{widetext}
\normalsize
  thus finally:

\begin{align}\nonumber
  & a^{\alpha\beta}_{k}=\sum_{\gamma,\delta,\eta, i}\left(r_{i \gamma }^{0}m^{-\frac{1}{2}}_{i}l_{\eta ik}+r_{i \eta }^{0}m^{-\frac{1}{2}}_{i}l_{\gamma ik}\right.+ &&\\\nonumber
  & \left.+m^{-1}_{i}\sum_{s=1}^{3N-6}\left(l_{\eta is}l_{\gamma ik}Q_{s}+l_{\eta ik}l_{\gamma is}Q_{s}\right)\epsilon_{\alpha\gamma\delta}\epsilon_{\beta\eta\delta}\right)=&&\\
  &=2\sum_{\gamma,\delta,\eta, i}r_{i \gamma }^{0}m^{-\frac{1}{2}}_{i}l_{\eta ik}\epsilon_{\alpha\gamma\delta}\epsilon_{\beta\eta\delta}&&
\end{align}
  As one could expect our interaction coefficient is symmetric with respect to indices permutation.
  \normalsize
\end{enumerate}
Now making use of above definitions and orthogonality of $L$ matrix we can find the following sum rules \cite{Watson1968}
\begin{equation} \label{sumrule1}
  \begin{split}
  \sum_{n=1}^{3N-6}C_{kn}^{\alpha}C^{\beta}_{ln}=\delta_{\alpha\beta}-\sum_{i=1}^{N}l_{\beta ik}l_{\alpha il}-\frac{1}{4}a^{\alpha\gamma}_{k}\left(I^{0}\right)^{-1}_{\gamma\delta}a^{\delta\beta}_{l}
  \end{split}
  \end{equation}
\begin{align}\nonumber
 & \sum_{k=1}^{3N-6}a^{\alpha\beta}_{k}a^{\gamma\delta}_{k}=4\sum_{i=1}^{N}\left(\delta_{\alpha\beta}\delta_{\gamma\delta}m_{i}r^{0}_{i\epsilon}r^{0}_{i\epsilon}-
  \delta_{\alpha\beta}m_{i}r^{0}_{i\gamma}r^{0}_{i\delta}\right.-&&\\\nonumber
  & -\delta_{\gamma\delta}m_{i}r^{0}_{i\alpha}r^{0}_{i\beta}+\delta_{\alpha\gamma}m_{i}r^{0}_{i\beta}r^{0}_{i\delta}-&&\\
  &\left.-\sum_{j,\theta,\eta,\epsilon,\xi}\epsilon_{\alpha\epsilon\eta}\epsilon_{\gamma\xi\theta}m_{i}r^{0}_{i\beta}r^{0}_{i\epsilon}m_{j}r^{0}_{j\delta}r^{0}_{j\xi}
  \left(I^{0}\right)^{-1}_{\eta\theta}\right) &&
\label{sumrule2}
\end{align}
\begin{align}\nonumber
 & \sum_{l=1}^{3N-6}\xi^{\alpha}_{kl}a^{\beta\gamma}_{l}=\frac{1}{2}\sum_{\epsilon}\epsilon_{\alpha\beta\gamma}a^{\epsilon\epsilon}_{k}-
  \sum_{\epsilon}\epsilon_{\alpha\beta\epsilon}a^{\epsilon\gamma}_{k}-&&\\
  &-\sum_{i,\epsilon,\xi,\delta}\epsilon_{\beta\delta\epsilon}m_{i}r^{0}_{i\delta}r^{0}_{i\gamma}
   \left(I^{0}\right)^{-1}_{\epsilon\gamma}a^{\xi\alpha}_{k} &&
\label{sumrule3}
\end{align}
  For future purposes it is essential to derive commutation relations for angular momentum operator. Thereby we claim that:
  \begin{theorem}
  Angular momentum associated with rotations of body-fixed coordinates frame with respect to space-fixed frame obeys following commutation relations
\begin{equation}
\left[\hat{J}_{\alpha},\hat{J}_{\beta}\right]=-i\hbar\epsilon_{\alpha\beta\gamma}\hat{J}_{\gamma}
\end{equation}
  \end{theorem}

  \begin{proof}
  We've shown that angular momentum is related to generalized momenta by linear expressions (cf.\ref{Angular-generalized momentum}). Then left hand side for choice of $x$ and $z$ components is equal to (on account of linearity of commutators):
  \begin{equation}
  \small
 \left[\hat{J}_{x},\hat{J}_{z}\right]=\left[\sin\chi\hat{P}_{\theta},\hat{P}_{\chi}\right]-\left[\frac{\cos\chi}{\sin\theta}\hat{P}_{\phi},\hat{P}_{\chi}\right]+
 \left[\cot\theta\cos\chi\hat{P}_{\chi},\hat{P}_{\chi}\right]=...
 \end{equation}
 \normalsize
 here comes perfect moment to introduce following useful lemma:
 \begin{align}\nonumber
  \left[f(x),\hat{p}_{x}\right]&=f(x)\hat{p}_{x}-p_{x}f(x)=-i\hbar f(x)\frac{d}{dx}+i\hbar\frac{df}{dx}+ &&\\
 & +i\hbar f(x)\frac{d}{dx}=i\hbar\frac{d}{dx}f(x) &&
 \end{align}
 In other words, commutator of any differentiable function of variable $x$ with $x$ component of quantum-mechanical momentum operator is proportional to derivative of function $f$ with respect to $x$. This relation simplifies many commutator operations.
   \begin{equation}
   \begin{split}
 ...=i\hbar\cos\chi \hat{P}_{\theta}+i\hbar\sin\chi\csc\theta\hat{P}_{\phi}+ \left[\cot\theta\cos\chi,\hat{P}_{\chi}\right]\hat{P}_{\chi}=\\
 =i\hbar \left(\cos\chi \hat{P}_{\theta}+\sin\chi\csc\theta\hat{P}_{\phi}-\cot\theta\sin\chi\hat{P}_{\chi}\right)=i\hbar\hat{J}_{y}
 \end{split}
 \end{equation}
 and similarly for other permutations.
  \end{proof}
  Above relations are historically named \textit{anomalous commutation relations} to be distinguished from commutation relations for space-fixed angular momentum operators, where there's no $"-"$ in front of right hand side. Now lets turn our attention into vibrational angular momentum operators:
    \begin{theorem}
  Angular momentum associated with vibrations of atoms in molecules with respect to body-fixed frame (Eckart frame)obeys following commutation relations
\begin{equation}
\left[\hat{j}_{\alpha},\hat{j}_{\beta}\right]=i\hbar\epsilon_{\alpha\beta\gamma}\hat{j}_{\gamma}
\end{equation}
  \end{theorem}

  \begin{proof}
 Lets expand left hand side from definition of vibrational angular momentum
\begin{align}\nonumber
&\left[\hat{j}_{\alpha},\hat{j}_{\beta}\right]=\sum_{k,l,m,n=1}^{3N-6}C^{\alpha}_{kl}C^{\beta}_{mn}\left[\hat{Q}_{k}\hat{P}_{l},\hat{Q}_{m}\hat{P}_{n}\right]=&&\\\nonumber
& =i\hbar\sum_{k,l,m,n=1}^{3N-6}C^{\alpha}_{kl}C^{\beta}_{mn}\left(-\hat{Q}_{k}\delta_{lm}\hat{P}_{n}+\hat{Q}_{m}\delta_{kn}\hat{P}_{l}\right)=&&\\\nonumber
& =i\hbar\sum_{k,l,m}\left(C^{\alpha}_{kl}C^{\beta}_{mk}-C^{\alpha}_{mk}C^{\beta}_{kl}\right)\hat{Q}_{m}\hat{P}_{l}=&&\\
& =i\hbar\sum_{k,l,m}\left(C^{\alpha}_{km}C^{\beta}_{lm}-C^{\alpha}_{km}C^{\beta}_{lm}\right)\hat{Q}_{k}\hat{P}_{l} &&
\end{align}
where the last equality comes from relabelling summation indices: $m\rightarrow k, k\rightarrow m, l \rightarrow l$ with subsequent transposition of $m,l$ on Coriolis coefficients. Further use of the sum rules \ref{sumrule1}-\ref{sumrule3} yields:
\begin{align}\nonumber
& \left[\hat{j}_{\alpha},\hat{j}_{\beta}\right]=i\hbar\sum_{k,l}\left[\delta_{\alpha\beta}\delta_{kl}-\sum_{i,\gamma,\delta}\left(l_{\beta ik}l_{\alpha il}+\frac{1}{4}a^{\alpha\gamma}_{k}\left(I^{0}\right)^{-1}_{\gamma\delta}a^{\delta\beta}_{l}\right)\right. &&\\\nonumber
&\left.-\left(\delta_{\beta\alpha}\delta_{kl}-\sum_{i,\gamma,\delta}\left(l_{\alpha ik}l_{\beta il}+\frac{1}{4}a^{\beta\gamma}_{k}\left(I^{0}\right)^{-1}_{\gamma\delta}a^{\delta\alpha}_{l}\right)\right)\right]\hat{Q}_{k}\hat{P}_{l}= &&\\\nonumber
&=i\hbar\sum_{k,l=1}^{3N-6}\sum_{i=1}^{N}\left(l_{\alpha ik}l_{\beta il}-l_{\beta ik}l_{\alpha il}\right)\hat{Q}_{k}\hat{P}_{l}-&&\\
& -\frac{i\hbar}{4}\sum_{k,l,\gamma,\delta}a^{\alpha\gamma}_{k}\left(I^{0}\right)^{-1}_{\gamma\delta}a^{\delta\beta}_{l}\left(\hat{Q}_{k}\hat{P}_{l}-
\hat{Q}_{l}\hat{P}_{k}\right)=...&&
\end{align}
in last transformation we changed summation indices $k\rightarrow l, \gamma \rightarrow \delta$. Now contracting first term in last line by noticing totally antisymmetric term, and applying Coriolis coefficient definition:
\begin{align}\nonumber
 & i\hbar\sum_{k,l=1}^{3N-6}\sum_{i=1}^{N}\left(l_{\alpha ik}l_{\beta il}-l_{\beta ik}l_{\alpha il}\right)\hat{Q}_{k}\hat{P}_{l}=&&\\\nonumber
 & =i\hbar\sum_{k,l=1}^{3N-6}\sum_{i=1}^{N}\epsilon_{\gamma\alpha\beta}l_{\alpha ik}l_{\beta il}\hat{Q}_{k}\hat{P}_{l}=&&\\
 & =\sum_{k,l}C^{\alpha}_{kl}\hat{Q}_{k}\hat{P}_{l}=i\hbar\epsilon_{\alpha\beta\gamma}\hat{j}_{\gamma} &&
\end{align}

\begin{equation}
 \begin{split}
...=i\hbar\epsilon_{\alpha\beta\gamma}\hat{j}_{\gamma}-\frac{i\hbar}{4}\sum_{k,l,\gamma,\delta}a^{\alpha\gamma}_{k}\left(I^{0}\right)^{-1}_{\gamma\delta}a^{\delta\beta}_{l}\left(\hat{Q}_{k}\hat{P}_{l}-
\hat{Q}_{l}\hat{P}_{k}\right)
\end{split}
\end{equation}
but last expression is a product of totally symmetric (interaction coefficients are and moment of inertia are symmetric tensors)  and totally antisymmetric factor, with summation over all indices, summing up eventually to zero.
  \end{proof}
  As introduced in \ref{kinetic6} the $\mu_{\alpha\beta}$ is defined as reciprocal of $I'_{\alpha\beta}$
  \begin{equation}
  \mu_{\alpha\beta}=\left(I'^{-1}\right)_{\alpha\beta}
  \end{equation}
  where
  \begin{equation}
  I'_{\alpha\beta}=I_{\alpha\beta}-\sum_{klm}C_{kl}^{\alpha}C_{lm}^{\beta}Q_{k}Q_{l}
  \end{equation}
  here terms linear in normal coordinates are the same in $I$ and $I'$, therefore interaction coefficients $a_{k}^{\alpha\beta}$ are the same for both.
Now we want to express our refined moment of inertia by normal coordinates
\begin{align}\nonumber
& I_{\alpha\beta}=\sum_{\gamma,\delta,\eta}\sum_{i}\epsilon_{\alpha\gamma\eta}\epsilon_{\beta\delta\eta}m_{i}r_{i\gamma}r_{i\delta}=
  \sum_{\gamma,\delta\eta}\sum_{i}\epsilon_{\alpha\gamma\eta}\epsilon_{\beta\delta\eta}m_{i}r^{0}_{i\gamma}r^{0}_{i\delta}+&&\\\nonumber
  & +\sum_{k}a^{\alpha\beta}_{k}Q_{k} +\sum_{\gamma,\delta,\eta}\sum_{i}\sum_{s,n}\epsilon_{\alpha\gamma\eta}\epsilon_{\beta\delta\eta}l_{\eta is}l_{\gamma in}Q_{n}Q_{s}= I^{0}_{\alpha\beta}+&& \\\nonumber
  &+\sum_{k}a^{\alpha\beta}_{k}Q_{k}+\sum_{\gamma,\delta,\eta}\sum_{i}\sum_{s,n}\left(\delta_{\alpha\beta}\delta_{\gamma\eta}-\delta_{\alpha\eta}\delta_{\beta\gamma}\right)l_{\eta is}l_{\gamma in}Q_{n}Q_{s}=  && \\\nonumber
  & =I^{0}_{\alpha\beta}+\sum_{k}a^{\alpha\beta}_{k}Q_{k}+\delta_{\alpha\beta}\sum_{\gamma}\sum_{i}\sum_{s,n}l_{\gamma is}\l_{\gamma im}Q_{n}Q_{s}-  && \\\nonumber
 & -\sum_{i}\sum_{s,n}l_{\alpha is}l_{\beta in}Q_{n}Q_{s}-\sum_{k}\delta_{\alpha\beta}Q_{k}^{2}+ &&\\\nonumber
 & +\sum_{i}\sum_{k,l}l_{\alpha il}l_{\beta ik}Q_{l}Q_{k}+\frac{1}{4}\sum_{k,l}\sum_{\gamma\delta}a^{\alpha\gamma}_{k}\left(I^{0}\right)^{-1}_{\gamma\delta}a^{\delta\beta}_{l}Q_{k}Q_{l}=&&\\
& =I^{0}_{\alpha\beta}+\sum_{k}a^{\alpha\beta}_{k}Q_{k}+\frac{1}{4}\sum_{k,l}\sum_{\gamma\delta}a^{\alpha\gamma}_{k}\left(I^{0}\right)^{-1}_{\gamma\delta}a^{\delta\beta}_{l}Q_{k}Q_{l} &&
\end{align}
  By taking a closer look at above expression one may infer that there exist a simpler factorized form of the last sum:
  \begin{equation}
  I'_{\alpha\beta}=\sum_{\gamma\delta}I''_{\alpha\gamma}\left(I'^{-1}\right)_{\gamma\delta}I''_{\delta\beta} \qquad i.e.\; I'=I''\left(I^{0}\right)^{-1}I''
  \end{equation}
  where $I''_{\alpha\beta}=I^{0}_{\alpha\beta}+\frac{1}{2}\sum_{k}a^{\alpha\beta}_{k}Q_{k}$. Note that $I'$ as a function of normal coordinates is much simpler than $I$, although physical significance of $I'$ is harder to visualize. That's the reason for introducing $I'$ and $\mu$ subsequently.
  In order to rearrange Podolsky Hamiltonian into simpler form it will be necessary to calculate two types of commutators, namely $\left[\hat{j}_{\alpha},\mu\right], \left[\hat{j}_{\alpha},\mu_{\alpha\beta}\right]$, hence lets start with matrix form of $\mu$ operator, for easier manipulation:
    \begin{equation}
 \resizebox{\hsize}{!}{$ \left[\hat{j}_{\alpha},\mu\right]=\left[\hat{j}_{\alpha},I''^{-1}I^{0}I''^{-1}\right]=\left[\hat{j}_{\alpha},I''^{-1}\right]I^{0}I''^{-1}+
  I''^{-1}I^{0}\left[\hat{j}_{\alpha},I''^{-1}\right] $}
  \end{equation}
  using the fact that $I''I''^{-1}=1$ results in equality: $\left[\hat{j}_{\alpha},I''\right]I''^{-1}+I''\left[\hat{j}_{\alpha},I''^{-1}\right]=0$ which when inserted into above expression yields in:
      \begin{equation}
  \left[\hat{j}_{\alpha},\mu\right]=-I''^{-1}\left[\hat{j}_{\alpha},I''\right]\mu-\mu\left[\hat{j}_{\alpha},I''\right]I''^{-1}
  \end{equation}
By considering matrix elements of $\mu$ we can continue calculations in following fashion:
      \begin{equation}
 \resizebox{\hsize}{!}{$ \sum_{\alpha}\left[\hat{j}_{\alpha},\mu_{\alpha\beta}\right]=-\sum_{\alpha,\gamma,\delta}\left(\left(I''^{-1}\right)_{\alpha\gamma}\left[\hat{j}_{\alpha},I''_{\gamma\delta}\right]\mu_{\delta\beta}-\mu_{\alpha\gamma}
  \left[\hat{j}_{\alpha},I''_{\gamma\delta}\right]\left(I''^{-1}\right)_{\delta\beta}\right)=0$}
  \end{equation}
This result is quite surprising, as the vibrational angular momentum is a differential operator, which usually do not commute with functions of coordinates. Having that done, attention can be turned into the commutator of the determinant $\mu$.  Lets note that since $\mu=\det\left(I''^{-1}\right)^{2}\det\left(I^{0}\right)$ is expressed by the determinant of $I''$ we can consider only $\left[\hat{j}_{\alpha},I''\right]$. After some calculations we obtain:
   \begin{equation}
  \left[\hat{j}_{\alpha},I''\right]=i\hbar\sum_{\gamma,\eta,\epsilon,\delta,\xi}I''\epsilon_{\gamma\eta\epsilon}K^{0}_{\eta\delta}\left(I^{0}\right)^{-1}_{\epsilon\xi}
  I''_{\xi\alpha}\left(I''^{-1}\right)_{\gamma\delta}
  \end{equation}
  At present stage we can attempt to write down the Podolsky Hamiltonian in a new form. Firstly lets deal with the Hamiltonian in \ref{Hamiltonian-quantum} investigating the two terms separately:
\begin{align}\nonumber
 & \mu^{\frac{1}{4}}\hat{M}_{i}\mu_{ij}\mu^{-\frac{1}{2}}\hat{M}_{j}\mu^{\frac{1}{4}}=\mu_{ij}\mu^{\frac{1}{4}}\hat{M}_{i}\mu^{-\frac{1}{2}}\hat{M}_{j}\mu^{\frac{1}{4}}=&&\\\nonumber
  & =\mu_{ij}\mu^{\frac{1}{4}}\hat{M}_{i}\mu^{-\frac{1}{4}}\hat{M}_{j}+\mu_{ij}\mu^{\frac{1}{4}}\hat{M}_{i}\mu^{-\frac{1}{2}}\left[\hat{M}_{j},\mu^{\frac{1}{4}}\right]=&&\\\nonumber
  & =\mu_{ij}\hat{M}_{i}\hat{M}_{j}+\mu_{ij}\mu^{\frac{1}{4}}\left[\hat{M}_{i},\mu^{-\frac{1}{4}}\right]\hat{M}_{j}+\mu_{ij}\mu^{-\frac{1}{4}}\hat{M}_{i}\left[\hat{M}_{j},\mu^{\frac{1}{4}}\right]+&&\\\nonumber
  & +\mu_{ij}\mu^{\frac{1}{4}}\left[\hat{M}_{i},\mu^{-\frac{1}{2}}\right]\left[\hat{M}_{j},\mu^{\frac{1}{4}}\right]=\mu_{ij}\hat{M}_{i}\hat{M}_{j}+\mu_{ij}\mu^{\frac{1}{4}}\left[\hat{M}_{i},\mu^{-\frac{1}{4}}\right]\hat{M}_{j}+&&\\\nonumber
  &\mu_{ij}\mu^{-\frac{1}{4}}\left[\hat{M}_{i},\left[\hat{M}_{j},\mu^{\frac{1}{4}}\right]\right] +\mu_{ij}\mu^{\frac{1}{4}}\left[\hat{M}_{i},\mu^{-\frac{1}{2}}\right]\left[\hat{M}_{j},\mu^{\frac{1}{4}}\right]=&&\\\nonumber
  & =\mu_{ij}\hat{M}_{i}\hat{M}_{j}+\mu_{ij}I''^{-\frac{1}{2}}\left[\hat{j}_{i},I''^{\frac{1}{2}}\right]\hat{M}_{j}+\mu_{ij}I''^{\frac{1}{2}}\left[\hat{j}_{j},I''^{-\frac{1}{2}}\right]\hat{M}_{i}+&&\\
  &+\mu_{ij}\mu_{ij}I''^{\frac{1}{2}}\left[\hat{j}_{j},\left[\hat{j}_{j},I''^{-\frac{1}{2}}\right]\right]+\mu_{ij}\mu_{ij}I''^{-\frac{1}{2}}\left[\hat{j}_{i},I''\right]\left[\hat{j}_{j},I''^{-\frac{1}{2}}\right]=... &&
\end{align}
where we've made use of commutation relation for $\mu$ tensor elements derived earlier, as well as of the expression for determinant of $\mu$ matrix.  Lets treat $\hat{j}_{i}$ as differential operator (as it is indeed), then:
\begin{equation}
\left[\hat{j}_{i},I''^{\frac{1}{2}}\right]=\frac{1}{2}I''^{-\frac{1}{2}}\left[\hat{j}_{i},I''\right]
  \end{equation}

\begin{align}\nonumber
  &...=\mu_{ij}\hat{M}_{i}\hat{M}_{j}+\mu_{ij}I''^{\frac{1}{2}}\left[\hat{j}_{i},-\frac{1}{2}I''^{-\frac{3}{2}}\left[\hat{j}_{j},I''\right]\right]-\\\nonumber
 &\frac{1}{2}\mu_{ij}I''^{-2}\left[\hat{j}_{i},I''\right]
  \left[\hat{j}_{j},I''\right] =\mu_{ij}\hat{M}_{i}\hat{M}_{j}-\frac{1}{2}\mu_{ij}I''^{-1}\left[\hat{j}_{i},\left[\hat{j}_{j},I''\right]\right]-\\\nonumber
  &-\frac{1}{2}\mu_{ij}I''^{\frac{1}{2}}
  \left[\hat{j}_{i},I''^{-\frac{3}{2}}\right]\left[\hat{j}_{j},I''\right]-\frac{1}{2}\mu_{ij}I''^{-2}
  \left[\hat{j}_{i},I''\right]\left[\hat{j}_{j},I''\right]=\\\nonumber
  &=\mu_{ij}\hat{M}_{i}\hat{M}_{j}-\frac{1}{2}\mu_{ij}I''^{-1}\left[\hat{j}_{i},\left[\hat{j}_{j},I''\right]\right]+\frac{3}{4}\mu_{ij}I''^{-2}\left[\hat{j}_{i},I''\right]\left[\hat{j}_{j},I''\right]- \\\nonumber
  &-\frac{1}{2}\mu_{ij}I''^{-2}\left[\hat{j}_{i},I''\right]\left[\hat{j}_{j},I''\right]=\mu_{ij}\hat{M}_{i}\hat{M}_{j}-\\
  &-\frac{1}{2}\mu_{ij}I''^{-1}\left[\hat{j}_{i},\left[\hat{j}_{j},I''\right]\right]+\frac{1}{4}\mu_{ij}I''^{-2}\left[\hat{j}_{i},I''\right]\left[\hat{j}_{j},I''\right]
\end{align}
  Now consider the vibrational kinetic energy term in \ref{Hamiltonian-quantum}:
\begin{align}\nonumber
 &  \mu^{\frac{1}{4}}\hat{P}_{k}\mu^{-\frac{1}{2}}\hat{P}_{k}\mu^{\frac{1}{4}}=\mu^{\frac{1}{4}}\hat{P}_{k}\mu^{-\frac{1}{4}}\hat{P}_{k}+ &&\\\nonumber
 &  +\mu^{\frac{1}{4}}\hat{P}_{k}\mu^{-\frac{1}{2}}\left[\hat{P}_{k},\mu^{\frac{1}{4}}\right]=\hat{P}_{k}^{2}+ &&\\
 & +\mu^{\frac{1}{4}}\left[\hat{P}_{k},\mu^{-\frac{1}{4}}\right]\hat{P}_{k}+  \mu^{\frac{1}{4}}\hat{P}_{k}\mu^{-\frac{1}{2}}\left[\hat{P}_{k},\mu^{\frac{1}{4}}\right]=... &&
\end{align}
  Evaluating separately appropriate commutators:
\begin{align}\nonumber
 &\left[\hat{P}_{k},\mu^{-\frac{1}{4}}\right]=\left(I^{0}\right)^{-\frac{1}{4}}\left[\hat{P}_{k},I''^{\frac{1}{2}}\right]=-i\hbar \left(I^{0}\right)^{-\frac{1}{4}}\frac{\partial I''^{\frac{1}{2}}}{\partial Q_{k}}\\\nonumber
& \left[\hat{P}_{k},\mu^{\frac{1}{4}}\right]=\left(I^{0}\right)^{\frac{1}{4}}\left[\hat{P}_{k},I''^{-\frac{1}{2}}\right]=-i\hbar \left(I^{0}\right)^{\frac{1}{4}}\frac{\partial I''^{-\frac{1}{2}}}{\partial Q_{k}}
\left[\hat{P}_{k},\mu^{\frac{1}{4}}\right]=&&\\
& =-i\hbar \left(I^{0}\right)^{-\frac{1}{2}}\frac{\partial I''}{\partial Q_{k}}&&
\end{align}
  hence,

\begin{align}\nonumber
& ...=\hat{P}_{k}^{2}-i\hbar I''^{-\frac{1}{2}}\frac{\partial I''^{\frac{1}{2}}}{\partial Q_{k}}\left(-i\hbar \frac{\partial}{\partial Q_{k}}\right)+
\mu^{-\frac{1}{4}}\hat{P}_{k}\left[\hat{P}_{k},\mu^{\frac{1}{4}}\right]+&&\\\nonumber
& +\mu^{\frac{1}{4}}\left[\hat{P}_{k},\mu^{-\frac{1}{2}}\right]\left[\hat{P}_{k},\mu^{\frac{1}{4}}\right]=\hat{P}_{k}^{2}-\hbar^{2}I''^{-\frac{1}{2}}\frac{\partial I''^{\frac{1}{2}}}{\partial Q_{k}} \frac{\partial}{\partial Q_{k}}-&&\\\nonumber
&-\hbar^{2}I''^{\frac{1}{2}}\frac{\partial^{2} I''^{-\frac{1}{2}}}{\partial Q_{k}^{2}}-\hbar^{2}I''^{\frac{1}{2}}\frac{\partial I''^{-\frac{1}{2}}}{\partial Q_{k}} \frac{\partial}{\partial Q_{k}}-\hbar^{2}I''^{-\frac{1}{2}}\frac{\partial I''}{\partial Q_{k}} \frac{\partial I''^{-\frac{1}{2}}}{\partial Q_{k}}=&&\\
& =\hat{P}_{k}^{2}+\frac{1}{2}\hbar^{2}I''^{-1}\frac{\partial^{2} I''}{\partial Q_{k}^{2}}-\frac{1}{4}\hbar^{2}I''^{-2}\left(\frac{\partial I''}{\partial Q_{k}}\right)^{2}&&
\end{align}

  Following Watson we denote respective terms in the residual Hamiltonian as $U_1,U_2,U_3,U_4$.

First lets evaluate $U_{1}$:
\begin{widetext}
\begin{align}\nonumber
 & U_{1}=\frac{1}{8}\sum_{\alpha\beta}I''^{-2}\mu_{\alpha\beta}\left[\hat{j}_{\alpha},I''\right]\left[\hat{j}_{\beta},I''\right]=\frac{1}{8}\sum_{\alpha\beta}\left(I'^{-1}\right)_{\alpha\beta}+ &&\\\nonumber
 & +i\hbar I''\sum_{\gamma,\eta,\epsilon,\xi,\delta}\epsilon_{\gamma\eta\epsilon}K^{0}_{\eta\delta}\left(I^{0-1}\right)_{\epsilon\xi}\left(I''\right)_{\xi\alpha}\left(I''^{-1}\right)
_{\gamma\delta}+ i\hbar\sum_{\theta,\phi,\psi,\nu,\omega}\epsilon_{\theta\phi\psi}K^{0}_{\phi\omega}\left(I^{0-1}\right)_{\psi\nu}\left(I''\right)_{\nu\beta}\left(I''^{-1}\right)
_{\theta\omega}=&&\\
&-\frac{\hbar^{2}}{8}\sum_{\alpha,\beta,\lambda,\kappa}\sum_{\gamma,\eta,\epsilon,\xi,\delta}\sum_{\theta,\phi,\psi,\nu,\omega}  \left(I''^{-1}\right)_{\alpha\lambda}
I^{0}_{\lambda\kappa}\left(I''^{-1}\right)_{\kappa\beta}\epsilon_{\gamma\eta\epsilon}\epsilon_{\theta\phi\psi}K^{0}_{\eta\delta}K^{0}_{\phi\omega}
\left(I^{0-1}\right)_{\phi\xi}
I''_{\xi\alpha}\left(I''^{-1}\right)_{\gamma\delta}\left(I^{0-1}\right)_{\psi\nu}I''_{\nu\beta}\left(I''^{-1}\right)_{\theta\omega}=...&&
\end{align}
\end{widetext}
  and making use of the \textit{Kronecker} deltas $\delta_{\xi\lambda},\delta_{\kappa\nu},\delta_{\nu\phi}$ we finally obtain
    \begin{equation}
  \begin{split}
...=-\frac{\hbar^{2}}{8}\sum_{\gamma,\eta,\epsilon,\delta}\sum_{\theta,\phi,\psi,\omega}\epsilon_{\gamma\eta\epsilon}\epsilon_{\theta\phi\psi}K^{0}_{\eta\delta}K^{0}_{\phi\omega}
\left(I''^{-1}\right)_{\gamma\delta}\left(I^{0-1}\right)_{\psi\phi}\left(I''^{-1}\right)_{\theta\omega}
   \end{split}
  \end{equation}
To simplify second term $U_{2}$ lets first compute following inner expression:
\begin{widetext}
\begin{align}\nonumber
& \sum_{\beta}\mu_{\alpha\beta}\left[\hat{j}_{\alpha},I''\right]=\sum_{\beta,\gamma,\delta}\left(I''^{-1}\right)_{\alpha\gamma}I^{0}_{\gamma\delta}\left(I''^{-1}\right)_{\delta\beta}\cdot
\sum_{\eta,\nu,\phi,\psi,\omega}i\hbar I''\epsilon_{\eta\nu\phi}K^{0}_{\nu\psi}\left(I^{0-1}\right)_{\phi\omega}I''_{\omega\beta}\left(I''^{-1}\right)_{\eta\psi}= &&\\
&=i\hbar I''\sum_{\gamma,\delta,\eta,\nu,\psi}\left(I''^{-1}\right)_{\alpha\gamma}\left(I''^{-1}\right)_{\eta\psi}K^{0}_{\nu\psi}\epsilon_{\eta\nu\gamma}
=i\hbar I''\sum_{\delta,\eta,\nu,\psi}\left(I''^{-1}\right)_{\eta\psi}K^{0}_{\nu\psi}\sum_{\gamma}\left(I''^{-1}\right)_{\alpha\gamma}\epsilon_{\eta\nu\gamma}=... &&
\end{align}
\end{widetext}
  using identity:
  \begin{equation}
  \sum_{\gamma}\left(I''^{-1}\right)_{\alpha\gamma}\epsilon_{\eta\nu\gamma}=\sum_{\theta,\mu}\epsilon_{\alpha\theta\mu}\left(I''\right)_{\eta\theta}\left(I''^{-1}\right)_{\gamma\mu}I''^{-1}
\end{equation}
we find:
    \begin{equation}
  \begin{split}
..=i\hbar I''\sum_{\delta,\eta,\nu,\psi,\theta,\mu}\left(I''^{-1}\right)_{\eta\psi}K^{0}_{\nu\psi}\left(I''\right)_{\eta\theta}\left(I''\right)_{\nu\mu}\epsilon_{\alpha\theta\mu}=\\
=i\hbar\sum_{\nu,\psi,\mu}K^{0}_{\nu\psi}\left(I''\right)_{\nu\mu}\epsilon_{\alpha\psi\mu}
   \end{split}
  \end{equation}
  With such preparation we're ready to evaluate $U_{2}$ expression:

\begin{align}\nonumber
& U_{2}=-\frac{1}{4}\sum_{\alpha,\beta}\left(I''^{-1}\right)\left[\hat{j}_{\alpha},\mu_{\alpha\beta}\left[\hat{j}_{\beta},I''\right]\right]=&&\\\nonumber
& =-\frac{1}{4}I''^{-1}\sum_{\alpha,\beta}\left(I''^{-1}\right)\left[\hat{j}_{\alpha},i\hbar\sum_{\nu,\psi,\mu}K^{0}_{\nu\psi}\left(I''\right)_{\nu\mu}\epsilon_{\alpha\psi\mu}\right]=&&\\
& -\frac{i\hbar}{4}I''^{-1}\sum_{\alpha,\beta,\nu,\psi,\mu}\epsilon_{\alpha\psi\mu}K^{0}_{\psi\nu}\left[\hat{j}_{\alpha},I''_{\nu\mu}\right]=...&&
\end{align}
  in the meantime lets expand above commutator
\begin{align}\nonumber
& \left[\hat{j}_{\alpha},I''_{\nu\mu}\right]=i\hbar\left(\sum_{\epsilon,\eta,\xi}\epsilon_{\alpha\nu\epsilon}I''_{\epsilon\mu}-\frac{1}{2}\epsilon_{\alpha\nu\mu}
I''_{\epsilon\epsilon}+\right.&&\\
&+\left.\epsilon_{\nu\eta\epsilon}K^{0}_{\eta\mu}\left(I^{0-1}\right)_{\xi\xi}I''_{\xi\alpha}\right) &&
\end{align}
  hence,
\begin{align}\nonumber
& ...=\frac{\hbar^{2}}{4}I''^{-1}\sum_{\alpha,\beta,\nu,\psi}\sum_{\mu,\phi,\eta,\xi}\left(\epsilon_{\alpha\psi\mu}\epsilon_{\alpha\nu\epsilon}K^{0}_{\psi\nu}I''_{\epsilon\mu}-\right(.&&\\
& \left.-\frac{1}{2}\epsilon_{\alpha\psi\mu}\epsilon_{\alpha\nu\mu}K^{0}_{\psi\nu}I''_{\epsilon\epsilon}+
\epsilon_{\alpha\psi\mu}\epsilon_{\nu\eta\epsilon}K^{0}_{\psi\nu}K^{0}_{\eta\mu}\left(I^{0-1}\right)_{\phi\xi}I''_{\xi\alpha}\right)&&
\end{align}
  after using the identity:
  \begin{equation}
  \epsilon_{\alpha\beta\gamma}\epsilon_{\alpha\eta\phi}=\delta_{\beta\eta}\delta_{\gamma\phi}-\delta_{\beta\phi}\delta_{\gamma\eta}
  \end{equation}
  
\begin{align}\nonumber
& ...=\frac{\hbar^{2}}{4}I''^{-1}\sum_{\epsilon,\mu,\nu,\eta,\psi,\nu,\xi}\left(K^{0}_{\nu\nu}I''_{\epsilon\epsilon}-K^{0}_{\epsilon\mu}I''_{\epsilon\mu}-\frac{1}{2}K^{0}_{\nu\nu}
I''_{\epsilon\epsilon}\right.+&&\\\nonumber
& \left.\frac{1}{2}K_{\mu\mu}^{0}I''_{\epsilon\epsilon}+\epsilon_{\alpha\psi\mu}\epsilon_{\nu\eta\epsilon}
K^{0}_{\psi\nu}K^{0}_{\eta\mu}\left(I^{0-1}\right)_{\phi\xi}I''_{\xi\alpha}\right)=&&\\
& \frac{\hbar^{2}}{4}I''^{-1}\sum_{\epsilon,\mu,\nu,\eta,\psi,\nu}\left(-K^{0}_{\epsilon\mu}+\epsilon_{\alpha\psi\mu}\epsilon_{\nu\eta\epsilon}K^{0}_{\psi\nu}K^{0}_{\eta\mu}\left(I^{0-1}\right)_{\epsilon\mu}\right)I''_{\epsilon\mu} &&
\end{align}
  In order to calculate the third term we will need a general lemma regarding differentiation of determinants:
  \begin{theorem}
  For a given reversible matrix $\textbf{A}$ containing elements remaining functions of $x$ the following equality holds:
  \begin{equation}
  \frac{\partial\det \textbf{A}}{\partial x}=\det \left(\textbf{A}\right) Tr\left(\textbf{A}^{-1}\frac{\partial \textbf{A}}{\partial x}\right)
  \end{equation}
  \end{theorem}
  Recalling:
       \begin{equation}
  \begin{split}
U_{3}=-\frac{\hbar^{2}}{8}I''^{-2}\sum_{k}\left(\frac{\partial I''}{\partial Q_{k}}\right)^{2}
   \end{split}
  \end{equation}
  we may use the lemma to find that:
         \begin{equation}
  \begin{split}
\frac{\partial I''}{\partial Q_{k}}=I''\sum_{\alpha\beta}\left(I''^{-1}\right)_{\alpha\beta}\frac{\partial I''_{\alpha\beta}}{\partial Q_{k}}=\frac{1}{2}\sum_{\alpha\beta}\left(I''^{-1}\right)_{\alpha\beta}a^{\alpha\beta}_{k}
   \end{split}
  \end{equation}
  due to the introduced definition of interaction coefficient \ref{interaction_coeff}, and the fact that $I''$ is linear in normal coordinates, thus $\frac{\partial I''_{\alpha\beta}}{\partial Q_{k}}=\left(\frac{\partial I''_{\alpha\beta}}{\partial Q_{k}}\right)_{0}$. Coming back to $U_{3}$,
         \begin{equation}
  \begin{split}
U_{3}=-\frac{\hbar^{2}}{8}I''^{-2}I''^{2}\sum_{k}\sum_{\alpha,\beta,\gamma,\delta}\left(I''^{-1}\right)_{\alpha\beta}\left(I''^{-1}\right)_{\gamma\delta}
a^{\alpha\beta}_{k}a^{\gamma\delta}_{k}=\\
-\frac{\hbar^{2}}{8}\sum_{k}\sum_{\alpha,\beta,\gamma,\delta}\sum_{\epsilon,\eta,\xi,\theta}\left(I''^{-1}\right)_{\alpha\beta}\left(I''^{-1}\right)_{\gamma\delta}
\left(\delta_{\alpha\beta}\delta_{\gamma\delta}K^{0}_{\epsilon\epsilon}\right.-\\
\left.-\delta_{\alpha\beta}K^{0}_{\gamma\delta}-\delta_{\gamma\delta}K^{0}_{\alpha\beta}+\delta_{\alpha\gamma}K^{0}_{\beta\delta}-\epsilon_{\alpha\epsilon\eta}
\epsilon_{\gamma\xi\theta}K^{0}_{\beta\epsilon}K^{0}_{\delta\xi}\left(I^{0-1}\right)_{\eta\theta}\right)
   \end{split}
  \end{equation}
  where we've made use of sum rule for interaction coefficients \ref{sumrule2}. Mostly in similar way we can evaluate the last term $U_{4}$
    \begin{equation}
  \begin{split}
U_{4}=-\frac{\hbar^{2}}{4}I''^{-1}\sum_{k}\frac{\partial^{2}I''}{\partial Q_{k}^{2}}
   \end{split}
  \end{equation}
  remembering determinant definition $I''=\frac{1}{6}\sum_{\alpha,\beta,\gamma,\delta,\epsilon,\xi}\epsilon_{\alpha\beta\gamma}\epsilon_{\delta\epsilon\xi}I''_{\alpha\delta}I''_{\beta\epsilon}I''_{\gamma\xi}$
     \begin{equation}
  \begin{split}
\frac{\partial^{2}I''}{\partial Q_{k}^{2}}=\frac{1}{4}\sum_{\alpha,\beta,\gamma,\delta,\lambda,\eta}\epsilon_{\eta\alpha\gamma}\epsilon_{\lambda\beta\delta}I''_{\eta\lambda}a^{\alpha\beta}_{k}a^{\gamma\delta}
   \end{split}
  \end{equation}
  inserting into original expression, and after short manipulations we get
      \begin{equation}
  \begin{split}
  U_{4}=\frac{\hbar^{2}}{4}I''^{-1}\sum_{\alpha,\beta,\theta,\xi,\epsilon,\lambda}I''_{\alpha\lambda}\left(2K^{0}_{\alpha\lambda}-\epsilon_{\lambda\beta\delta}
\epsilon_{\theta\epsilon\xi}K^{0}_{\beta\epsilon}K^{0}_{\delta\xi}\left(I^{0-1}\right)_{\alpha\theta}\right)
   \end{split}
  \end{equation}
  Upon adding $U_{1}$ and $U_{3}$ we notice that last term in $U_{3}$ cancels $U_{1}$, hence
\begin{align}\nonumber  
    \small
U_{1}+U_{3}=&
-\frac{\hbar^{2}}{8}\sum_{k}\sum_{\alpha,\beta,\gamma,\delta}\sum_{\epsilon,\eta,\xi,\theta}\left(I''^{-1}\right)_{\alpha\beta}\left(I''^{-1}\right)_{\gamma\delta}\times\\
&\times\left(\delta_{\alpha\beta}\delta_{\gamma\delta}K^{0}_{\epsilon\epsilon}-\delta_{\alpha\beta}K^{0}_{\gamma\delta}-\delta_{\gamma\delta}K^{0}_{\alpha\beta}+\delta_{\alpha\gamma}K^{0}_{\beta\delta}\right)
\end{align}
\normalsize
  and remaining terms sum up to:
            \begin{equation}
  \begin{split}
U_{2}+U_{4}=-\frac{\hbar^{2}}{4}I''^{-1}\sum_{\alpha,\lambda}\left(K^{0}_{\alpha\lambda}I''_{\alpha\lambda}-2K^{0}_{\alpha\lambda}I''_{\alpha\lambda}\right)=\\
\frac{\hbar^{2}}{4}I''^{-1}\sum_{\alpha,\lambda}K^{0}_{\alpha\lambda}I''_{\alpha\lambda}
   \end{split}
  \end{equation}
  making use of relation
  \begin{equation}
    \begin{split}
I''_{\alpha\delta}=\frac{1}{2}I''\sum_{\gamma,\epsilon,\xi,\beta}\epsilon_{\alpha\beta\gamma}\epsilon_{\delta\epsilon\xi}I''^{-1}_{\epsilon\beta}I''^{-1}_{\xi\gamma}
   \end{split}
  \end{equation}
  we find
              \begin{equation}
  \begin{split}
U_{2}+U_{4}=-\frac{\hbar^{2}}{8}\sum_{\alpha,\lambda,\gamma,\epsilon,\xi,\beta}K^{0}_{\alpha\lambda}\epsilon_{\alpha\beta\gamma}\epsilon_{\lambda\epsilon\xi}I''^{-1}_{\epsilon\beta}I''^{-1}_{\xi\gamma}
   \end{split}
  \end{equation}
\begin{align}\nonumber  
    \small
U=-&\frac{\hbar^{2}}{8}\sum_{\alpha,\beta,\gamma,\delta,\epsilon}\left(\left(I''^{-1}\right)_{\alpha\beta}
\left(I''^{-1}\right)_{\gamma\delta}\left(\delta_{\alpha\beta}\delta_{\gamma\delta}K^{0}_{\epsilon\epsilon}-\delta_{\alpha\beta}K^{0}_{\gamma\delta}
-\right.\right.\\
& \left.\left.-\delta_{\gamma\delta}K^{0}_{\alpha\beta}+\delta_{\alpha\gamma}K^{0}_{\beta\delta}\right)-\sum_{\lambda,\xi}K^{0}_{\alpha\lambda}\epsilon_{\alpha\beta\gamma}\epsilon_{\lambda\epsilon\xi}I''^{-1}_{\epsilon\beta}I''^{-1}_{\xi\gamma}\right)
\end{align}
\normalsize
  changing indices in last sum above $\epsilon \rightarrow \alpha,\quad \xi \rightarrow \delta$ and utilizing fact that $I''$ is symmetric matrix we have:
    \begin{align}\nonumber  
    \small
U=-&\frac{\hbar^{2}}{8}\sum_{\alpha,\beta,\gamma,\delta,\epsilon,\lambda}\left(I''^{-1}\right)_{\alpha\beta}
\left(I''^{-1}\right)_{\gamma\delta}\left(\delta_{\alpha\beta}\delta_{\gamma\delta}K^{0}_{\epsilon\epsilon}-\delta_{\alpha\beta}K^{0}_{\gamma\delta}
-\right.\\
&-\left.\delta_{\gamma\delta}K^{0}_{\alpha\beta}+\delta_{\alpha\gamma}K^{0}_{\beta\delta}K^{0}_{\epsilon\lambda}\epsilon_{\epsilon\beta\delta}\epsilon_{\lambda\alpha\gamma}I''^{-1}_{\epsilon\beta}I''^{-1}_{\xi\gamma}\right)
\end{align}
\normalsize
  Finally it is required to use conversion formula between \textit{Levi-Civita }tensor product and \textit{Kronecker} Deltas:
  \begin{equation}
  \epsilon_{\epsilon\beta\delta}\epsilon_{\lambda\alpha\gamma}=\left|\begin{array}{ccc}
  \delta_{\epsilon\lambda} & \delta_{\epsilon\alpha} & \delta_{\epsilon\gamma}\\
  \delta_{\beta\lambda} & \delta_{\beta\alpha} & \delta_{\beta\gamma}\\
  \delta_{\delta\lambda} & \delta_{\delta\alpha}& \delta_{\delta\gamma}
  \end{array}\right|
  \end{equation}
  by cancelling emerging identical terms we continue to
    \begin{align}\nonumber  
    \small
U=-&\frac{\hbar^{2}}{8}\sum_{\alpha,\beta,\gamma,\delta,\epsilon,\lambda}\left(I''^{-1}\right)_{\alpha\beta}
\left(I''^{-1}\right)_{\gamma\delta}\left(\delta_{\alpha\delta}\delta_{\beta\gamma}K^{0}_{\epsilon\epsilon}+\right.\\
&\left.+\delta_{\alpha\gamma}K^{0}_{\beta\delta}
-\delta_{\beta\gamma}K^{0}_{\alpha\delta}-\delta_{\alpha\delta}K^{0}_{\gamma\beta}\right)
\end{align}
\normalsize
 rearranging some terms
    \begin{align}\nonumber  
    \small
U=-&\frac{\hbar^{2}}{8}\sum_{\alpha,\beta,\gamma,\delta,\epsilon,\lambda}\left(I''^{-1}\right)_{\alpha\beta}
\left(I''^{-1}\right)_{\delta\gamma}\left(\delta_{\alpha\delta}\left(\delta_{\beta\gamma}K^{0}_{\epsilon\epsilon}-K^{0}_{\gamma\beta}\right)\right.\\
&\left.+K^{0}_{\beta\delta}\delta_{\alpha\gamma}
-K^{0}_{\alpha\delta}\delta_{\beta\gamma}\right)
\end{align}
\normalsize
  Now lets recall the definition of equilibrium moment of inertia $I^{0}_{\alpha\beta}=\delta_{\alpha\beta}K_{\gamma\gamma}^{0}-K_{\alpha\beta}^{0}$
As both $K^{0}_{\alpha\beta}$ and $I''$ are symmetric tensors and we perform summation over all indices, hence last two term in brackets give antisymmetric contribution while $I''$ give symmetric contribution to the product, yielding vanishing expression.
        \begin{equation}
  \begin{split}
U=-\frac{\hbar^{2}}{8}\sum_{\beta,\gamma,\delta}\left(I''^{-1}\right)_{\delta\beta}I^{0}_{\beta\gamma}
\left(I''^{-1}\right)_{\gamma\delta}=-\frac{\hbar^{2}}{8}Tr\mu
   \end{split}
  \end{equation}
  which proves the result obtained by Watson.
\textit{ Eckart-Watson} Hamiltonian approach have been successfully applied to potassium cyanide rovibrational problem by Tennyson and	Sutcliffe \cite{Ten1982KCNMolPhys}. However, there are numerous examples in the literature where other forms of the molecular Hamiltonian provide better convergence of energies, by proper treatment of large amplitude motions.

\section{\label{sec:rotor}Rigid rotor approximation}
In the last section of this paper we present solution and a short discussion of the model of rigid rotor in 3D. Resulting eigenstates are ubiquitous in nuclear motion theory, thereby this model should be deemed as important.
In section \ref{sec:level2} we ended with an approximate form for rotational Hamiltonian:
\begin{equation}\label{rotorhamil}
\hat{H}_{rot}=\frac{1}{2}\sum_{\alpha=1}^{3}\mu_{\alpha\alpha}^{e}\hat{J}_{\alpha}^{2}
\end{equation}
Hence the corresponding rigid rotor Schr\"{o}dinger equation reads:
\begin{equation}
\frac{1}{2}\sum_{\alpha=1}^{3}\mu_{\alpha\alpha}^{e}\hat{J}_{\alpha}^{2}\Phi_{rot}(\theta, \phi, \chi)=E_{rot}\Phi_{rot}(\theta, \phi, \chi)
\end{equation}
For spectroscopic purposes it will be convenient to rewrite the above equation into wavenumber units, because energy scale of rotational transitions is of order of tens/hundreds wavenumbers. Introducing rotational constants, e.g. $A_{e}=\frac{\hbar^{2}\mu_{aa}^{e}}{2hc}$, with convention that $A_{e}\ge B_{e} \ge C_{e}$ we find a new form of SE:
 \begin{equation}
\resizebox{\hsize}{!}{$\hbar^{-2}\left(A_{e}\hat{J}_{a}^{2}+B_{e}\hat{J}_{b}^{2}+C_{e}\hat{J}_{c}^{2}\right)\Phi_{rot}(\theta, \phi, \chi)=E_{rot}\Phi_{rot}(\theta, \phi, \chi)$}
\end{equation}
This is called the asymmetric top Schr\"{o}dinger equation, since all rotational constants take different values. The simplest molecular asymmetric top model may be applied to $H_{2}D^{+}$ cation, or simply water. There are three general cases of top systems:
\begin{itemize}
\item The symmetric top molecule, in which two rotational constants are equal, dividing into two groups:
\subitem Prolate symmetric top when $A_{e}>B_{e}=C_{e}$ (e.g. $CH_{3}Cl, trans-Ni(H_{2}O)_{4}Cl_{2}, CO_2)$(linear subtype),$propyne$ (cigar or rugby shaped molecules)
\subitem Oblate symmetric top when $A_{e}=B_{e}>C_{e}$ (e.g. $SO_{3}, BF_{3}, NO_{3}^{-}, C_{6}H_{6}, NH_{3})$ (disk or doughnut shaped molecules)
\item The spherical top molecule when $A_{e}=B_{e}=C_{e}$ (e.g. $SF_{6}, CH_{4}, SiH_{4})$
\item The asymmetric top molecule when $A_{e}\ne B_{e} \ne C_{e} \ne A_{e}$ (e.g. $H_{2}O, NO_{2}$)
\end{itemize}
It should be noted that symmetric top molecules must contain three-fold or higher symmetry axis, as an element of the point group. Molecules possessing at most two-fold axis are asymmetric tops (most of larger molecules are).
\subsection{Symmetric top molecule}
Common convention for symmetric tops is to choose space-fixed $z$ axis as the one with the distinct rotational constant $A_{e}$.
Prolate symmetric top yields in following SE:
 \begin{equation}\label{rigidSE}
\resizebox{.98\hsize}{!}{$\hbar^{-2}\left(A_{e}\hat{J}_{z}^{2}+B_{e}(\hat{J}_{b}^{2}+\hat{J}_{c}^{2})\right)\Phi_{rot}(\theta, \phi, \chi)=E_{rot}\Phi_{rot}(\theta, \phi, \chi)$}
\end{equation}
Angular momentum operators defined in the previous section, together with the use of chain rule results in:

\begin{align}\nonumber
 & \left(\frac{1}{\sin\theta}\frac{\partial}{\partial \theta}\left(\sin\theta\frac{\partial}{\partial\theta}\right)+\frac{1}{\sin^{2}\theta}\frac{\partial^{2}}{\partial \phi^{2}}\right.+ &&\\
 &\resizebox{\hsize}{!}{$ \left.+\left(\cot^{2}\theta+\frac{A_{e}}{B_{e}}\right)\frac{\partial^{2}}{\partial \chi^{2}}-2\frac{\cos\theta}{\sin^{2}\theta}\frac{\partial^{2}}{\partial \phi \partial \chi}+\frac{E_{rot}}{B_{e}}\right)\Phi_{rot}(\theta, \phi, \chi)=0 $}&&
\end{align}
 
As Euler angles $\chi, \phi$ occur only in derivatives, they are cyclic coordinates hence, we can make an anzatz:
 \begin{equation}
\Phi_{rot}(\theta, \phi, \chi)=\Theta(\theta)e^{im\phi}e^{ik\chi}
\end{equation}
Cyclic boundary conditions $\Phi_{rot}(\theta, \phi, \chi+2\pi)=\Phi_{rot}(\theta, \phi+2\pi, \chi)=\Phi_{rot}(\theta, \phi, \chi)$  imply $m,k \in Z$. After some manipulations we obtain the following equation for $\Theta(\theta)$:
 \begin{equation}
\resizebox{.98\hsize}{!}{$\left(\frac{1}{\sin\theta}\frac{d}{d \theta}\left(\sin\theta\frac{d}{d\theta}\right)+\left[\Delta-\frac{m^{2}-2mk\cos\theta+k^{2}}{\sin^{2}\theta}\right]\right)\Theta(\theta)=0$}
\label{theta}
\end{equation}
where
\begin{equation}\label{Rigid:delta}
  \Delta=\frac{\left(E_{rot}-(A_{E}-B_{e})k^{2}\right)}{B_{e}}
\end{equation}

This is already one dimensional second-order linear differential equation, which can be solved analytically by writing
  \begin{equation}
\Theta(\theta)=x^{\frac{|k-m|}{2}}(1-x)^{\frac{|k+m|}{2}}F(x)
\end{equation}
where $x=\frac{(1-\cos\theta)}{2}$. Note that also
\begin{equation}
\Theta(\theta)=\left(\sin\frac{\theta}{2}\right)^{|k-m|}\left(\cos\frac{\theta}{2}\right)^{|k+m|}F(\sin^{2}\frac{\theta}{2})
\end{equation}
After some lengthy differentiation and algebra eq. \ref{theta} reduces to
 \begin{equation}
x(1-x)\frac{d^{2}F}{dx^{2}}+(\alpha-\beta x)\frac{dF}{dx}+\gamma F=0
\label{Hypergeometric}
\end{equation}
with
 \begin{equation}\label{Rigid:alpha}
\alpha=1+|k-m|
\end{equation}
 \begin{equation}\label{Rigid:beta}
\beta=\alpha+1+|k+m|
\end{equation}
and
\begin{equation}\label{Rigid:gamma}
\gamma=\Delta-\frac{\beta(\beta-2)}{4}
\end{equation}
Following Frobenius method we write expansion form of our solution:
\begin{equation}\label{Frobenius}
F(x)=\sum_{n=0}^{\infty}a_{n}x^{n}
\end{equation}
First and second derivative read:
\begin{align}\nonumber
 & \frac{dF}{dx}=\sum_{n=1}^{\infty}a_{n}nx^{n-1} &&\\
 & \frac{d^{2}F}{dx^{2}}=\sum_{n=2}^{\infty}a_{n}n(n-1)x^{n-2} &&
\end{align}
Substituting into \ref{Hypergeometric}:
\begin{align}\nonumber
 & \gamma a_{0}+\alpha a_{1} +\left[\left(\gamma -\beta\right)a_{1}+2\left(1+\alpha)\right)a_{2}\right]x &&\\
 & +\sum_{n=2}^{\infty}\left[\left(\gamma - n(n-1)-\beta n\right)a_{n}+(n+1)(n+\alpha)a_{n+1}\right]x^{n}=0
\end{align}
Due to linear independence of basis monomials we can equate consequent coefficients to $0$
\begin{align}\nonumber
 & a_{1}=\frac{\gamma}{2}a_{0}&&\\\nonumber
 & a_{2}=a_{1}\frac{\beta-\gamma}{2+2\alpha} &&\\
 & a_{n+1}=\frac{-\gamma +\beta n +n(n-1)}{(n+1)(n+\alpha)}a_{n} &&
  \label{coeffs}
\end{align}
  The $a_{0}$ coefficient is chosen so that the rotational wavefunction is normalised.
  In order to $\Theta(\theta)$ be a proper representation of wavevector it must be finite, thus the series expansion \ref{Frobenius} must truncate at finite term labeled by $n_{max}$, then

  \begin{equation}\label{max}
    a_{n_{max}+1}=0
  \end{equation}
  giving the condition, which provides eigenvalues of rigid rotor SE (cf.\ref{rigidSE})

  \begin{equation}\label{Rigid:condition}
    \beta n_{max} +n_{max}(n_{max}-1)-\gamma=0
  \end{equation}
  Substituting eqs. \ref{Rigid:delta},\ref{Rigid:beta},\ref{Rigid:gamma} and making the abbreviation $J=n_{max}+\frac{|k+m|+|k-m|}{2}$ gives the eigenenergies as functions of $J$ and $k$:

  \begin{equation}\label{Rigid:eigenenergy}
    E_{rot}=B_{e}J(J+1)+(A_{e}-B_{e})k^{2}
  \end{equation}
  We can easily figure out that:
    \begin{align}\nonumber
   & J=0,1,2,... \; & k=0,\pm 1,\pm 2,...,\pm J \; and &\\
   & & m=0,\pm 1,\pm 2,...,\pm J &
  \label{Rigid:numbers}
    \end{align} 
  It turns out that $F(x)$ solution is hypergeometric function \cite{Andrews}, which enables us writing rotational wavefunction in compact form:

  \begin{equation}\label{Rigid:Solution}
    \resizebox{\hsize}{!}{$\Phi_{Jkm}(\theta,\phi,\chi)=N_{Jkm}x^{\frac{|k-m|}{2}}(1-x)^{\frac{|k+m|}{2}}F(\frac{1}{2}\beta-J-1,\frac{1}{2}\beta+J;\alpha,x)e^{im\phi}e^{ik\chi}$}
  \end{equation}
  The normalization constant $N_{Jkm}$ is defined by the condition:

  \begin{equation}\label{Rigid:normalization}
    \int_{0}^{2\pi}\int_{0}^{2\pi}\int_{0}^{\pi}\Phi_{Jkm}^{*}\Phi_{Jkm}\sin\theta d\theta d\phi d\chi =1
  \end{equation}
  and a choice of phase factor. In fact symmetric top rotational wavefunction have many different representations. One most popular is given below:
   \begin{equation}\label{Rigid:Solution2}
    \Phi_{Jkm}(\theta,\phi,\chi)=\sqrt{\frac{2J+1}{8\pi^{2}}}D^{(J)}_{km}(\theta,\phi,\chi)
  \end{equation}
  where $D^{(J)}_{km}(\theta,\phi,\chi)$ is \textit{Wigner D-matrix}\cite{Bunker}. Sometimes abbreviation of type $ \Phi_{Jkm}(\theta,\phi,\chi)\equiv |J,k,m\rangle$ is made, but it should be treated with caution, as \textit{ket} vector represent \textit{state vector} from \textit{Hilbert space} and  $\Phi_{Jkm}$ states for spectral representation of this state vector (hence is representation-dependent) and belongs to $L^{2}(C^{3},d^{3})$ space of square-integrable functions, which in fact is isomorphic with the original \textit{Hilbert space}. The explicit form of the wavefunction $\Phi_{Jkm}(\theta,\phi,\chi)$ is
  \begin{equation}\label{Rigid:Solution3}
  \resizebox{\hsize}{!}{$ N\left[\sum_{\sigma}(-1)^{\sigma}\frac{(\cos\frac{1}{2}\theta)^{2J+k-m-2\sigma}(-\sin\frac{1}{2}\theta)^{m-k+2\sigma}}{\sigma!(J-m-\sigma)!(m-k+\sigma)!(J+k-\sigma)!}\right]
   e^{im\phi}e^{ik\chi}$}
  \end{equation}
  where,
\begin{equation}\label{Rigid:normalization2}
 N=\left[\frac{(J+m)!(J-m)!(J+k)!(J-k)!(2J+1)}{8\pi^{2}}\right]^{\frac{1}{2}}
  \end{equation}
  and $\sigma$ runs from 0 or $(k-m)$, whichever is the larger, up to $(J-m)$ or $(J+k)$, whichever is smaller \cite{Bunker}. Note that the symmetric top Hamiltonian commutes with $\hat{J}^{2}, \hat{J}_{\rho_{3}},\hat{J}_{z}$, therefore they have common eigenfunctions. Apparently different situation emerges in case of asymmetric top (all three rotational constants are different). Here the Hamiltonian \ref{rotorhamil} commutes with $\hat{J}^{2}, \hat{J}_{\rho_{3}}$ but not with $\hat{J}_{z}$:
  \begin{equation}
\left[\hat{H},\hat{J}_{z}\right]=\left[\hat{J}_{x},\hat{J}_{y}\right]_{+}\left(B_{e}-A_{e}\right)
  \end{equation}
  where molecule fixed angular momentum operators satisfy commutation relations:
  \begin{equation}\label{anomalcommutation}
  \left[\hat{J}_{i},\hat{J}_{j}\right]=-i\hbar \epsilon_{ijk} \hat{J}_{k}
  \end{equation}
  Therefore $\hat{H}$ will have common eigenbasis with $\hat{J}^{2}$ and $\hat{J}_{\rho_{3}}$ but not with $\hat{J}_{z}$. In order to obtain eigenfunctions, a  good starting point would be to express Hamilton matrix in the symmetric top eigenbasis. Such approach provides final states representation as linear combination of \textit{Wigner} functions with different $k$ values:

\[  \left( \begin{array}{ccc}
\Psi_{1} \\
\vdots \\
\Psi_{2J+1} \end{array} \right)=\textbf{C}_{J,m} \left( \begin{array}{ccc}
|J,-J,m\rangle \\
\vdots \\
|J,J,m\rangle \end{array} \right).\] 
where $\textbf{C}$ diagonalizes $\hat{H}$ block for a given $(J,m)$. At this stage it is convenient to write explicitly Hamilton matrix elements in symmetric top basis. To do this, lets rewrite Hamiltonian as below
\begin{align}\nonumber
\hat{H}= &\hbar^{-2}\left[\frac{1}{2}\left(B_{e}+C_{e}\right)\hat{J}^{2}+\left[A_{e}-\frac{1}{2}\left(B_{e}+C_{e}\right)\right]\hat{J}_{z}^{2}\right.+&&\\
&\left.+\frac{1}{4}\left(B_{e}-
C_{e}\right)\left((\hat{J}_{m}^{+})^{2} + (\hat{J}_{m}^{-})^{2}\right)\right] &&
\label{asymetricham}
\end{align}
\normalsize
where we introduced molecule fixed ladder operators: $\hat{J}_{m}^{\pm}:=\hat{J}_{x}\pm i\hat{J}_{y}$. The purpose of abbreviating this new operators lays within their marvellous abilities to ladder and lower $k$ quantum number in our basis. And this feature is due to special commutation relations fulfilled by ladder operators:
 \begin{equation}\label{laddercommutation}
\left[\hat{J}_{z},\hat{J}_{m}^{\pm}\right]=\mp\hbar\hat{J}_{m}^{\pm}
\end{equation}
Because $|J,k,m\rangle$ are eigenfunctions of $\hat{J}_{z}$: $\hat{J}_{z}|J,k,m\rangle=\hbar k|J,k,m\rangle$ the laddered function $\hat{J}_{m}^{\pm}|J,k,m\rangle$ is also eigenfunction of $\hat{J}_{z}$ to the $\hbar(k \mp 1)$ eigenvalue. 
Firstly lets act on standard basis element with a product of operators: $\hat{J}_{z}\hat{J}_{m}^{\pm}|J,k,m\rangle$. This is equal, on account of \ref{laddercommutation} to $\hat{J}_{m}^{\pm}\hat{J}_{z}|J,k,m\rangle+\hat{J}_{m}^{\pm}|J,k,m\rangle$. Now utilizing eigenequation for $\hat{J}_{z}$ we find
\begin{equation}
\hat{J}_{z}\hat{J}_{m}^{\pm}|J,k,m\rangle=\hbar (k\mp1)\hat{J}_{m}^{\pm}|J,k,m\rangle
\end{equation}
laddered function to be also eigenfunction of $\hat{J}_{z}$ but to eigenvalue shifted by $1$. That's the general property of operators obeying \ref{laddercommutation} commutation relations. The inverse in sign here is due to anomal molecule fixed commutation relations \ref{anomalcommutation}.
Therefore we presume that laddered function should be proportional to its 'neighbour' with shifted $k$ value:
\begin{equation}
\hat{J}_{m}^{\pm}|J,k,m\rangle=N_{J,k}|J,k\mp 1,m\rangle
\end{equation}
Requiring orthonormality for both hand sides states it's straightforward to write down final relation:
\begin{equation}
\hat{J}_{m}^{\pm}|J,k,m\rangle=\hbar \sqrt{J(J+1)-k(k\mp 1)}|J,k\mp 1,m\rangle
\end{equation}
where we made use of equality: $\hat{J}_{m}^{\mp}\hat{J}_{m}^{\pm}=\hat{J}^{2}-\hat{J}_{z}(\hat{J}_{z}\mp \hbar)$.
In order to obtain general symmetric top wavefunction from some known and easy to obtain generating function, it's needed to find a way to ladder $m$ values (see ref. \cite{Zak_HO}), i.e. space fixed z-axis projection of angular momentum quantum number. This goal may be acheived by introducing space fixed angular momentum ladder operators in similar way to molecule fixed case. Now angular momentum operators satisfy normal commutation relations:
  \begin{equation}\label{anomalcommutation}
  \left[\hat{J}_{\rho i},\hat{J}_{\rho j}\right]=i\hbar \epsilon_{ijk} \hat{J}_{\rho k}
  \end{equation}
 and analogical definition of ladder operators $\hat{J}_{s}^{\pm}:=\hat{J}_{\rho1}\pm i\hat{J}_{\rho2}$ may be applied to retrieve result similar to the previous case:
  \begin{equation}
\hat{J}_{\rho3}\hat{J}_{s}^{\pm}|J,k,m\rangle=\hbar( m\pm1)\hat{J}_{s}^{\pm}|J,k,m\rangle
\end{equation}
  and
  \begin{equation}
\hat{J}_{s}^{\pm}|J,k,m\rangle=\hbar \sqrt{J(J+1)-m(m\pm 1)}|J,k,m\pm 1\rangle
\end{equation}
allowing to derive the following important relation:
 \begin{equation}
|J,\pm|k|,\pm|m|>=N\left(\hat{J}_{m}^{\mp}\right)^{k}\left(\hat{J}_{s}^{\pm}\right)^{m}|J,0,0\rangle
\end{equation} 
Normalisation factor is given by
   \begin{equation}
N=\hbar^{-(|k|+|m|)}\sqrt{\frac{(J-|m|)!(J-|k|)!}{(J+|m|)!(J+|k|)!}}
\end{equation}
Note the analogy to the second quantization procedures. Here we obtained a general state as a result of acting with creation operators on the 'vacuum' state of rotations space. These operators were derived from position representation differential operators, giving rare link between traditional quantum mechanics and quantum field theory formalism.
Knowing the above symmetric top wavefunctions, it's easy to write down matrix elements of a particular angular momentum associated operators:

\begin{align}\nonumber
&<J,k,m|\hat{J}^{2}|J,k,m>=\hbar^{2}J(J+1)\\\nonumber
&<J,k,m|\hat{J}_{z}|J,k,m>=k\hbar\\\nonumber
&<J,k,m|\hat{J}_{\rho3}|J,k,m>=m\hbar\\\nonumber
&<J,k,m\pm1|\hat{J}^{\pm}_{s}|J,k,m>=\hbar\sqrt{J(J+1)-m(m\pm1)}\\
&<J,k\mp1,m|\hat{J}^{\pm}_{m}|J,k,m>=\hbar\sqrt{J(J+1)-k(k\mp1)}
\label{elements}
\end{align}
At this stage we're ready to build the Hamilton matrix for the asymmetric top, which a $J(2J+1)^{2}$ dimensional symmetric matrix. Nevertheless on account of relations from eq.\ref{elements} all elements between states with different $J$ vanish, making the matrix block-diagonal. Because eigenenergies of $\hat{H}$ don't depend on $m$ (see \ref{asymetricham}) each block with fixed $m$ will yield with identical eigenvalues, hence we can restrict to only one of them, e.g. $m=0$. Every $m$ replica is $2J+1$ dimensional with states labelled by $k=-J,-J+1,...,J-1,J$. For a given $J$ there are $2J+1$ identical blocks because $m=-J,-J+1,...,J-1,J$. The whole matrix may be theoretically analytically diagonalized by diagonalising subsequent blocks with rising $J$. And so for $J=0$, $m=0$ and $k=0$, hence it's already diagonalized with eigenvalue equal to $0$, which means that for asymmetric top we can find states with no rotational motion at all! The anticipated wavefunction is a constant $|0,0,0\rangle=\sqrt{\frac{1}{8\pi^{2}}}$. For $J=1$ we have three dimensional block with basis functions $|1,-1,0\rangle, |1,0,0\rangle, |1,1,0\rangle$. In order to diagonalize this block we can take advantage of one of the symmetries of Hamiltonian, namely the parity. Due to Hamiltonian invariance under parity operation the states must transform according to irreducible representation of parity group, which in fact is $C_{i}$ with two representations: $A$ and $B$. For rotations this means that the proper wavefunction must take into account clockwise and anti-clockwise rotations of a systems. Equivalently states with $k$ and $-k$ must be allowed. Symmetry adapted states are therefore in-phase and out-of-phase linear combinations of the mentioned two:
   \begin{equation}
|1,1,0,\pm\rangle=\frac{1}{\sqrt{2}}\left(|1,1,0\rangle\pm|1,-1,0\rangle\right)
\end{equation}
 By lucky coincidence this unitary transformation already diagonalizes the $J=1$ block. The corresponding eigenvalues read
\begin{align}\nonumber
&\langle1,1,0,+|\hat{H}|1,1,0,+\rangle=A_{e}+B_{e}&\\
&\langle1,1,0,-|\hat{H}|1,1,0,-\rangle=A_{e}+C_{e} \\
&\langle1,0,0|\hat{H}|1,0,0\rangle=B_{e}+C_{e}
\end{align}
Results have proper dimension of $cm^{-1}$. Note that in present case the $|k|$ no longer labels the states (the states are not eigenstates of $J_{z}$), therefore cannot constitute a \textit{good quantum number}. Instead new index called \textit{parity} has been introduced. Heading now towards $J=2$ block which is $5\times 5$ lets build similar symmetry adapted states:
\begin{align}\nonumber
&|2,2,0,E^{\pm}\rangle=\frac{1}{\sqrt{2}}\left(|2,2,0\rangle\pm|2,-2,0\rangle\right) &\\
&|2,0,0,E^{+}\rangle=|2,0,0\rangle\\
&|2,1,0,O^{\pm}\rangle=\frac{1}{\sqrt{2}}\left(|2,1,0\rangle\pm|2,-1,0\rangle\right)
\end{align}
The block factorizes into subblocks $E^{+},E^{-},O^{+},O^{-}$. For $J$ even (as in our case) the $E^{+}$ block has dimension $\frac{J+2}{2}$ while the other three $\frac{J}{2}$, and for $J$ odd the $E^{-}$ block has dimension $\frac{J-1}{2}$ while the other three $\frac{J+1}{2}$. Hence in our case $E^{+}$ block is two dimensional and the remaining three are one dimensional, providing eigenvalues:
    \begin{equation}
    \begin{split}
\langle2,2,0,E^{-}|\hat{H}|2,2,0,E^{-}\rangle=4A_{e}+B_{e}+C_{e}\\
\langle2,1,0,O^{+}|\hat{H}|2,1,0,O^{+}\rangle=A_{e}+4B_{e}+C_{e}\\
\langle2,1,0,O^{-}|\hat{H}|2,1,0,O^{-}\rangle=A_{e}+B_{e}+4C_{e}\\
\end{split}
\end{equation}
The $2\times 2$ block
\[  \left( \begin{array}{ccc}
3(B_{e}+C_{e}) & \sqrt{3}(B_{e}-C_{e}) \\
\sqrt{3}(B_{e}-C_{e}) & 4A_{e}+B_{e}+C_{e}
\end{array} \right)\] 
can be easily diagonalized to give appropriate eigenvectors and eigenvalues, what we leave for the reader as an exercise. The important fact is that basis states with even/odd $k$ couple only to states with even/odd $k$ and states with $\pm$ parity couple to states with $\pm$ parity. 

\section*{Acknowledgements}
\color{white}I would like to thank the homeless stranger from King Cross station who kindly spoke to me in following words: \textit{Before you reach for all those superb whiskies and liqueurs for the bold and the beautiful, try to deeply understand and appreciate the taste of simple beer, unless you want to end up as a useless drunker like me}.
\color{black}
\newpage
\bibliography{normalne.bib}
\end{document}